\pdfoutput=1

\documentclass{ut-thesis}

\usepackage[usenames]{color}
\usepackage{epsfig}
\usepackage{amssymb}
\usepackage{amsmath}
\usepackage{amsfonts}

\usepackage[noend]{algorithmic}

\def\qed{$\Box$}

\newcommand {\dx}{DescribeX}
\newcommand {\xp}{XPlainer}
\newcommand {\dxp}{DescribeX-Eclipse}
\newcommand {\xpe}{XPlainer-Eclipse}

\newcommand {\out}[1]{}

\newtheorem{algorithm}{Algorithm}[chapter]

\newtheorem{proposition}{Proposition}[chapter]
\newtheorem{lemma}{Lemma}[chapter]
\newtheorem{corollary}{Corollary}[chapter]
\newtheorem{definition}{Definition}[chapter]
\newtheorem{example}{Example}[chapter]

\newtheorem{proof}{Proof}[chapter]

\degree{Doctor of Philosophy} \department{Computer Science}
\gradyear{2008}
\author{Flavio Rizzolo}
\title{\dx: A Framework for Exploring and Querying XML Web Collections}

\begin{document}


\begin{preliminary}

\maketitle


\begin{abstract}

The nature of semistructured data in web collections is evolving.
Even when XML web documents are valid with regard to a schema, the
actual structure of such documents exhibits significant variations
across collections for several reasons: an XML schema may be very
lax (e.g., to accommodate the flexibility needed to represent
collections of documents in RSS\footnote{\small \tt{
http://www.rss-specifications.com/}} feeds), a schema may be large
and different subsets used for different documents (e.g., this is
common in industry standards like UBL\footnote{\small \tt{
http://oasis-open.org/committees/ubl/}}), or open content models
may allow arbitrary schemas to be mixed (e.g., RSS extensions like
those used for podcasting). A schema alone may not provide
sufficient information for many data management tasks that require
knowledge of the actual structure of the collection.

Web applications (such as processing RSS feeds or web service
messages) rely on XPath-based data manipulation tools. Web
developers need to use XPath queries effectively on increasingly
larger web collections containing hundreds of thousands of XML
documents. Even when tasks only need to deal with a single
document at a time, developers benefit from understanding the
behaviour of XPath expressions across multiple documents (e.g.,
what will a query return when run over the thousands of hourly
feeds collected during the last few months?). Dealing with the
(highly variable) structure of such web collections poses
additional challenges.

This thesis introduces \dx, a powerful framework that is capable
of {\em describing} arbitrarily complex XML summaries of web
collections, providing support for more efficient evaluation of
XPath workloads. \dx\ permits the declarative description of
document structure using all axes and language constructs in
XPath, and generalizes many of the XML indexing and summarization
approaches in the literature. \dx\ supports the construction of
heterogenous summaries where different document elements sharing a
common structure can be declaratively defined and refined by means
of path regular expressions on axes, or \emph{axis path regular
expression} (AxPREs). \dx\ can significantly help in the
understanding of both the structure of complex, heterogeneous XML
collections and the behaviour of XPath queries evaluated on them.

Experimental results demonstrate the scalability of \dx\ summary
\emph{refinements} and \emph{stabilizations} (the key enablers for
tailoring summaries) with multi-gigabyte web collections. A
comparative study suggests that using a \dx\ summary created from
a given workload can produce query evaluation times orders of
magnitude better than using existing summaries. \dx's light-weight
approach of combining summaries with a file-at-a-time XPath
processor can be a very competitive alternative, in terms of
performance, to conventional fully-fledged XML query engines that
provide DB-like functionality such as security, transaction
processing, and native storage.

\end{abstract}


\begin{dedication}

\begin{flushright}
\emph{To my parents,} \\ \emph{Ofelia and Juan Carlos}
\end{flushright}

\end{dedication}

\newpage


\begin{acknowledgements}

This thesis is the culmination of my graduate studies at the
University of Toronto. It goes without saying that no scientific
work can be carried out in isolation, and this one is no
exception.  I would like to thank here all the people who helped
me along the way.

First and foremost, I wish to express my utmost gratitude to
Alberto Mendelzon and Ren\'ee Miller, both of whom have been a
source of inspiration for many years. They were also the Alpha and
the Omega of my graduate studies: Alberto was my advisor during my
master's and the early years of my PhD, and also the first one to
suggest the idea of a framework upon which this thesis is based;
Renee supervised the final stages of the work and made sure that
all the pieces fitted together. Both provided me with support and
guidance well beyond the call of duty. They really made this work
possible.

Special thanks go to Jos\'e Mar\'ia Turull Torres and Alejandro
Vaisman, great friends and mentors. Jos\'e Mar\'ia introduced me
to the fascinating world of scientific research and encouraged me
to pursue graduate studies. I could never thank him enough for his
guidance in the first steps of my research career. Alejandro's
encouragement and insight were always invaluable, especially
during the most difficult times of my PhD (he was my thesis
advisor in disguise for many years). I consider them my academic
role models for their integrity and professionalism.

I would like to thank the members of my PhD committee, Kelly
Lyons, Thodoros Topaloglou, and John Mylopoulos, for their
insightful comments, and Frank Tompa for his thorough external
appraisal. I also wish to acknowledge the contribution of Mariano
Consens, who helped in the development of some core ideas of this
thesis.

I am deeply indebted to the administrative staff of the Department
of Computer Science for assisting me in so many different ways.
Joan Allen and Linda Chow deserve a special mention.

I am also grateful to the Department of Computer Science, the
Natural Sciences and Engineering Research Council of Canada, and
the IBM Center for Advanced Studies for their generous financial
support at different times over my years of graduate studies.

I want to express my heartfelt appreciation to my many friends in
Toronto and abroad, especially to my best officemate, Attila, and
to my musical buddies, Diego, Danny, Fabricio, Gustavo and Santi;
and in general to all those that had to put up with me for years:
Adriana, Adrian, Alberto, Alejandra, Christian, Ceci, Carlos,
Clau, Fernanda, Fernando, Frank, Jime, Lily, Lore, Mara, Pachi,
Patricia, Rosana, Sebastian and Vivi. Each in their one way made
my life more meaningful and enjoyable.

I am also in debt with my parents, Ofelia and Juan Carlos, for
helping me to become who I am. I owe them much, and regret that I
missed my best opportunities to repay. I dedicate this thesis to
them.

Above all, I have to thank a thousand times to my wife, Mariana,
the person without whom this thesis could have never been written.
Her support and unconditional love are beyond words. It is hard to
know who I would be without her; I hope never to have the occasion
to find out.

\end{acknowledgements}

\tableofcontents

\listoftables

\listoffigures

\end{preliminary}



\chapter{Introduction} \label{section:Intro}

XML is widely used as a common format for web accessible data
(e.g., hypertext collections like Wikipedia) as well as for data
exchanged among web applications (e.g., blogs, news feeds,
podcasts, web services messaging). This data is often referred to
as \emph{semistructured} for the lack of a clear separation
between data and metadata it represents: tags (metadata) and
content (data) are mixed together in the same XML file.

The vast majority of software tools used for managing XML rely on
XPath~\cite{W3C:XPath/02} as the core dialect for XML querying.
Hence, web developers use XPath queries for many of the tasks
involved in the processing of XML collections. Such collections
are normally handled one document at a time, whether the document
is an individual RSS\footnote{\small \tt{
http://www.rss-specifications.com/}} file (used by content
distributors to deliver to subscribers frequently updated content
over the Web), a single SOAP\footnote{\small \tt{
http://www.w3.org/TR/soap/}} message, or a Wikipedia article in
XHTML.

Even when XML collections have a schema (which can be either a
DTD\footnote{\small \tt{ http://www.w3.org/TR/REC-xml}} or an XML
Schema\footnote{\small \tt{ http://www.w3.org/TR/xmlschema-1/}}),
the actual structure present in each document may exhibit
significant variations for several reasons. First, schemas can be
very lax. One reason for this is the extensive use of the
\verb|<xsd:choice>| construct in XML schemas, which allows
optional elements to occur any number of times, including zero.
Such a construct is very common in RSS for instance. Second, a
schema can be very large and only subsets are actually used in a
given instance. This is the situation with several industry
specific standards that contain hundreds of elements, such as
UBL\footnote{\small \tt{ http://oasis-open.org/committees/ubl/ } }
or HR-XML\footnote{\small \tt{ http://hr-xml.org }}. UBL and
HR-XML are standard libraries of XML schemas that support a
variety of business processes. UBL is designed to handle supply
chain transactions such as purchase orders, shipping notices, and
invoices, whereas HR-XML contains schemas for human resource
management such as resumes, payroll information, and benefits
enrollment. Finally, a schema can be extended by using the
\verb|<xsd:any>| XML Schema construct, which allows arbitrary
content from other schemas to appear under a given element. Such a
construct enables different user communities to pick and choose
how to combine schemas. Consequently, it provides great
flexibility, but makes it harder to determine the structure of the
documents that actually appear in a given collection. Examples of
the \verb|<xsd:any>| extensions can be found in a wide variety of
industry standards, including RSS, UBL and HR-XML. For instance,
the UBL standard permits a contractor to represent invoice
documents that include HR-XML TimeCard elements for the contractor
employee's time and expenses. The actual structure of invoice
collections will vary significantly across contractors and
customers. If an enclosing messaging schema is used, even the UBL
and HR-XML fragments in the document can be replaced by other
invoicing and time billing schemas. In these scenarios, schemas
alone are insufficient for understanding the structure (metadata)
of the documents in the collection for either writing or
optimizing XPath evaluation.

A developer working with this type of collection faces several
challenges. She must learn enough about the structure present in
the XML collection to be able to write meaningful XPath queries.
She must also develop an understanding of how the XPath
expressions behave across different documents in the collection.
Even when a task deals with a single document at a time, the
developer needs to extrapolate the behaviour of queries over a
single document across the entire collection over which the task
may be repeatedly applied. In this context, understanding the
actual metadata of a web collection can be a significant barrier,
even for collections validated against a schema.

\begin{figure*}
    \centering
        \includegraphics{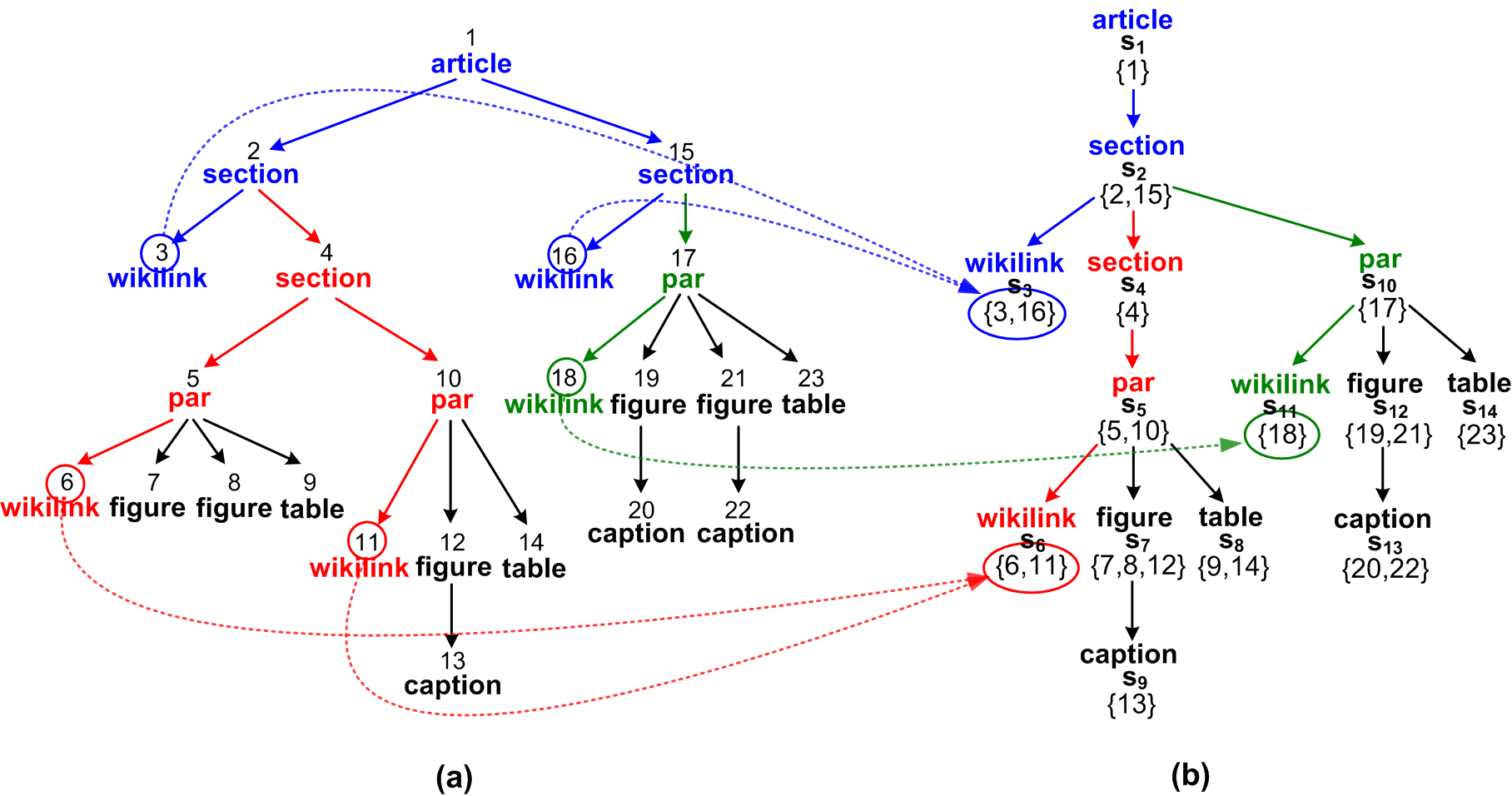}
    \caption{Wikipedia document graph (a) and its incoming path summary (b) }
    \label{fig:smallSummary}
\end{figure*}

XML structural summaries are graphs representing relationships
between sets of XML elements (i.e., extents). Unlike schemas,
which prescribe what may and may not occur in an instance,
summaries provide a description of the metadata that is actually
present in a given collection. Figure \ref{fig:smallSummary} (a)
shows the instance graph of a Wikipedia sample document in which
nodes correspond to XML elements in the document. Nodes have an id
and are labeled with their element names. The structure in Figure
\ref{fig:smallSummary} (b) is a typical summary that groups
together instance nodes with the same incoming label paths. In
such a summary, two nodes that have the same incoming label path
from the root belong to the same \emph{extent} (sets located below
each summary node in the figure). For instance, \emph{wikilink}
elements appear at the end of three different label paths --
\emph{article.section.wikilink} (in blue),
\emph{article.section.section.wikilink} (in red), and
\emph{article.par.wikilink} (in green). Consequently, there are
three different \emph{wikilink} nodes in the summary, one with
extent $\{3, 16\}$, another with extent $\{6,11\}$, and a last one
with extent $\{18\}$. Extents can also be viewed as
\emph{mappings} between instance (document) nodes and summary
nodes -- represented in the figure by dashed arrows linking
\emph{wikilink} nodes in the document graph (left) and
\emph{wikilink} extents in the incoming path summary (right). An
edge $(s_i, s_j)$ in the summary means that at least one node in
the extent of $s_i$ is the parent of at least one node in the
extent of $s_j$. For instance, an edge from node $s_7$ to $s_9$
means that some \emph{figure} elements within \emph{par} have
\emph{caption} elements, but not necessarily all of them have (for
this document only node 12 has a \emph{caption} element).

Describing metadata in semistructured collections was a major
motivation in one of the earliest summary proposals in the
literature \cite{NUWC97,GW97}. Since then, research on summaries
has focused on query processing, making summaries one of the most
studied techniques for query evaluation and indexing in XML (and
other semistructured) data models
\cite{MS99,KBNK02,KSBG02,QLO03,BCF+05}, as well as for providing
statistics useful in XML query estimation and optimization
\cite{PG06b}.

Most of the existing summary proposals define all extents using
the same criteria, hence creating {\em homogeneous} summaries.
These summaries are based on common element paths (in some cases
limited to length $k$), including incoming paths (e.g.,
representative objects \cite{NUWC97}, dataguides \cite{GW97},
1-index \cite{MS99}, ToXin \cite{RM01}, A(k)-index \cite{KSBG02}),
both incoming and outgoing paths (e.g., F\&B-Index \cite{KBNK02}),
or sequences of outgoing paths (e.g., Skeleton \cite{BCF+05}). The
few examples of {\em heterogeneous} summaries that can
adapt/change their structure based on a dynamic query workload
(e.g., APEX \cite{CMS02}, D(k)-index \cite{QLO03}, XSKETCH
\cite{PG06b}) compute the extents from statistics and workload
information.

However, none of these proposals can help us to find elements
based on order and cardinality criteria. Consider again the
instance in Figure \ref{fig:smallSummary}. What are the \emph{par}
elements that contain two \emph{figures}? How many \emph{section}
elements contain a \emph{figure} with \emph{caption} next to a
\emph{table}? How many of those contain more than one
\emph{figure}? These are questions that cannot be answered with
any of the summaries mentioned above.

Moreover, since these proposals are algorithmically defined, it is
hard to determine how they can be used together for processing
today's increasingly heterogeneous and large web collections
effectively. Specifically, the summary information is not defined
declaratively, limiting the ease with which these summaries can be
used within standard data management tasks.

In this thesis, we propose a novel approach for flexibly
summarizing the structure of metadata actually present in an XML
collection. We introduce \dx, a framework that supports
constructing \emph{heterogenous} summaries, where each set in the
partition can be defined by means of path regular expressions on
axes, or \emph{axis path regular expression} (\emph{AxPRE}, for
short). AxPREs provide the flexibility necessary for
\emph{declaratively} defining complex mappings between instance
nodes and summary nodes capable of expressing order and
cardinality, among other properties. Each AxPRE can be specified
by the user or obtained from any expression in the complete XPath
language (all the axes, document order, use of parenthesis, etc.).
Given an arbitrary XPath expression posed by the user, \dx\ can
create a partition defined by an AxPRE that captures exactly the
structural commonality expressed by a query. AxPRE summaries have
a unique capability that makes them suitable for describing the
structure of XML collections: they are the first summaries capable
of declaratively defining and refining the summary extents using a
powerful language. In addition, \dx\ summaries express
relationships between instance nodes that go beyond the
traditional parent-child (e.g., next sibling, following,
preceding, etc.). Last but not least, \dx\ captures most summary
proposals in the literature by providing a declarative definition
for them for the first time.

This thesis argues that \dx\ can significantly help not only in
the understanding of the structure of large collections of XML
documents, but also in the evaluation of XPath queries posed on
them. In fact, \dx\ summaries can also be used to significantly
speed up (and scale up) XPath evaluation over existing
file-at-a-time tools, enabling fast exploration of the results of
XPath workloads on large collections. The experimental results
demonstrate that using a summary created from a given workload can
produce query evaluation times that are two orders of magnitude
better than using existing summaries (in particular, summaries on
incoming paths like 1-index \cite{MS99}, APEX \cite{CMS02},
A(k)-index \cite{KSBG02}, and D(k)-index \cite{QLO03}). The
experiments also validate that \dx\ summaries allow file-at-a-time
XPath processors to be a competitive light-weight alternative (in
terms of performance) to conventional DB-like XML query engines
supporting additional functionality such as security, transaction
processing, and native storage.

\dx\ also has applications to helping a user write and understand
XPath queries on large XML collections. Several software tools
have been developed to help XPath users debug query expressions
(e.g., Oxygen XML Editor\footnote{\small \tt{
http://www.oxygenxml.com/}}, Altova XMLSpy\footnote{\small \tt{
http://www.altova.com/}}, etc.) A recent research project includes
a tool, \xp -Eclipse \cite{CLR07}, that provides visual
explanations of XPath expressions. An explanation returns
precisely the nodes in a document that contribute to the answer, a
useful debugging technique. However, the main limitation of
traditional XPath debugging tools in the context of large XML
collections is that they provide debugging mechanisms only for a
\emph{single} document. Understanding queries over collections
containing thousands of documents (or even 650,000 documents, like
in the Wikipedia XML Corpus \cite{wikipediaxml:2006}) using these
tools can be an impractical and very time-consuming task. \dx\
provides an important foundation on which such a large-scale XML
collection understanding tool could be built, as evidenced by the
\dxp\ tool presented in Appendix \ref{sec:debugExample}.

\section{Major contributions}

This thesis identifies the growing need for describing the
structure of web collections (encoded in XML) using mechanisms
that go beyond providing one or more schemas. We propose the use
of highly customizable summaries that represent the actual
structure of metadata labels as used in a given collection. The
following are the major contributions of this thesis.

\subsection{AxPRE summaries}

AxPRE summaries rely on the novel concept of a \emph{summary
descriptor} (\textbf{SD}). Traditional summaries consist of a
labeled graph that describes the label paths in the instance
(which we call an SD graph) together with an \emph{extent}
relation between summary nodes and sets of instance nodes. An SD
incorporates three key original features:

\paragraph{A description of the \emph{neighborhood} of a node expressed
by path regular expressions on axes (i.e., binary relations
between nodes), \emph{AxPREs} for short (Chapter
\ref{sec:axpres}).} AxPREs are evaluated on an \emph{axis graph},
which is an abstract representation of the XPath data model
\cite{W3C:XPath/02} extended with edges that represent XPath axis
binary relations. Edges are labeled by axis names and nodes are
labeled by element or attribute names (including namespaces), or
by new labels defined using XPath.

Given an axis graph $\mathcal{A}$, an AxPRE $\alpha$ applied to a
node $v$ in $\mathcal{A}$ returns an \emph{AxPRE neighbourhood} of
$v$ which provides a description of the subgraph local to $v$ that
satisfies $\alpha$. The AxPRE neighbourhood of $v$ by $\alpha$ is
computed by intersecting the automaton constructed from the axis
graph and the automaton accepting the language generated by the
AxPRE and all its prefixes.

The AxPRE neighbourhood of a node $v$ is used to determine to
which equivalence class $v$ belongs. That is, if two nodes in
$\mathcal{A}$ have \emph{similar} AxPRE neighbourhoods (i.e. they
cannot be distinguished by $\alpha$), they belong to the same
equivalence class. This way, an AxPRE can be used to define a
\emph{partition} of nodes in $\mathcal{A}$ in which each set is
the \emph{extent} of a node $s$ in the SD. The notion of
similarity we use is the familiar notion of
bisimulation~\cite{PT87}.

The use of AxPREs neighbourhoods supports the definition of
summaries that go beyond the traditional parent and child
hierarchical relationships covered by the abundant literature on
summaries \cite{GW97,MS99,KBNK02,KSBG02,PG02a,QLO03,BCF+05,PG06b}.
In particular, AxPREs can describe \emph{heterogeneous} SDs, i.e.,
SDs described by multiple AxPREs.

\paragraph{An \emph{extent expression} (\textbf{EE}) capable of computing
precisely the set of elements in the extent of a given SD node
(Chapter \ref{sec:acyclic}).} Since an AxPRE $\alpha$ is used to
compute by bisimulation an entire partition, we can say that all
sets in the partition share the same AxPRE $\alpha$. Thus, AxPREs
cannot be used to uniquely identify each equivalence class
(extent) in such partition (unless the partition contains only one
set).

For a large class of neighbourhoods, it is possible to precisely
characterize the extent of an SD node $s$ with a new type of
expression we call \emph{extent expression} (EE, for short). The
EE $e_s$ of $s$ with AxPRE $\alpha$ is generated from the
\emph{bisimilarity contraction} of the $\alpha$ neighbourhoods of
the elements in the extent of $s$. (Recall that all nodes in the
extent of $s$ have bisimilar AxPRE neighborhoods.) Thus, we pick
any element in the extent of $s$, compute its $\alpha$
neighbourhood, and then compute its bisimilarity contraction. The
\emph{representative neighbourhood} thus obtained is guaranteed to
be bisimilar to all neigbourhoods in the extent of $s$. A
representative neighbourhood provides the sequence of axis
compositions and labels that will appear in the EE that computes
the extent of $s$. EEs can be expressed in XPath and function like
virtual views (see Chapter \ref{sec:changingSDswithxpath}).

\paragraph{The notion of \emph{AxPRE refinements} of SD nodes
(Chapter \ref{sec:NeighbourhoodStabilization}).} Exploring
collections of XML documents typically requires knowledge of the
metadata present in the collection. SDs provide a descriptive tool
for representing metadata as SD graphs. The description provided
by a node in the SD can be changed by an operation that modifies
its AxPRE and thus its AxPRE neighbourhood. This operation is
called an \emph{AxPRE refinement} of an SD node. Refinement refers
to applying summarization to selectively produce more or less
detailed SDs.

The notion of refinement is well-known in the XML literature
\cite{PT87}. Intuitively, two nodes in the same equivalence class
may be refined into different classes, and two nodes from
different classes will always be refined into separate classes. An
SD node can be refined by changing its AxPRE definition. This
produces SDs that are tailored to the exploration needs of the
user. Using successive node refinements, SD nodes can be refined
to produce SDs that provide a more detailed description of the
data.

Previous proposals perform global refinements on the entire SD
graph \cite{KBNK02,KSBG02} or local refinements based on
statistics or workload \cite{QLO03,HY04,PG06b}, without the
ability to define the refinement declaratively. In contrast, we
can precisely characterize the neighbourhood considered for the
refinement with an AxPRE \cite{CRV08}.

The notion of refinement is tightly related to that of
\emph{stabilization}. An edge stabilization determines the
partition of an extent into two sets based on the participation of
the extent nodes in the axis relation the edge represents.

\subsection{Refinement lattice} We show the existence of a
hierarchical relationship between summaries and provide a concise
description of the hierarchy within the \dx\ framework based on a
\emph{refinement lattice}. A refinement lattice describes a
refinement relationship between entire summaries.

The \dx\ refinement lattice provides a mechanism for capturing
earlier summary proposals, and understanding how those proposals
relate to each other and to richer SDs that were never previously
considered in the literature (see Chapter \ref{sec:lattice}). Each
node in the lattice corresponds to a homogeneous SD defined by an
AxPRE. The top (coarsest) summary of the lattice corresponds to
the label SD where each node is partitioned by label, and the
bottom (finest) summaries of the lattice each corresponds to a
distinct combination of axes.

\subsection{System implementation}

In Chapter~\ref{Section:Implementation}, we present the
implementation of the \dx\ summarization engine for interactively
creating and refining AxPRE summaries given large collections of
XML documents. Chapter~\ref{Section:Experiments} provides
experimental results that validate the performance of the
techniques employed by \dx.

The engine uses Berkeley DB Java Edition\footnote{{\small
\tt{http://www.oracle.com/technology/products/berkeley-db/je/index.html}}}
to store and manage indexed collections, and supports an arbitrary
XPath processor for the evaluation of XPath expressions. A visual
interactive tool based on the \dx\ framework, \dxp\ (see Appendix
\ref{sec:debugExample}), was developed as an
Eclipse\footnote{{\small \tt{http://www.eclipse.org/}}} plug-in.
In addition to the \dx\ summarization engine presented in this
thesis, \dxp\ provides retrieval and visualization tools
implemented by other colleagues \cite{ACSR08}.

Our experiments (employing gigabyte XML collections) provide
strong evidence of the advantages of using \dx\ to build and
exploit summaries for exploration and XPath query evaluation.
These results demonstrate that the simple mechanism of accessing a
summary extent employed by the \dx\ implementation yields speedup
factors of over two orders of magnitude over commercial and open
source implementations.

\subsection{Answering queries using extents} For evaluating
a query using an SD, we need to find the SD nodes that participate
in the answer. Since our framework relies on XPath EEs for
defining the extents, the problem of answering queries using
extents is related to that of XPath containment
\cite{XPathContainment:2004}.

\dx\ can derive AxPREs from queries and use them to change the
description provided by the SD. Since AxPREs describe only
structural constraints and XPath queries may contain predicates on
values, extents resulting from AxPRE manipulation rarely provide
the exact answer without further filtering. The main reason for
this is that the addition of an XPath value predicate either
reduces the size of the answer or leaves the answer unchanged.
Thus, \dx\ finds first the SD nodes that participate in the answer
(i.e., those whose extents contain at least part of the answer),
then evaluate the entire expression on them and take the union of
the results to get the \emph{exact} answer (see Chapter
\ref{sec:candidates}). The experimental results provided in
Chapter \ref{sec:queryevalwithSDs} considerably expand the
preliminary results presented in \cite{CR07}.

\section{Motivating example: exploring RSS feeds with summaries and XPath queries} \label{sec:workloadExample}

This section walks through a concrete example to illustrate how
\dx\ summaries can help developers perform collection-wide
exploration and XPath query evaluation.

Consider a developer, Sue, who has to implement a web application
that retrieves RSS feeds from several content providers to produce
an aggregated meta-feed. The feed may span several days or weeks,
and there might be more than one item in the feed per day.
Figure~\ref{fig:RSS} shows the instances of two sample RSS feeds
represented as \emph{axis graphs}.

\begin{figure}
    \centering
        \includegraphics{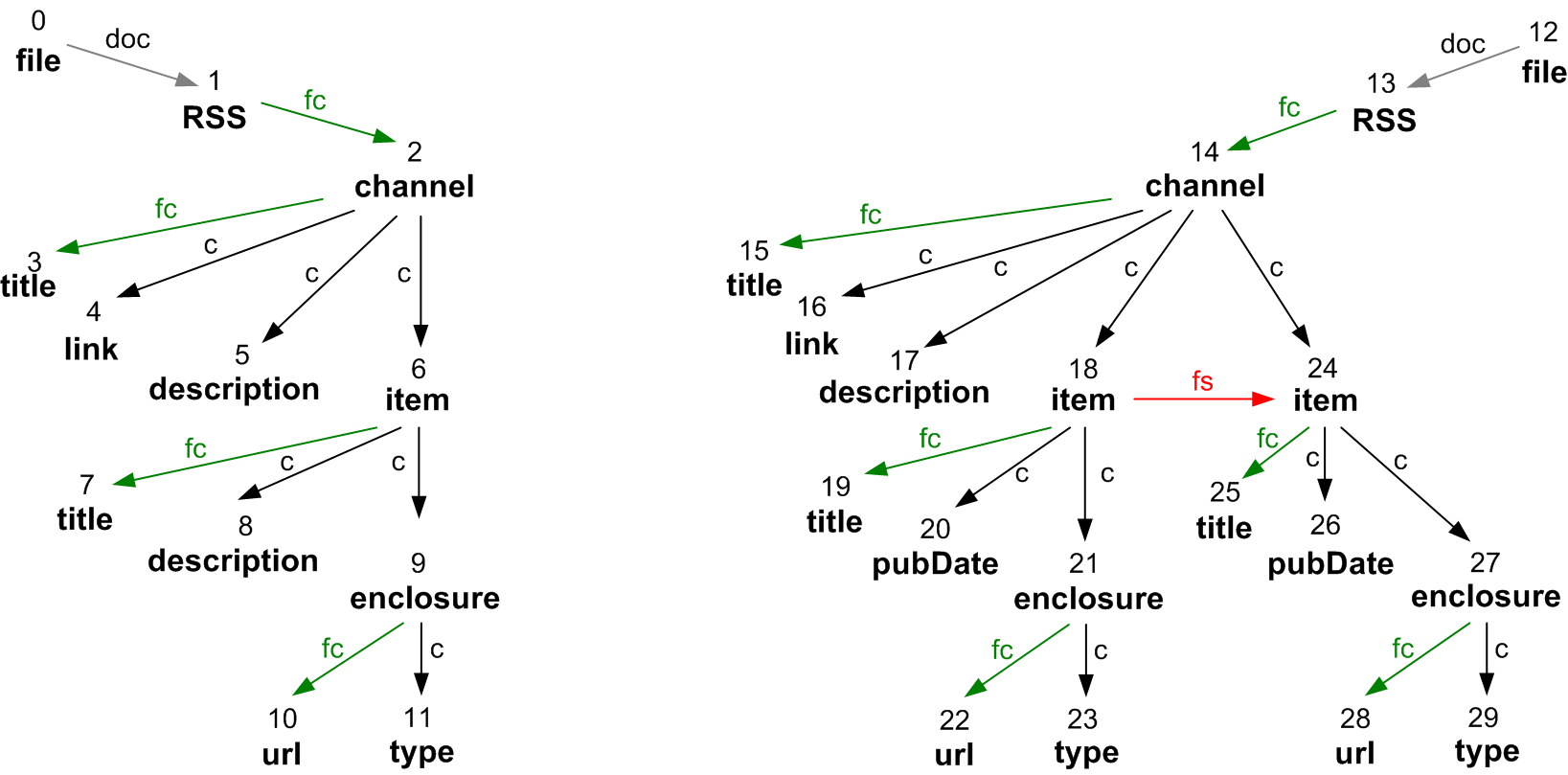}
    \caption{Axis graphs of RSS feed samples}
    \label{fig:RSS}
\end{figure}

An axis graph can display selected binary relations between
elements in an XML document tree, like $doc,$ $c,$ $\mathit{fs},$
and $\mathit{fc}$ shown in the figure (shorthands for XPath axes
\emph{document}, \emph{child}, \emph{following-sibling}, and for
the derived axis \emph{firstchild}, respectively). The semantics
of these axes is straightforward: the edge from element $6$ to $7$
labeled $\mathit{fc}$ means that $7$ is the first child of $6$ in
document order, and the edge from element $18$ to $24$ labeled
$\mathit{fs}$ means that $24$ is a following sibling of $18$ in
document order. For simplicity, even though every first child is
also a child, we do not draw the $c$ edge between two nodes when
an $\mathit{fc}$ edge exits between them. Being binary relations,
axes have inverses, e.g., the inverse of $c$ is $p$ (shorthand for
\emph{parent}) and the inverse of $\mathit{fs}$ is $ps$ (shorthand
for \emph{preceding-sibling}). These inverses are not shown in the
figure.

Using \dx, Sue can create a \emph{summary descriptor} (SD for
short) like the one shown on Figure \ref{fig:smallRSS-combined}
(a). This \emph{label SD}, created from the two feeds in Figure
\ref{fig:RSS}, partitions the elements in the feeds by element
name. For example, SD node $s_6$ represents all the $item$
elements in the two documents, $\{6,18, 24\}$ (this set is called
the \emph{extent} of $s_6$).

\begin{figure}
    \centering
        \includegraphics{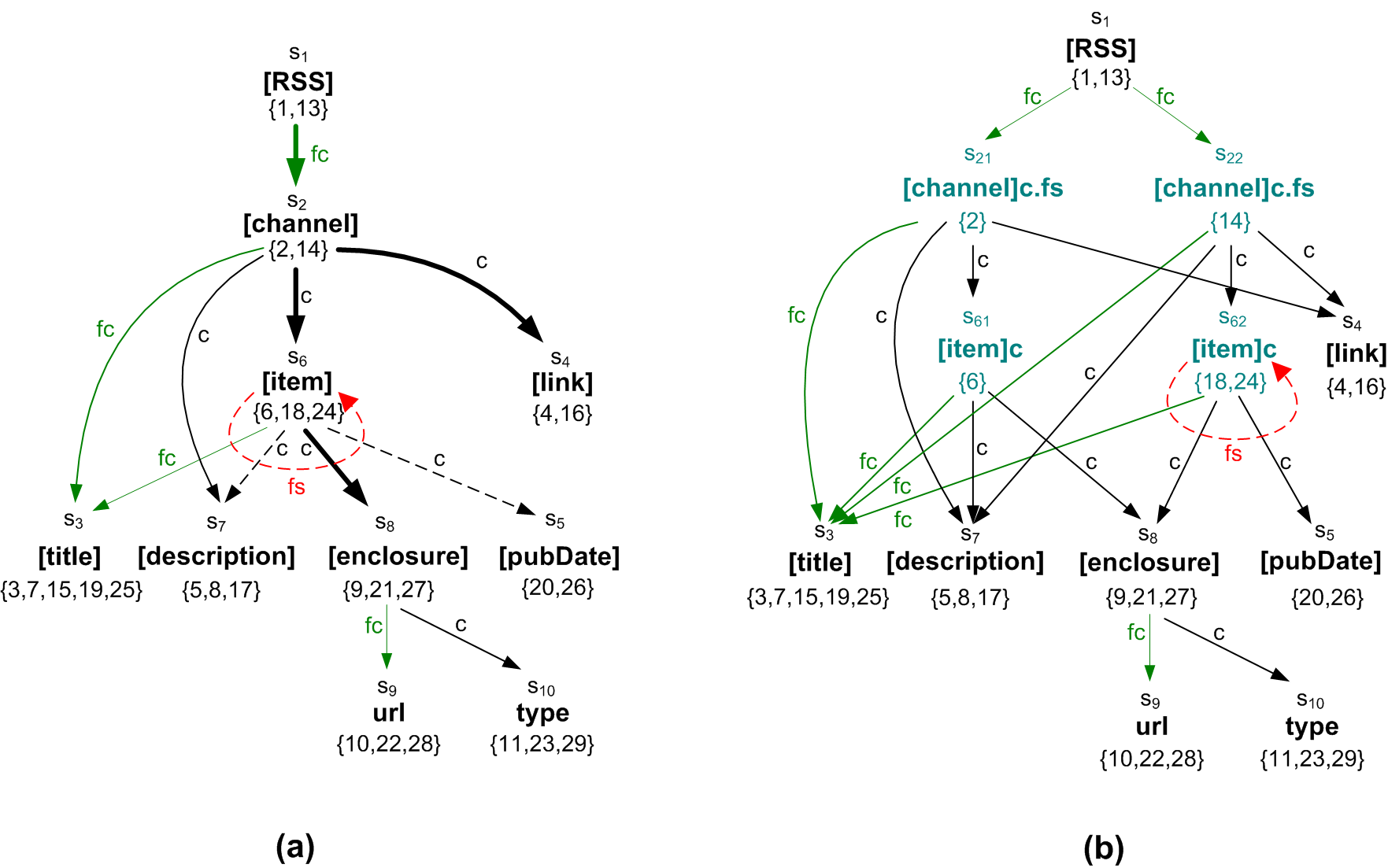}
    \caption{Label SD (a), and heterogeneous SD (b) of the RSS feed samples}
    \label{fig:smallRSS-combined}
\end{figure}

An SD edge is labeled by the axis relation it represents. For
instance, edge $(s_6,s_5)$ is labeled by $c$, which means that
there is a $c$ axis relation between elements in the extent of
$s_6$ and $s_5$. Figure \ref{fig:smallRSS-combined} (a) shows
three kinds of edges, depending on properties of the sets that
participate in the axis relation: dashed, regular, and bold.
Dashed edges, like $(s_6,s_5)$ labeled $c$, mean that some element
in the extent of $s_6$ has a child in the extent of $s_5$. Regular
edges, like ($s_6,s_3$) labeled $\mathit{fc}$, mean that every
element in the extent of $s_6$ has a first child in the extent of
$s_3$. Finally, bold edges, like $(s_6,s_8)$ labeled $c$, mean
that every element in the extent of $s_8$ is a child of some
element in the extent of $s_6$ and that every element in the
extent of $s_6$ has some child in the extent of $s_8$.

From the label SD Sue learns that \emph{channel} elements in the
collection always contain \emph{title}, \emph{link},
\emph{description}, and \emph{item} subelements. However, the
structure of \emph{item} elements may vary. An \emph{item} in the
two sample feeds always includes \emph{title} and \emph{enclosure}
elements, but may contain any combination of \emph{description}
and \emph{pubDate} elements. Note that the label SD does not
provide information on exactly which combinations actually appear.
At this point Sue has two options:

\begin{enumerate}
\item She can interactively \emph{refine} the SD node $s_2$ in the
label SD in order to learn how many different types of channels
exist in the collection (i.e., how many subsets of \emph{title},
\emph{enclosure}, \emph{description} and \emph{pubDate} are
present within \emph{item} elements). \item Since she already
knows that some \emph{item} elements have a \emph{pubDate} from
the label SD and she is interested in channels that contain such
items, she can write query Q1 to retrieve them.

\begin{center}
{\small \texttt{Q1 = /rss/channel[item[pubDate]] } }
\end{center}

\end{enumerate}

Sue can now decide either to run Q1 using the current SD or to
make \dx\ \emph{adapt} the current SD to Q1. If she picks the
former option, \dx\ finds the only SD node that contains a
superset of the answer ($s_2$) and runs Q1 on its entire extent.
If Sue chooses the latter option, \dx\ changes the SD by
partitioning the single \emph{channel} node $s_2$ in Figure
\ref{fig:smallRSS-combined} (a), which represents all channels in
the collection, into two \emph{channel} nodes: one with a
\emph{pubDate} within their \emph{item} elements and another
without a \emph{pubDate} ($s_{22}$ and $s_{21}$ in Figure
\ref{fig:smallRSS-combined} (b), respectively). Both SDs can be
used to evaluate query Q1, but notice this latter refinement (the
SD of Figure~\ref{fig:smallRSS-combined} (b)) will yield a more
efficient evaluation.

Summaries in \dx\ are defined and manipulated via AxPREs. AxPREs
describe the \emph{neighbourhood} of the elements in a given
extent. A neighbourhood of an element $v$ for an AxPRE $\alpha$ is
the subgraph local to $v$ that matches $\alpha$. For instance, the
$p^*$ AxPRE describes the neighbourhood of $v$ containing all
label paths from $v$ to the root, $c^*$ all label paths from $v$
to the leaves, and $\mathit{fc}.ns^*$ the sequence of $v$'s child
labels. AxPREs can also be derived from a query in order to adapt
an SD to it. For example, the $[channel].c.c$ AxPRE of node
$s_{21}$ in Figure \ref{fig:smallRSS-combined} (b) was derived
from Q1 and describes the neighbourhood of \emph{channel} elements
with common outgoing label paths of length $2$ (more on this in
Chapter \ref{sec:axpres}). Sue could have written the
$[channel].c.c$ herself had she wanted to \emph{refine} the
\emph{channel} node $s_2$ according to the substructure of the
\emph{channel} elements in the extent of $s_2$ (since she knows
from the label SD that the variability within \emph{channel}
elements may only come from \emph{description} and \emph{pubDate}
within \emph{item} subelements, the $c.c$ AxPRE representing
outgoing label paths of length $2$ suffices).

Suppose further that Sue is also interested in \emph{item}
elements containing both \emph{title} and \emph{enclosure}
subelements, but she does not know whether such items exist in the
collection and, if they do, how common they are. In addition, she
wants those items to be part of a series (i.e., to belong to
\emph{channel} elements that contain more than one \emph{item}
element, as done in feeds for podcasts published daily).
Therefore, she writes another query:

\begin{center}
{\small \texttt{Q2 = /rss/channel[item/following-sibling::item]}}

{\small \texttt{
\textcolor[rgb]{0.50,0.50,0.50}{[not(pubDate=../item[1]/pubDate)]}/item[title][enclosure]
} }
\end{center}

Q2 contains structural (in black) and non-structural (in grey)
XPath constructs. The expression that results from removing all
non-structural constraints is called the \emph{structural
subquery} of Q2. A structural subquery provides insight into the
behaviour of the entire query and can be used by \dx\ to refine an
SD.

As with Q1, Sue can decide to either evaluate Q2 on the current SD
(the label SD with the refined \emph{channel} node) or to add Q2
to the workload and make \dx\ adapt the current SD. Assuming she
chooses the second option, the system partitions the \emph{item}
node $s_6$ from Figure \ref{fig:smallRSS-combined}(a) into the
nodes $s_{61}$ and $s_{62}$ in Figure
\ref{fig:smallRSS-combined}(b) that describe the structure of the
collection with respect to the workload including Q2 and Q1. Note
that the extent of node $s_{62}$ is exactly the answer to the
structural subquery of Q2, and thus a superset of the answer of
Q2. The elements in this extent are called \emph{candidate
elements}. Hence, by adapting the SD to the structural subquery,
\dx\ has considerably reduced the search space for computing the
entire query.

In a document-at-a-time approach to query evaluation, adapting an
SD to a workload can reduce the number of documents on which
queries in the workload need to be evaluated, potentially yielding
a significant speedup (see Chapter~\ref{Section:Experiments}).
That is, after adapting the SD to a given query $Q$, \dx\ can
evaluate $Q$ only on those documents (called \emph{candidate
documents}) that are guaranteed to provide a non-empty answer for
the structural subquery of $Q$. Those candidate documents that do
contain an answer for the entire query are called \emph{answer
documents}.

It is important to note that \dx\ can recognize two kinds of
channels with different structure beyond the elements directly
contained by them, a capability not available using DTD's (unless
channel elements are renamed, which is not a possibility when the
original DTD or the instances cannot be modified). In particular,
proposals to infer a DTD from an instance (such as
\cite{BNST06,GGRS+03}) by suggesting (general, but succinct)
regular expression from the strings of child elements, do not help
to identify the two kinds of channels as done above. For instance,
the DTD expression {\small \texttt{<!ELEMENT channel (title, link,
description, item)>}} can be inferred for the \emph{channel}
elements occurring in the instances shown in Figure~\ref{fig:RSS}.
However, a DTD can only give a rule for the children of
\emph{channel}, there is no mechanism for giving rules relating
\emph{channel} elements to their grandchildren (or any other
elements farther away). In contrast, the AxPRE summary in
Figure~\ref{fig:smallRSS-combined} (b) can distinguish between a
\emph{channel} containing an \emph{item} with a \emph{pubDate}
element from those that contain a \emph{description}, and also
between \emph{item} elements that belong to a multi-item
\emph{channel} from single-item ones.

As we will show in this thesis, \dx\ is not only more expressive
than DTD's and XML Schemas, but also more expressive than other
summary proposals making it a robust foundation for managing large
document collections.

\section{Organization}

This thesis is structured as follows. Chapter~\ref{sec:related}
gives an overview of the large body of related work in the
literature. Chapter~\ref{sec:axpres} introduces the \dx\
framework, including the AxPRE language and some basic notions
such as neighbourhood, bisimilarity, and summary descriptor (SD).
Chapter~\ref{sec:lattice} revisits some of the related work
discussed in Chapter~\ref{sec:related} and explains how they can
be captured by the \dx\ framework and how \dx\ offers significant
new functionality. Chapter~\ref{sec:refinements} presents two new
operations, AxPRE refinement and stabilization, for declaratively
changing the description provided by an SD using AxPREs.
Refinement and stabilization are central to the use of summaries
for both structure understanding and query processing.
Chapter~\ref{sec:changingSDswithxpath} introduces a novel
mechanism to characterize an SD node with an XPath expression
whose evaluation returns exactly the elements in the extent. It
also discusses how to compute AxPRE refinements and stabilizations
with XPath expressions and how to evaluate XPath queries using
\dx\ summaries. Chapter~\ref{Section:Implementation} describes the
implementation of the \dx\ summarization engine for creating and
manipulating SDs of XML collections.
Chapter~\ref{Section:Experiments} provides experimental results,
using gigabyte size XML collections, that validate the performance
of the techniques employed by our framework. We conclude in
Chapter~\ref{section:conc} by presenting some future research
issues. In addition, Appendix~\ref{Section:XPathLanguage} provides
a concise definition of the formal semantics of XPath 1.0, and
Appendix~\ref{sec:debugExample} presents a visual interactive tool
built on top of the \dx\ summarization engine.

\chapter{Related work} \label{sec:related}

In this chapter, we discuss contributions from the literature on
structural summaries and other areas related to our work, such as
path summaries for object-oriented data, hierarchical encodings,
answering XML queries using views, and validating summaries.

\section{Structural summaries}

The large number of summaries that have been proposed in recent
years clearly establishes the value and usefulness of these
structures for describing semistructured data, assisting with
query evaluation, helping to index XML data, and providing
statistics useful in XML query optimization.

Most of the summary proposals in the literature define synopses of
predefined subsets of paths in the data. They construct a labeled
graph that represents relationships between sets of XML elements.
Examples of such summaries are region inclusion graphs (RIGs)
\cite{CM94}, representative objects (ROs)\cite{NUWC97}, dataguides
\cite{GW97}, reversed dataguides \cite{LS00}, 1-index, 2-index and
T-index \cite{MS99}, and more recently, ToXin \cite{RM01},
A(k)-index \cite{KSBG02}, F-Index, B-index, and F\&B-Index
\cite{KBNK02}. Dataguides and ROs group nodes into sets according
to the label paths incoming to them (each node may appear more
than once in the dataguide if the document instance is not just a
tree). RIGs, 1-index, T-index, ToXin, F\&B-Index, and F+B-Index,
on the other hand, partition the data nodes into equivalence
classes (called \emph{extents} in the literature) so that each
node appears only once in the summary. The partition is computed
in different ways: according to the node labels (RIGs), the label
paths incoming to the nodes (1-index, ToXin, A(k)-index), the
label paths going out from the nodes (reversed dataguides), or
label paths both incoming and outgoing (F\&B-Index and F+B-Index).
The length of the paths in the summary also varies: ToXin, 1-index
and F\&B-Index/F+B-Index summarize paths of any length, whereas
A(k)-index is a synopsis of paths of a fixed length. Updates to
structural summaries have been studied in \cite{KBNS02} and
\cite{YHSY04}.

RIGs were one of the first summaries proposed in the literature,
introduced in the context of region algebras \cite{CM94,YT03}.
Dataguides \cite{GW97} group nodes in a rooted data graph into
sets called \emph{target sets} according to the label paths from
the root they belong to. Since the label paths form a language,
its deterministic finite automaton (DFA) is used as a more concise
representation of the label paths. The construction of a dataguide
from a data graph is equivalent to the conversion of a NFA (the
XML tree) into a deterministic finite automaton (the dataguide)
\cite{NUWC97}.

An index family was presented in \cite{MS99} (1-index, 2-index,
and T-index). Like dataguides, the 1-index summarizes root-to-leaf
paths. In the 1-index, the nodes of a XML tree are partitioned
into equivalence classes according to the label paths they belong
to. Since the 1-index extents constitute a partition of the XML
nodes, the number of 1-index nodes can never be bigger than the
XML tree. The extreme case is the one in which every XML node
belongs to a separate equivalence class (which is in fact the data
instance). The 1-index partition is computed by using
\emph{bisimulation}~\cite{PT87}.

Based also on bisimulation, the A(k)-index was introduced in
\cite{KSBG02}. The construction of the summary is based on
$k$-bisimilarity (bisimilarity computed for paths of length k).
Thus, the A(0)-index creates the partition based on the labels of
the nodes (0-bisimilarity), and the A(h)-index uses
$h$-bisimilarity which creates the partition based on incoming
label paths of length $h$.

Another index family was introduced in \cite{KBNK02}. The
F\&B-Index construction uses bisimulation like the 1-index, but
applied to the edges and their inverses in a recursive procedure
until a fix-point. With this construction, the F\&B-Index's
equivalence classes are computed according to the incoming and
outgoing label paths of the nodes. The same work introduces the
F+B-index, which applies the recursive procedure only twice, once
for the edges and another reversing the edges. Both F\&B-Index and
F+B-index are special cases of the BPCI(k,j,m) index, where $k$
and $j$ controls the lengths of the paths  and $m$ the iterations
of the bisimulation on the edges and their inverses.

ToXin consists of three index structures: the ToXin schema, the
path index, and the value index. The ToXin schema is equivalent to
a strong dataguide. The path index contains additional structures
that keep track of the parent-child relationship between
individual nodes in different extents. A recent proposal,
TempIndex \cite{MRV04} extends ToXin with the temporal dimension
in order to speed-up path queries on a temporal XML data model.
TempIndex summarizes incoming paths that are \emph{valid}
continuously during a certain time interval and is part of the
TSummary framework \cite{RV08}.

Based on the A(k)-index, a recent proposal \cite{FGWG+07} defines
partitions of paths, rather than nodes, called P(k)-partitions --
where $k$ is the maximum length of the paths summarized. This work
also introduces an algebraization of the navigational core of
XPath in order to define XPath fragments that can be coupled to
P(k)-partitions for fast evaluation of queries in the fragments.
Since this proposal is based on navigational XPath, it supports
only expressions containing composition of $parent$, $ancestor$,
$child$, and $descendant$ axes. In contrast, \dx\ can be used to
evaluate expressions in the complete XPath language (with all the
axes, functions, use of parenthesis, etc.).

Other summaries are augmented with \emph{statistical information}
of the instance for selectivity estimation, including
path/branching distribution (XSKETCH \cite{PG02a,DPGM04}), value
distributions (XCLUSTER \cite{PG06}), and additional statistical
information for approximate query processing (TREESKETCH
\cite{PGI04}).

A few \emph{adaptive} summaries like APEX \cite{CMS02}, D(k)-index
\cite{QLO03}, and M(k)-index \cite{HY04} use dynamic query
workloads to determine the subset of incoming paths to be
summarized. APEX is a summary of frequently used paths that
summarizes incoming paths to the nodes and adapts to changes in
the workload by changing the set of path considered in the
synopsis. That is, instead of keeping all paths starting from the
root, it maintains paths that have some ``support'' (i.e., paths
that appear a number of times over a certain threshold in the
workload). The workload APEX considers are expressions containing
a number of $child$ axis composition that may be preceded by a
$descendant$ axis, without any predicate. APEX summarizes incoming
paths to the nodes and adapts to changes in the workload by
changing the set of paths summarized. D(k)-index and M(k)-index,
in contrast, summarize variable-length paths based on both the
workload and local similarity (the length of each path depends on
its location in the XML instance).

There has been almost no work on summaries that capture the node
ordering in the XML tree: the only proposals we are aware of are
the early region order graphs (ROGs) \cite{CM94} and the Skeleton
summary \cite{BCF+05} that clusters together nodes with the same
subtree structure. Skeleton has additional structures that store
relationships between individual nodes that belong to different
equivalence classes.

In contrast to these proposals, \dx\ is capable of declaratively
defining complex mappings between instance nodes and summary nodes
for expressing order, cardinality, and relationships that go
beyond the traditional parent-child (e.g., next sibling,
following, preceding, etc.) In addition, \dx\ provides a
declarative definition for the first time for most of the
proposals discussed above (for more details on how \dx\ captures
other structural summaries see Chapter \ref{sec:lattice}).

\section{Path summaries for OO data}

We can trace the origin of structural summaries for XML to the
OODB community. This community has been quite active in the past
in the area of path summaries for object-oriented data. Examples
are path indexes \cite{Ber94}, access support relations
\cite{KM90}, and join index hierarchies \cite{XH94}. All three
proposals materialize frequently traversed paths in the database.
Access support relations are designed to support joins along
arbitrary reference chains leading from one object instance to
another. They also support collection-valued attributes by
materializing frequently traversed reference chains of arbitrary
length. Access support relations are a generalization of the
binary join indices originally proposed for the relational model
\cite{Val87}. One fundamental difference with respect to join
indices, however, is that rather than relating only two relations
(or object types), access support relations materialize access
paths of arbitrary length.

A path index can materialize the same class of paths as an access
support relation. It stores the sequence of nodes (objects) that
define a given path. In contrast, a join index hierarchy
constructs hierarchies of join indices to optimize navigation via
a sequence of objects and classes. A join index stores the pairs
of identifiers of objects of two classes that are connected via
logical relationships. Since all these OODB approaches are based
on the paths found in the OO schema, they can only be adapted to
XML documents for which either a DTD or an XML Schema is present.
In contrast, \dx\ permits summarization of collections without any
schema.

\section{Hierarchical encodings}

We should mention that, in addition to the use of summaries, query
evaluation can be facilitated by \emph{encoding} the hierarchical
structure of an XML instance. \emph{Node encoding} evaluations use
some sort of interval encodings \cite{SK85} to label each node
with its positional information within the XML instance. This
positional information is used by join algorithms to efficiently
reconstruct paths and label paths. Recent proposals for node
encoding evaluations are region algebras \cite{CM94,YT03}, path
joins (XISS) \cite{LM01}, relative region coordinates
\cite{KYU01}, structural joins \cite{AKJP+02,CVZ+02}, holistic
twig joins \cite{BKS02,JWLY03}, partition-based path joins
\cite{LM03}, XR-Tree \cite{JLWO03}, PBiTree \cite{WJLY03,VMT04},
extended Dewey encoding for holistic twig joins \cite{LLCC05}, and
FIX \cite{ZOIA06}, a feature-based indexing technique.

\emph{Structural encoding} proposals are based on mapping the XML
tree structure into strings and use efficient string algorithms
for query processing. Since the size of each string grows with the
length of the encoded path, many approaches use some sort of
compression to offset this overhead. Examples of those are Index
Fabric \cite{CSF+01}, tree signatures \cite{ADR+04}, materialized
schema paths \cite{BW03}, PathGuides \cite{CYWY03}, and tree
sequencing (ViST \cite{WPFY03}, PRIX \cite{RM04}, and NoK
\cite{ZKO04}). These encodings can be used in conjunction with
structural summaries to improve query evaluation performance. In
fact, the availability of summaries can be of great assistance to
an XML optimizer \cite{AttilaVLDB:2005}.

\dx\ uses an interval encoding derived from \cite{SK85} in which
each element in the collections is represented by its start and
end positions (the character offset from the beginning of the
document they belong to).

\section{Answering XML queries using views}

Another area closely related to summarization is answering queries
using views. As in traditional database systems, the performance
of XML queries can be improved by rewriting them using caching and
materialized views containing information relevant to the
computation of the query. A recent contribution in this area
includes a framework for XPath view materialization and query
containment \cite{BOB+04} that uses value and structure indexes on
views. Another framework was proposed in \cite{MS05} for
maintaining a semantic cache of XPath query results as
materialized views used to speed-up query processing. Other work
has considered the problem of deciding the existence of a query
rewriting and finding a minimal rewriting using XPath views
\cite{XO05}, and computing maximal contained rewriting for tree
pattern queries (a core subset of XPath) \cite{LWZ06}.

For XQuery, query rewriting poses additional challenges. One of
them is that queries may be nested. Another challenge comes from
the mix of list, bag and set semantics supported by XQuery, which
makes testing equivalence more difficult. In this context, there
has been some work on query rewriting for nested XQuery queries
using nested XQuery views \cite{ODPC06}. A recent contribution for
extended tree patterns views (a subset of XQuery) proposes
containment and equivalent rewriting strategies based on a
dataguide enhanced with integrity constraints \cite{ABMP07}. This
proposal considers only queries described by tree patterns.

We must point out here that most of the work in this area could be
applied to our framework to expand the query evaluation techniques
we present in Chapter \ref{sec:changingSDswithxpath}.

\section{Validating summaries}

DTDs\footnote{\small \tt{ http://www.w3.org/TR/REC-xml}} and XML
Schemas\footnote{\small \tt{ http://www.w3.org/TR/xmlschema-1/}}
are proposals used for validation and verification of XML
documents. A DTD is a context-free grammar and an XML Schema is a
typed definition language. Both are schemas in the database sense,
and thus describe classes of documents and constrain their
structure. However, they provide only a limited description of the
instances that satisfy them and no mechanism to locate specific
instance fragments. In contrast, summaries are constructed for a
particular instance and consequently provide a tighter description
of the data. They also contain the necessary information for
locating the instance fragments they describe. DTDs and XML
Schemas can be used to constrain the construction of summaries but
they are no substitute for them. Moreover, summaries can be
constructed even when DTDs and XML Schemas are not present.

\begin{figure}
    \centering
        \includegraphics{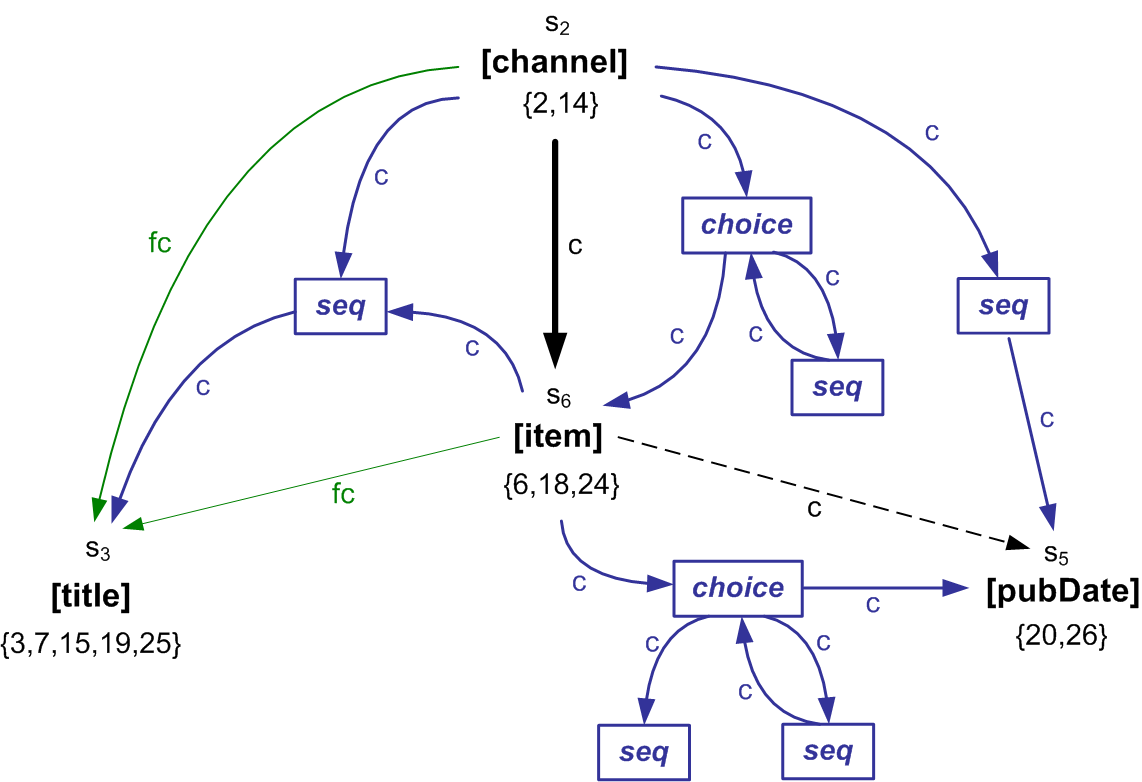}
    \caption{Label SD fragment with XML graph schema model annotations}
    \label{fig:RSS-label-relax}
\end{figure}

In addition to describing an instance, \dx\ summaries could
potentially be used for prescribing or constraining the data by
adding schema constructs. Figure \ref{fig:RSS-label-relax} shows a
fragment of the label SD from Figure \ref{fig:smallRSS-combined}
(a) annotated with XML graph schema constructs \cite{MS07} in
blue. These constructs, which contain choice and sequence nodes
(and others not shown in the figure), are able to express XML
schema languages like DTDs, XML Schemas, and Relax
NG\footnote{\small \tt{
http://www.oasis-open.org/committees/relax-ng/spec-20011203.html/}}.
(For a survey on XML schema languages see \cite{MLMK05}.) The SD
of Figure \ref{fig:RSS-label-relax} represents channels that
contain exactly one title, one description and one or more items
that contain themselves one title and a sequence of zero or more
description elements. In the figure, choice and sequence nodes are
used to represent the number of occurrences of an element, which
can be zero, one, or unbounded. The DTD corresponding to the
elements that appear in Figure \ref{fig:RSS-label-relax} is the
following:

{\small \texttt{<!ELEMENT channel (title, item+, description)>}}

{\small \texttt{<!ELEMENT item (title, description?)>}}

{\small \texttt{<!ELEMENT title (\#PCDATA)>}}

{\small \texttt{<!ELEMENT pubDate (\#PCDATA)>}}

In an SD, schema annotations have to be consistent with instance
descriptions. For example, the existential edge $(s_6, s_5)$ is
compatible with an schema permitting any number of occurrences of
$s_5$ (even zero). In contrast, the same edge is incompatible with
an schema requiring at least one $s_5$ element because the dashed
edge allows some items to have no descriptions.

This is just an example of how schema constructs can be integrated
with \dx\ summaries. There are many other ways of approaching the
subject, but we do not consider it further in this thesis.

\vspace{.2in}

This chapter provided a discussion of related work on structural
summaries and four other areas close to our work: path summaries
for OO data, hierarchical encodings, answering XML queries using
views and validating summaries. In the next chapter, we begin
introducing one of the major contributions of this thesis, the
\dx\ framework. We will show how \dx\ generalizes and extends both
structural and path summaries, and how \dx\ summaries can be used
in query processing.

\chapter{AxPRE summaries}\label{sec:axpres}

This chapter provides an overview of the \dx\ framework. The
framework includes a powerful language based on \emph{axis path
regular expressions} (AxPREs) for describing each set in a
partition of instance nodes (extents). AxPREs provide the
flexibility necessary for declaratively specifying the mapping
between instance nodes and summary nodes for a given collection.
These AxPRE mappings are capable of expressing order and
cardinality, among other properties. AxPREs are evaluated on a
graph (called \emph{axis graph}) in which nodes are XML elements
and edges are binary relations between them. Hence, AxPREs can be
viewed as path regular expressions on binary relations. These
relations include all XPath axes and additional ones that can be
expressed in XPath.

Extents are defined using a novel approach: selective bisimilarity
applied to subgraphs described by AxPREs (i.e., \emph{AxPRE
neighbourhoods}). This particular use of bisimulation supports the
definition of summaries that go beyond the traditional parent and
child hierarchical relationships covered by the abundant
literature on summaries. Intuitively, nodes that have bisimilar
subgraphs ``around'' them (i.e., neighbourhoods) belong to the
same extent. For instance, \dx\ can define extents containing only
nodes with the same set of outgoing label paths matching a given
sequence of axes. Neighbourhoods are a key mechanism in the
declarative definition of \dx\ summaries.

\section{A regular expression language on axes}

In this section, we introduce the AxPRE language for describing
neighbourhoods in an SD. For representing an XML instance, \dx\
uses a labeled graph model called an \emph{axis graph}.

\begin{definition}[Axis Graph]
An axis graph $\mathcal{A} = (\mathit{Inst},$ $Axes,$ $Label,$
$\lambda)$ is a structure where $\mathit{Inst}$ is a set of nodes,
$Axes$ is a set of binary relations $\{ E_1^\mathcal{A}, \ldots,
E_n^\mathcal{A} \}$ in $\mathit{Inst} \times \mathit{Inst}$ and
their inverses, $Label$ is a finite set of node names, and
$\lambda$ is a function that assigns labels in $Label$ to nodes in
$\mathit{Inst}$. Edges are labeled by axis names. \qed
\end{definition}

An axis graph is an abstract representation of the XPath data
model \cite{W3C:XPath/02} extended with edges that represent XPath
binary relations between elements. It can also include additional
axes, such as $\mathit{fc}$ (where $\mathit{fc} := child::*[1]$),
$\mathit{ns}$ (where $\mathit{ns} :=
following\mbox{-}sibling::*[1]$), \emph{id-idrefs} or any binary
relation that can be expressed in XPath. When representing an XML
instance, axis graph nodes are labeled by element or attribute
names (including namespaces).

\begin{figure}
    \centering
        \includegraphics{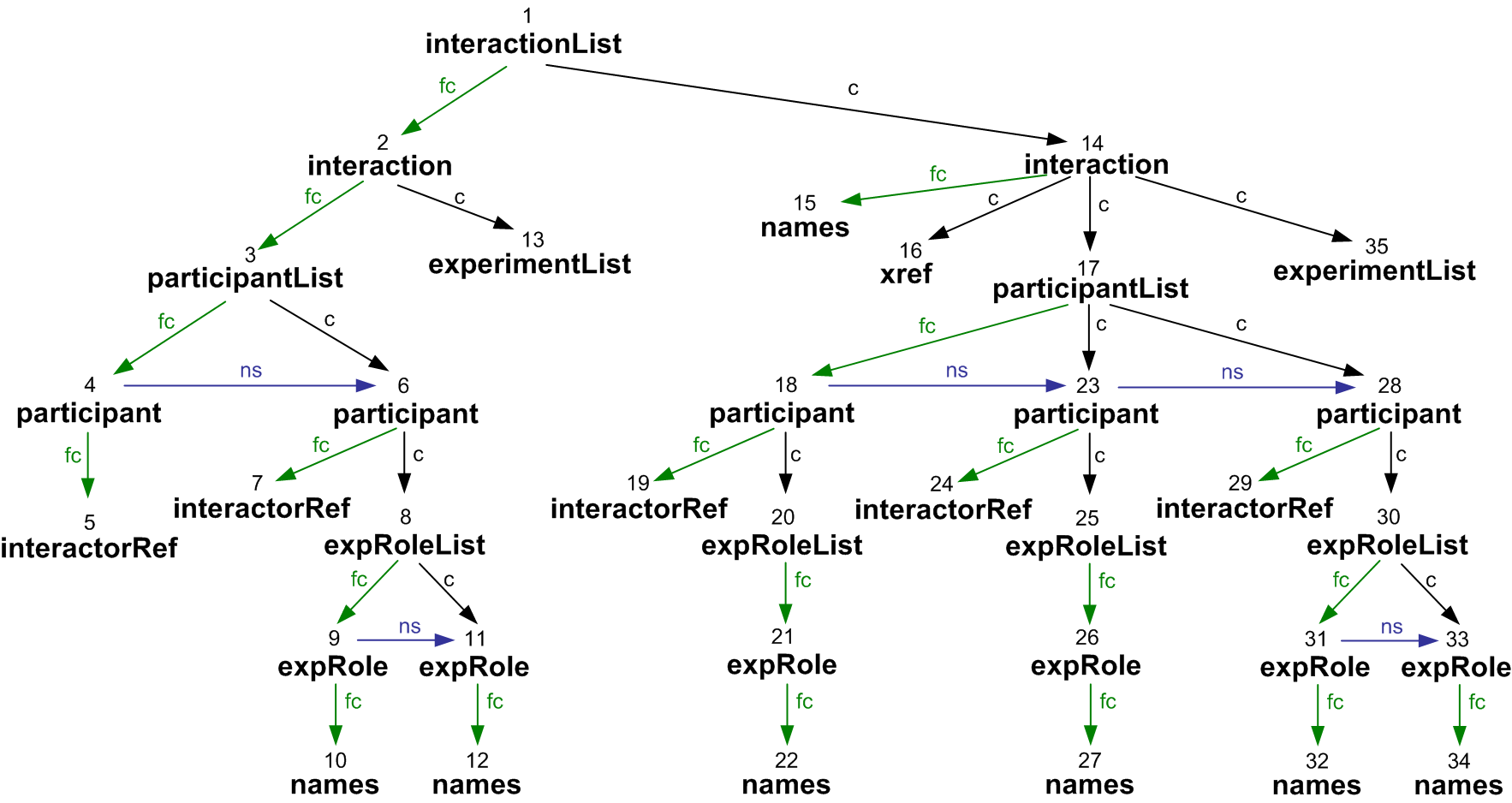}
    \caption{The axis graph of two PSI-MI interactions}
    \label{fig:PSIMI}
\end{figure}

\begin{example}[PSI-MI Axis Graph]
Figure \ref{fig:PSIMI} shows an axis graph for our running
example, which is a sample of a protein-protein interaction (PPI)
dataset in PSI-MI\footnote{\small
\tt{http://psidev.sourceforge.net/mi/xml/doc/user/} } format.
PSI-MI stands for Proteomics Standards Initiative
Molecular-Interaction and is the de-facto model for PPI used by
many molecular interaction databases such as
BioGRID\footnote{\small \tt{http://www.thebiogrid.org/} }, Human
Protein Reference Database (HPRD)\footnote{\small
\tt{http://www.hprd.org/} }, and IntAct\footnote{\small
\tt{http://www.ebi.ac.uk/intact/} }. The PSI-MI XML schema has a
large number of optional elements to allow flexibility, with the
result that PSI-MI data can be very heterogeneous. Since different
databases use different fragments of the schema, finding common
instance patterns and understanding schema usage can be
challenging \cite{SCKT07}.

Each interaction consists of an experimentList element with all
the experiments in which the interaction has been determined, a
participantList element with the molecules that participate in the
interaction and some optional elements like the name of the
interaction and a reference (xref) to an interaction database.
Each participantList contains two or more participants, which are
the molecules participating in the interaction. A participant
element contains a description of the molecule, either by
reference to an element of the interactorList, or directly in an
interactor element. In addition, each participant contains a list
of all the roles it plays in the experiments (e.g., bait, prey,
neutral, etc.)

Note that, for the sake of clarity, we have omitted many edges
depicting relations that actually exist. For example, the
$\mathit{fc}$ ($firstchild$) relation is included in the $c$
($child$) relation, so any $\mathit{fc}$ edge is also a $c$ edge.
The inverses of each relation are not shown in the figure, e.g.,
for each $c$ relation, a $p$ ($parent$) relation exists (since
$p=c^{-1}$). \qed
\end{example}

An AxPRE gives a declarative description of a partition of
elements in an SD, something not provided by any other proposal in
the literature.

In an axis graph we define paths and label paths as usual. We call
a path defined on edges an \emph{axis path}, and the string
resulting from the concatenation of its labels is an \emph{axis
label path}.

\begin{definition} [Axis Path and Axis Label Path] \label{def:path}
Let $\mathcal{N}$ be a connected subgraph of an axis graph
$\mathcal{A}$, and $v,v_n$ be two nodes in $\mathcal{N}$ such that
there is a path $p=(v, axis_1$, $v_1, axis_2$, $\ldots, axis_n,
v_n)$ from $v$ to $v_n$. The \emph{axis path} of $p$ is the string
$ap = axis_1$.$axis_2. \dots$ .$axis_n$. The \emph{axis label
path} of $p$ is the string $\lambda(p) =
axis_1[\lambda(v_1)].axis_2[\lambda(v_2)]. \ldots.$ $axis_n$
$[\lambda(v_n)]$. \qed
\end{definition}

\begin{example}
Consider the axis graph of Figure \ref{fig:PSIMI}. Two of the
paths from node $6$ to $11$ are $p=(6, c, 8, \mathit{fc}, 9, ns,
11)$ and $p'=(6, c, 8, c, 11)$. Their axis paths are
$ap=c.\mathit{fc}.ns$ and $ap'=c.c$, respectively. Finally, the
axis label paths of $p$ and $p'$ are $\lambda(p)=$
$c[expRoleList].$ $\mathit{fc}[expRole].$ $ns[expRole]$ and
$\lambda(p')=$ $c[expRoleList].$ $c[expRole]$, respectively. \qed
\end{example}

\begin{definition}[Axis Path Regular Expressions] \label{def:axpre}
An \emph{axis path regular expression} (AxPRE) is an expression
generated by the grammar
\[ E \longleftarrow axis \mid axis[B(l)]
    \mid (E \mid E)  \mid (E)^* \mid E.E \mid \epsilon \mid [B(l)] \]
where $axis \in Axes$ and $\epsilon$ is the symbol representing
the empty expression. \qed
\end{definition}

Definition \ref{def:axpre} describes the syntax of path regular
expressions on the binary relations (labeled edges) of the axis
graph including node label tests. The function $B(l)$ is a boolean
function on a label $l \in Label$ that supports elaborate tests
beyond just matching labels.

An AxPRE defines a pattern we want to find in an instance. We need
a way of computing all occurrences of such pattern in an axis
graph -- each occurrence will be called a neighbourhood. We do
this by computing an automaton for the AxPRE, another for the axis
graph, and then taking the intersection. Finally, a summary will
group nodes with similar patterns together into an extent (\dx\
uses bisimulation as the notion of similarity).

The AxPRE semantics (Definition \ref{def:axpreSemantics}) is given
by the notion of \emph{AxPRE neighbourhood} of a node (Definition
\ref{def:neighbourhood}). In order to compute an AxPRE
neighbourhood we need first to define an automaton from the axis
graph. Such an automaton will have two states for each node in the
axis graph, one named $head$ and the other $tail$. In addition,
edges in the graph will be represented as transitions between
$tail$ and $head$ states, and node labels as transitions between
$head$ and $tail$ states.

\begin{figure}
    \centering
        \includegraphics{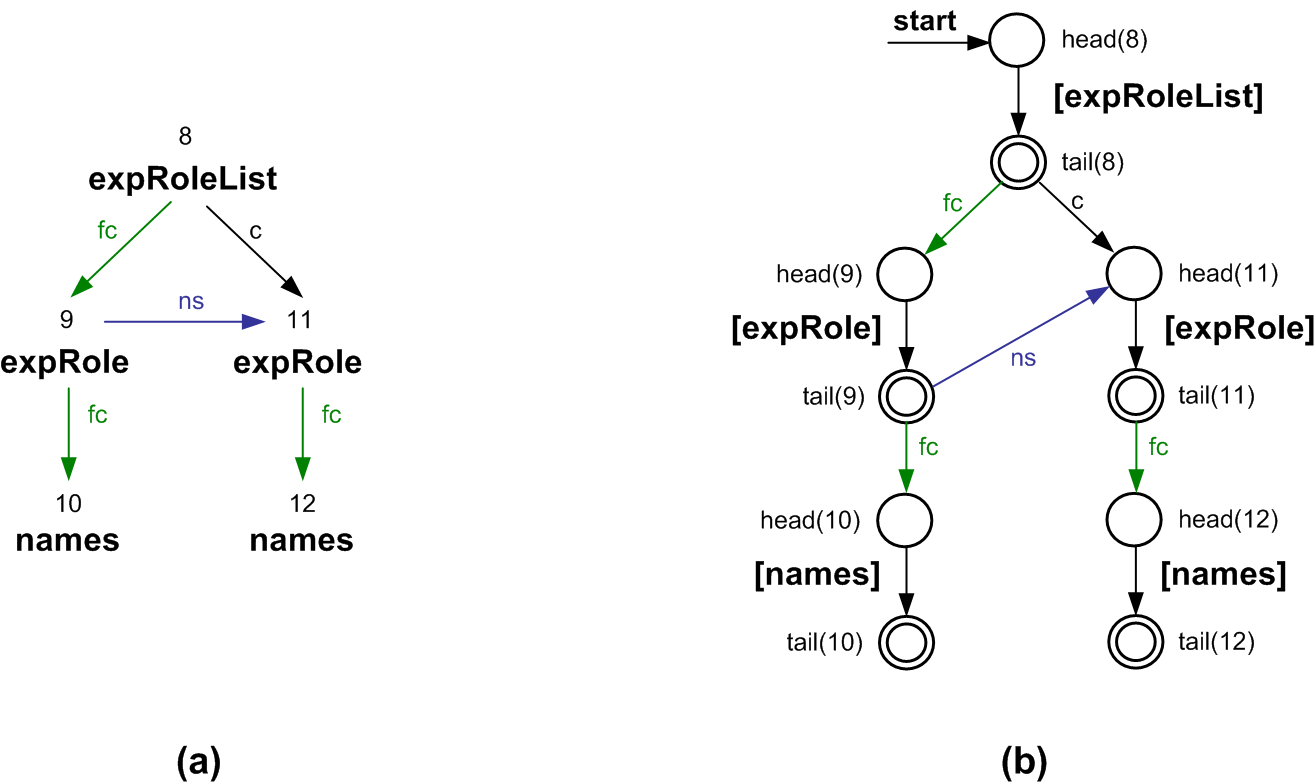}
    \caption{Axis graph fragment from node 8 (a) and its automaton $\mathcal{M}_\mathcal{A}(8)$ (b)}
    \label{fig:axisgraphauto}
\end{figure}

\begin{definition}[Axis Graph Automaton] \label{def:axisAutomaton}
Let $\mathcal{A} = (\mathit{Inst},$ $Axes,$ $Label,$ $\lambda)$ be
an axis graph and $v$ a node in $\mathcal{A}$. The \emph{axis
graph automaton of} $\mathcal{A}$ \emph{from} $v$,
$\mathcal{M}_\mathcal{A}(v)=\{Q, \Sigma, \delta, q_0, F \}$, is an
automaton \cite{book:HopcroftU79} defined as follows:
\begin{itemize}
\item For each node $w \in Inst$ there is a state $head(w) \in Q$,
a state $tail(w) \in Q$ and a transition $\delta(head(w),
[\lambda(w)])=tail(w)$; \item For each edge $(w_i,w_j)$ labeled
$axis$ in $\mathcal{A}$ there is a transition
$\delta(tail(w_i),axis)=head(w_j)$; \item All $tail(w)$ states in
$Q$, $w \in Inst$, are final states in $F$, and $head(v)$ is the
initial state $q_0$.
\end{itemize}
\qed
\end{definition}

\begin{example}
Consider node $8$ of our running example. Figure
\ref{fig:axisgraphauto} shows on the left hand side a fragment of
the axis graph that contains node 8. The axis graph automaton from
node $8$ (on the right hand side of the figure) has $head(8)$ as
initial state and all $tail$ states as final. Each node in the
axis graph fragment is unfolded into a head and a tail states in
the automaton and its label is represented by a transition between
them. Consider node $11$ with label $expRole$ that has $ns$ and
$c$ incoming edges and a $\mathit{fc}$ outgoing edge in the axis
graph. In the automaton, $11$ is represented by a $head(11)$ state
that has $ns$ and $c$ incoming transitions and an outgoing
transition $[expRole]$ to $tail(11)$. The outgoing $\mathit{fc}$
edge is translated into a $\mathit{fc}$ transition from $tail(11)$
to the head state of the corresponding node, which is $12$. \qed
\end{example}

An automaton can be obtained from an AxPRE following the usual
Thompson's construction for regular expressions with a minor
change to the basis steps to account for AxPRE semantics (which
require accepting all prefixes of the language). The language
accepted by the so called \emph{AxPRE automaton} thus constructed
will always be prefix-closed. (A language $L$ is said to be
prefix-closed if, given any word $l \in L$, all prefixes of $l$
are also in $L$ \cite{book:HopcroftU79}.)

\begin{figure}
    \centering
        \includegraphics{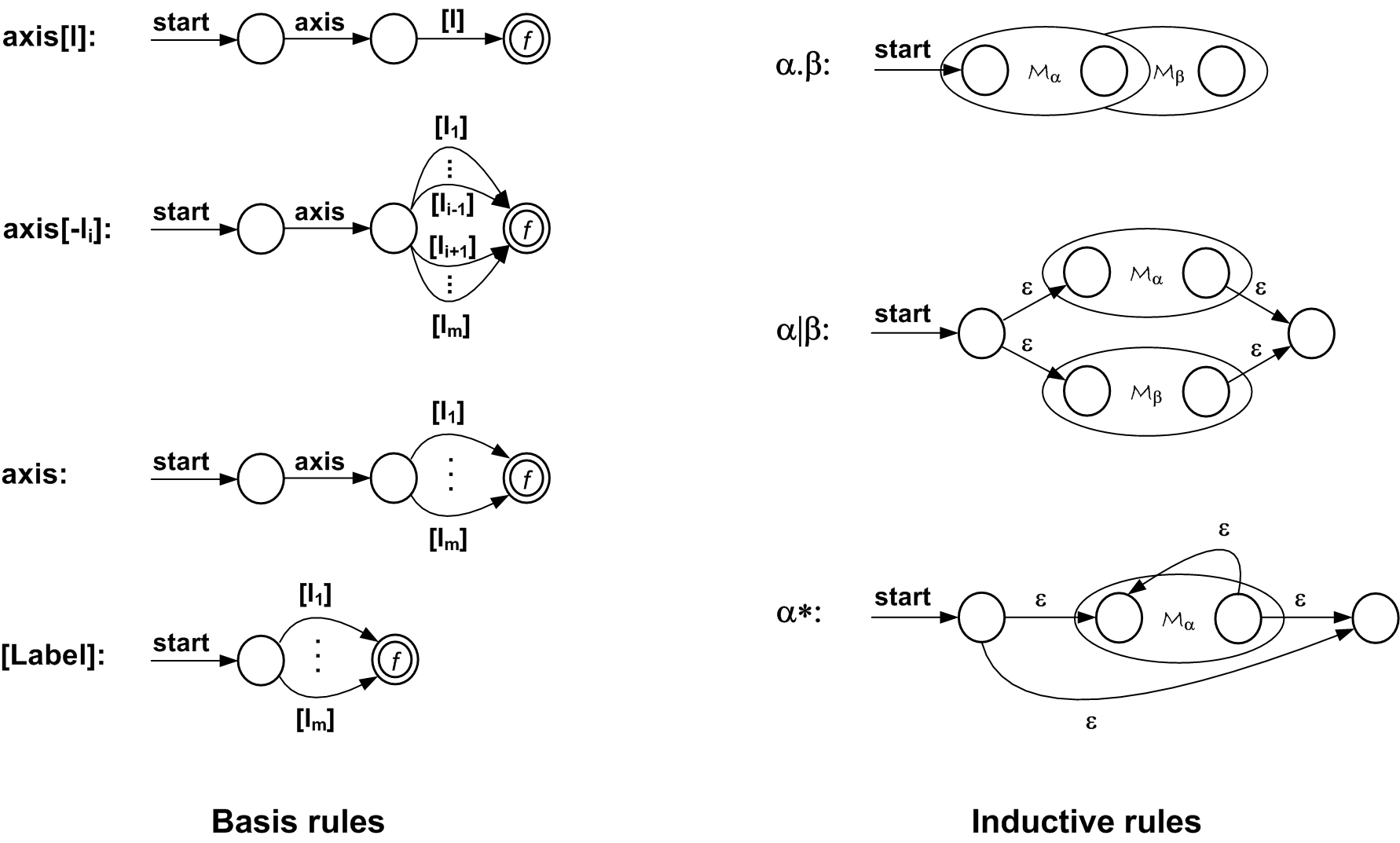}
    \caption{Rules of the modified Thompson's construction}
    \label{fig:thompson}
\end{figure}

\begin{definition}[AxPRE Automaton] \label{def:AxPREAutomaton}
Let $\alpha$ be an AxPRE. The \emph{AxPRE automaton} of $\alpha$
is an automaton $\mathcal{M}_\alpha$ obtained from $\alpha$ with a
modified Thompson's construction \cite{book:HopcroftU79} for
accepting all prefixes (Figure \ref{fig:thompson}), in which only
the final states of the basis rules are kept as final in the
resulting automaton (the inductive rules for concatenation,
disjunction and Kleene closure do not mark any additional state as
final). The transition function $\hat{\delta}(q_\alpha,axis)$
returns the states that can be reached by an $axis$ transition
after following an arbitrary number (possibly zero) of $\epsilon$
transitions. \qed
\end{definition}

\begin{figure}
    \centering
        \includegraphics{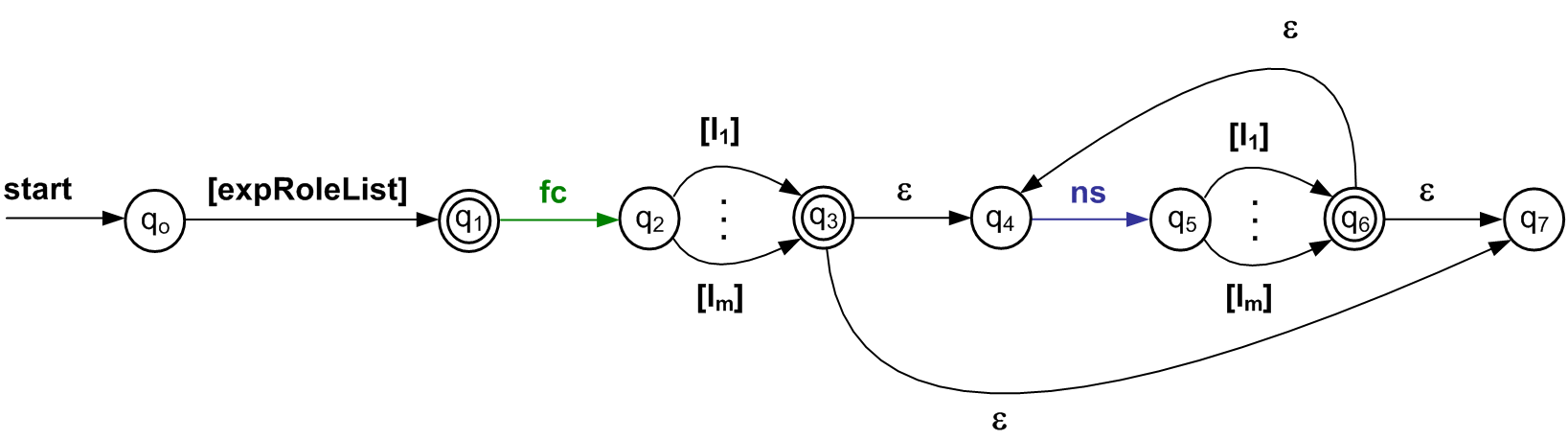}
    \caption{AxPRE automaton $\mathcal{M}_{[expRoleList].\mathit{fc}.ns^*}$}
    \label{fig:AxpreAuto}
\end{figure}

\begin{example}
Consider the AxPRE $[expRoleList].\mathit{fc}.ns^*$ and its
automaton in Figure \ref{fig:AxpreAuto}. The application of rule
$axis[l]$ of the modified Thompson's construction creates states
$q_0$, $q_1$ and the $[expRoleList]$ transition between them. The
application of rule $axis$ creates $q_2$, $q_3$, $q_5$, $q_6$, and
the $[l_1], \ldots, [l_m]$ transitions from $q_2$ to $q_3$ and
from $q_5$ to $q_6$ (there is one transition $[l_i]$ for each
string in $Label$). The final automaton is obtained by applying
the concatenation and Kleene closure rules. \qed
\end{example}

An automaton for the intersection of two languages can be
constructed by taking the product of the automata for the two
languages \cite{MW95,Yann90}.

\begin{definition}[Intersection Automaton] \label{def:intersection}
Let $\mathcal{M}_\mathcal{A}(v)$ be the automaton of an axis graph
$\mathcal{A}$ from a node $v$, and $\mathcal{M}_\alpha$ be the
automaton of an AxPRE $\alpha$. The \emph{intersection automaton}
$\mathcal{M}_\mathcal{A}(v) \cap \mathcal{M}_\alpha$ is an
automaton in which states are pairs $(q_\mathcal{A},$ $q_\alpha)$
consisting of a state $q_\mathcal{A} \in
\mathcal{M}_\mathcal{A}(v)$ and a state $q_\alpha \in
\mathcal{M}_\alpha$, and there is a transition
$\delta((q_\mathcal{A},$ $q_\alpha),\mathcal{X})=(q'_\mathcal{A},$
$ q'_\alpha )$ if there are transitions
$\delta(q_\mathcal{A},\mathcal{X})=q'_\mathcal{A}$ in
$\mathcal{M}_\mathcal{A}(v)$ and
$\hat{\delta}(q_\alpha,\mathcal{X})=q'_\alpha$ in
$\mathcal{M}_\alpha$, where $\mathcal{X}$ is either an axis or a
label. A state $\langle q_\mathcal{A}, q_\alpha \rangle$ is final
(initial) if both $q_\mathcal{A}$ and $q_\alpha$ are final
(initial). \qed
\end{definition}

\begin{figure}
    \centering
        \includegraphics{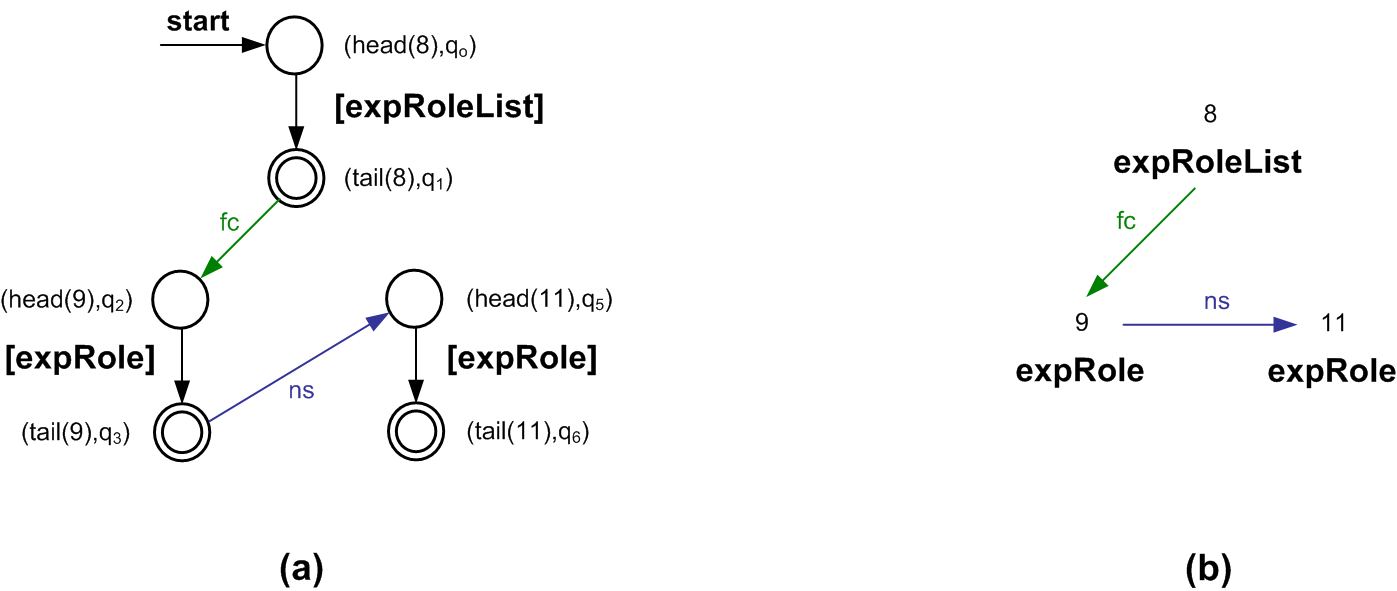}
    \caption{Intersection automaton $\mathcal{M}_\mathcal{A}(8) \cap \mathcal{M}_{[expRoleList].\mathit{fc}.ns^*}$ (a) and resulting AxPRE neighbourhood $\mathcal{N}_{[expRoleList].\mathit{fc}.ns^*}(8)$ (b)}
    \label{fig:intersection}
\end{figure}

The machinery introduced in Definitions \ref{def:axisAutomaton}
through \ref{def:intersection} is required for computing AxPRE
neighbourhoods of nodes in the axis graph. The neighbourhood of a
node $v$ by $\alpha$ can be obtained by taking the intersection
between the axis graph automaton from $v$ and the AxPRE automaton
of $\alpha$, and then converting the resulting automaton to an
axis graph fragment as described in Definition
\ref{def:neighbourhood}.

\begin{definition}[AxPRE Neighbourhood of a Node] \label{def:neighbourhood}
Let $\mathcal{A} = (\mathit{Inst},$ $Axes,$ $Label,$ $\lambda)$ be
an axis graph, $v$ a node in $\mathcal{A}$, $\alpha$ an AxPRE, and
$\mathcal{M}_\mathcal{A}(v) \cap \mathcal{M}_\alpha$ the
intersection automaton of $\mathcal{M}_\mathcal{A}(v)$ and
$\mathcal{M}_\alpha$. The \emph{AxPRE neighbourhood} of $v$ by
$\alpha,$ denoted $\mathcal{N}_{\alpha}(v),$ is the subgraph of
$\mathcal{A}$ defined as follows:
\begin{itemize}
\item For each transition
$\delta((head(w),q_\alpha),l)=(tail(w),q'_\alpha)$, where
$(tail(w),q'_\alpha)$ is a final state, there is a node $w$ with
label $l$ in $\mathcal{A}$;  \item For each transition
$\delta((tail(w_i),q_\alpha),axis)=(head(w_j),q'_\alpha)$, where
$(tail(w_i),q_\alpha)$ is a final state, there is an edge
$(w_i,w_j)$ labeled $axis$ in $\mathcal{A}$.
\end{itemize}
\qed
\end{definition}

\begin{example}[AxPRE Neighbourhood of a Node]
Consider node 8 of our running example.  The intersection
automaton $\mathcal{M}_\mathcal{A}(8)$ $\cap$
$\mathcal{M}_{[expRoleList].\mathit{fc}.ns^*}$ is depicted in
Figure \ref{fig:intersection} (a).  States are labeled by pairs
$(q_\mathcal{A}, q_\alpha)$, where $q_\mathcal{A}$ is a state in
automaton $\mathcal{M}_\mathcal{A}(8)$ and $q_\alpha$ is a state
in automaton $\mathcal{M}_{[expRoleList].\mathit{fc}.ns^*}$. The
intersection has been computed following Definition
\ref{def:intersection}. The figure shows only the states that have
some incoming or outgoing transition. Note that transition $c$
between $tail(8)$ and $head(11)$ is not part of the intersection
because $\mathit{fc}$ is the only outgoing transition from $q_1$
in $q_\alpha$.

Figure \ref{fig:intersection} (b) shows the AxPRE neighbourhood of
node 6, $\mathcal{N}_{[participant]c.\mathit{fc}.ns^*}(6)$,
obtained by converting the intersection automaton to an axis graph
fragment as described in Definition \ref{def:neighbourhood}. Note
that transitions from $(head(v),\ldots)$ to $(tail(v),\ldots)$ in
the intersection are node labels in the AxPRE neighbourhood and
that transitions from $(tail(v),\ldots)$ to $(head(w),\ldots)$ are
edge labels (axes) in the neighbourhood.

Consider now the five $[participant].c.\mathit{fc}.ns^*$
neighbourhoods depicted in Figure \ref{fig:5-neighbourhoods}.
Neighbourhood (a) matches a prefix of the AxPRE
($[participant].c$) whereas (b) through (e) match the entire AxPRE
but with a different number of iterations in the Kleene closure
for $ns$: 1 for (b) and (e), and 0 for (c) and (d). \qed
\end{example}

\begin{figure}
    \centering
        \includegraphics{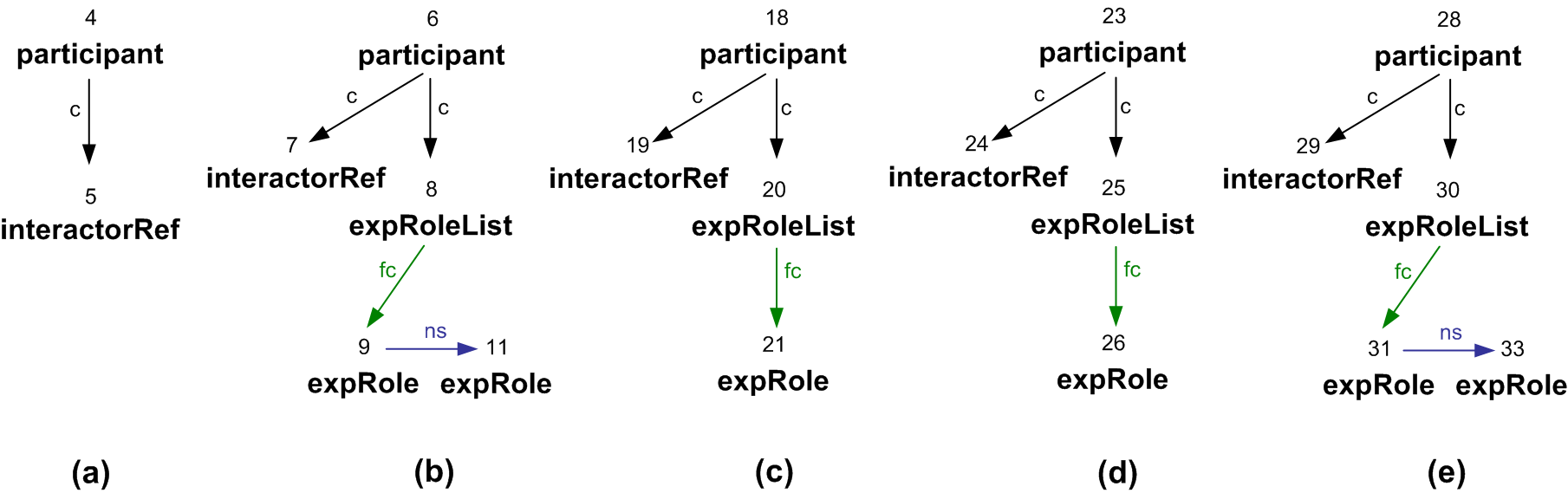}
    \caption{All $[participant].c.\mathit{fc}.ns^*$ neighbourhoods}
    \label{fig:5-neighbourhoods}
\end{figure}

We formalize next the notion of AxPRE semantics based on AxPRE
neighbourhoods.

\begin{definition} [AxPRE Semantics] \label{def:axpreSemantics}
Let $\mathcal{A} = (\mathit{Inst},$ $Axes,$ $Label,$ $\lambda)$ be
an axis graph and $v$ a node in $\mathcal{A}$. The evaluation of
an AxPRE $\alpha$ on $v$ returns the AxPRE neighbourhood of $v$ by
$\alpha$. \qed
\end{definition}

\section{Neighbourhoods and bisimulation}

AxPRE neighbourhoods allow us to define a notion of similarity
between nodes in an axis graph. The idea underlying \dx\ is that
nodes with similar AxPRE neighbourhoods will be grouped together.
In particular, \dx\ uses the familiar concept of \emph{labeled
bisimulation} applied to AxPRE neighbourhoods, formalized by
Definition \ref{def:bisimulation}.

\begin{definition}[Labeled Bisimulation and Bisimilarity] \label{def:bisimulation}
Let $\mathcal{N}_{\alpha}(v_0)$ and $\mathcal{N}_{\beta}(w_0)$ be
two AxPRE neighbourhoods of an axis graph $\mathcal{A} =
(\mathit{Inst},$ $Axes,$ $Label,$ $\lambda)$, such that
$Axes_{\alpha} \subseteq Axes$ and $Axes_{\beta} \subseteq Axes$.
A \emph{labeled bisimulation} between $\mathcal{N}_{\alpha}(v_0)$
and $\mathcal{N}_{\beta}(w_0)$ is a symmetric relation $\approx$
such that for all $v \in \mathcal{N}_{\alpha}(v_0)$, $w \in
\mathcal{N}_{\beta}(w_0)$, $E_i^{\alpha} \in Axes_{\alpha}$, and
$E_i^{\beta} \in Axes_{\beta}$: if $v \approx w,$ then
$\lambda(v)=\lambda(w)$; if $v \approx w,$ and $\langle v, v'
\rangle \in E_i^{\alpha}$, then $\langle w, w' \rangle \in
E_i^{\beta}$ and $v' \approx w'$. Two nodes $v \in
\mathcal{N}_{\alpha}(v_0)$, $w \in \mathcal{N}_{\beta}(w_0)$ are
\emph{bisimilar}, in notation $v \sim w$, iff there exist a
labeled bisimulation $\approx$ between $\mathcal{N}_{\alpha}(v_0)$
and $\mathcal{N}_{\beta}(w_0)$ such that $v \approx w$. Similarly,
two neighbourhoods $\mathcal{N}_{\alpha}(v_0)$ and
$\mathcal{N}_{\beta}(w_0)$ are \emph{bisimilar}, in notation
$\mathcal{N}_{\alpha}(v_0) \sim \mathcal{N}_{\beta}(w_0)$, iff
$v_0 \sim w_0$. \qed
\end{definition}

Definition \ref{def:bisimulation} captures outgoing label paths
from the nodes. Bisimulation provides a way of computing a double
homomorphism between graphs. The widespread use of bisimulation in
summaries is motivated by its relatively low computational
complexity properties. The bisimulation contraction of a labelled
graph can be done in time $O(m \log n)$ (where $m$ is the number
of edges and $n$ is the number of nodes in a labelled graph) as
shown in \cite{PT87}, or even linearly for acyclic graphs, as
shown in \cite{DPP04}. Using bisimulation also allows us to
capture all the existing bisimulation-based proposals in the
literature (Chapter \ref{sec:lattice}).

\begin{example} \label{exa:bisimulation}
Let us consider the nodes $6$ and $18$ in the axis graph of Figure
\ref{fig:PSIMI}. Their $[participant].c.\mathit{fc}.ns^*$
neighbourhoods are depicted in Figure \ref{fig:5-neighbourhoods}
(b) and (c), respectively. Based on Definition
\ref{def:bisimulation}, we can define a labeled bisimulation
$\approx$ between nodes $7$ and $19$ because they have the same
labels and they do not have outgoing edges. For the same reasons
we have $11 \approx 21$. However, it is not possible to define a
labeled bisimulation between $9$ and $21$ because, even though
they have the same labels, $9$ has one outgoing edge whereas $21$
does not. Thus, $9 \not \approx 21$. This prevents us from
defining a label bisimulation between $8$ and $20$ because they
each have only one outgoing $\mathit{fc}$ edge, but to nodes $9$
and $20$, which are not bisimilar. Therefore,  $8 \not \approx
20$. Similarly, $6 \not \approx 18$ because they have edges with
the same labels ($c$) to nodes that are not bisimilar ($8$ and
$20$). Consequently, neighbourhoods (b) and (c) of Figure
\ref{fig:5-neighbourhoods} are \emph{not} bisimilar.

In contrast, let us compare now nodes $6$ and $18$ but with
respect to their $[participant].c^*$ neighbourhoods, which are
depicted in Figure \ref{fig:other-5-neighbourhoods} (b) and (c),
respectively. In this case we can have $9 \approx 21$ and $11
\approx 21$ because all of them are leaves and have the same
label. Therefore, $8 \approx 20$ because the outgoing edges from
$8$ go to nodes $9$ and $11$, which are bisimilar to the target
node ($21$) of the only outgoing edge from $20$. Thus, $6 \approx
18$ because they have edges with the same labels ($c$) to nodes
that in this case are bisimilar ($7 \approx 19$ and $8 \approx
20$). Consequently, neighbourhoods (b) and (c) of Figure
\ref{fig:other-5-neighbourhoods} are in fact bisimilar. \qed
\end{example}

\begin{definition} [AxPRE Bisimilarity] \label{def:axprebisimilarity}
Let $\mathcal{A} = (\mathit{Inst},$ $Axes,$ $Label,$ $\lambda)$.
When two nodes $v_0$ and $w_0$ in $\mathcal{A}$ have bisimilar
neighbourhoods by the same AxPRE $\alpha$, that is
$\mathcal{N}_{\alpha}(v_0) \sim \mathcal{N}_{\alpha}(w_0)$, we say
that $v_0$ and $w_0$ are \emph{AxPRE bisimilar} by $\alpha$ or
$\alpha$\emph{-bisimilar}, in notation $v_0  \sim^{\alpha} w_0$.
\qed
\end{definition}

\begin{example} \label{exa:axprebisimilarity}
Consider again the neighbourhoods in Figure
\ref{fig:5-neighbourhoods}. Nodes $6$ and $18$ have non-bisimilar
$[participant].c.\mathit{fc}.ns^*$ neighbourhoods and thus $6 \not
\sim^{\alpha} 18$, where AxPRE
$\alpha=[participant].c.\mathit{fc}.ns^*$. However, if we consider
now their $[participant].c^*$ neighbourhoods, which are bisimilar,
then $6 \sim^{\alpha'} 18$ for AxPRE $\alpha'= [participant].c^*$.
\qed
\end{example}

\begin{figure*}
    \centering
        \includegraphics{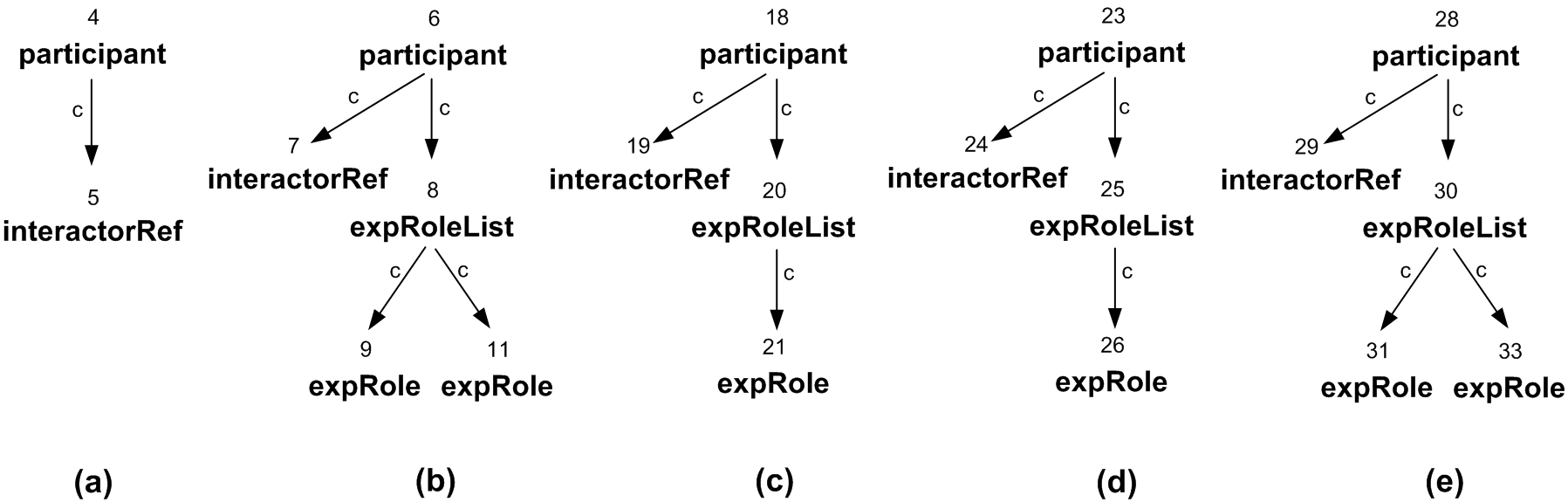}
    \caption{All $[participant].c^*$ neighbourhoods}
    \label{fig:other-5-neighbourhoods}
\end{figure*}

AxPRE bisimilarity is used for defining partitions of an axis
graph. Intuitively, a so called \emph{AxPRE partition} assigns two
nodes $v$ and $w$ in an axis graph to the same class if their
AxPRE neighbourhoods by a given $\alpha$ are bisimilar. This is
formalized by Definition \ref{def:AxPREPartition}.

\begin{definition} [AxPRE Partition] \label{def:AxPREPartition}
Let $\mathcal{A} = (\mathit{Inst},$ $Axes,$ $Label,$ $\lambda)$ be
an axis graph and $\alpha$ an AxPRE. An \emph{AxPRE partition} of
$\mathit{Inst}$ by $\alpha$, denoted $\mathcal{P}_\alpha$, is a
set of pairwise disjoint subsets of $\mathit{Inst}$ whose union is
$\mathit{Inst}$ defined as follows: two nodes $v,w \in
\mathit{Inst}$ belong to the same set $P_\alpha^i \in
\mathcal{P}_\alpha$ iff $v  \sim^{\alpha} w$.
\end{definition}

\begin{definition} [Positive Classes] \label{def:positiveClasses}
Let $\mathcal{A} = (\mathit{Inst},$ $Axes,$ $Label,$ $\lambda)$ be
an axis graph, $\alpha$ an AxPRE and $P_\alpha^\emptyset = \{ v
\in \mathit{Inst} \; | \; \mathcal{N}_{\alpha}(v) = \emptyset \}$
the set of the empty neighbourhoods in the AxPRE partition of
$\mathit{Inst}$ by $\alpha$. Then, $\mathcal{P}^+_\alpha =
\mathcal{P}_\alpha - P_\alpha^\emptyset$ is the set of
\emph{positive classes} of $\mathcal{P}_\alpha$. \qed
\end{definition}

Since all nodes that have an empty AxPRE neighbourhood belong to
the same equivalence class, $\mathcal{P}_\alpha$ and
$\mathcal{P}^+_\alpha$ differ in at most one set.

\begin{example}
Consider the AxPRE partitions by $[l_1], \ldots, [l_n]$, where
$l_1, \ldots, l_n$ are the different node labels that appear in
the axis graph, have one positive class each because each
neighbourhood represents a different node label. (Note that the
$n$ different positive classes do not overlap.) Moreover, the
union of those $n$ sets (each coming from a different partition)
also constitute a partition of $\mathit{Inst}$. In contrast, if we
take only a proper subset of $m$ node labels, $m < n$, the $m$
positive classes of the resulting AxPRE partitions do not
constitute a partition because their union does not have all nodes
in $\mathit{Inst}$. \qed
\end{example}

Given an AxPRE, the positive classes plus one additional class for
the empty neighbourhood forms a partition. If we have another
AxPRE whose positive classes fall exclusively within this empty
neighbourhood class, then these two AxPREs may be used together to
summarize an axis graph. We are interested in sets of AxPREs whose
positive classes define a partition of $\mathit{Inst}$, which is
formalized next.

\begin{definition} [Positive Partition] \label{def:positivePartition}
Let $\mathcal{A} = (\mathit{Inst},$ $Axes,$ $Label,$ $\lambda)$ be
an axis graph. A set $\mathbb{A} = \{\alpha_1, \ldots, \alpha_n
\}$ of AxPREs defines a \emph{positive partition} of
$\mathcal{A}$, denoted $\mathcal{P}_\mathbb{A}$, iff $\bigcup_i
\mathcal{P}^+_{\alpha_i}$ is a partition of $\mathit{Inst}$. \qed
\end{definition}

The intuition behind the notion of positive partition from a set
of AxPREs $\mathbb{A} = \{ \alpha_1, \ldots, \alpha_n \}$ can be
explained as follows. We know, by Definition
\ref{def:positivePartition}, that each $\alpha_i$ in $\mathbb{A}$
defines an AxPRE partition which has positive classes and a unique
empty neighbourhood class. In order for the set $\mathbb{A}$ to
define a positive partition, the empty neighbourhood class of
$\alpha_i$ has to be further partitioned by some $\alpha_j$ in
$\mathbb{A}$. In other words, when the entire set $\mathbb{A}$ is
considered, every node that belongs to the empty neighbourhood of
some $\alpha_i$ also belongs to some positive class of some
$\alpha_j$.

\begin{example}[Positive Partition]
Positive partitions play a key role in our framework. This
requires a thorough understanding of the semantics of the AxPREs,
and the partitions they define. We discuss now some particular
cases of our running example of Figure \ref{fig:PSIMI}.

Let us consider first the AxPRE $\epsilon$, which evaluated on each axis
graph node will produce as many different neighborhoods as there
are different labels in the axis graph (each neighbourhood
containing a single node). Since all nodes with bisimilar
neighbourhoods will belong to the same class, if there are $n$
different labels in the axis graph the $\epsilon$ positive
partition will contain $n$ classes (Figure \ref{fig:PSIMI-label}
shows below each SD node the sets of the partition for our running
example). The same positive partition can be obtained with the set
of expressions $\mathbb{A} = \{[l_1], \ldots, [l_n]\}$, where
$l_1, \ldots, l_n$ are all the different node labels that appear
in the axis graph. In our running example, the set of expressions
equivalent to $\epsilon$ would contain $[interaction]$,
$[participant]$, etc.

Let us consider now the AxPRE $[participant].$ The partition by
$[participant]$ is obtained as follows: for each node in the axis
graph, we compute the AxPRE neighbourhood corresponding to
$[participant]$, and all nodes with bisimilar neighbourhoods
(i.e., all nodes that are $[participant]$-bisimilar) will belong
to the same class. Thus, the partition will consist of two
classes: one containing all the nodes $v$ such that
$\lambda(v)=participant,$ which is the set $\{4,6,18,23,28\}$ (the
positive class), and the other one with the remaining nodes (the
empty neighbourhood class). On the other hand, the $[\neg
participant]$ partition will create as many classes as nodes $v$
with labels $\lambda(v) \neq participant$ exist in
$\mathit{Inst}$. In our running example, the $[\neg participant]$
partition will have nine positive classes (one per label different
from ``participant'') whereas all nodes with ``participant'' label
will belong to the empty neighbourhood class. The two AxPREs
$[participant]$ and $[\neg participant]$, when put together,
define a positive partition with ten classes (one for each label).
\qed
\end{example}

\section{Describing summaries with AxPREs}

In the previous sections, we have introduced the basic machinery
we need to define \emph{summary descriptor} (SD, for short). An SD
is defined from an axis graph and a set of AxPREs. Intuitively, an
SD consists of an axis graph in which each node has associated an
AxPRE and a set in its AxPRE partition, and whose edges represent
axis relationships between those sets.

\begin{definition}[Summary Descriptor]\label{def:SD}
Let $\mathcal{A} = (\mathit{Inst},$ $Axes,$ $Label,$ $\lambda)$ be
an axis graph of an instance. A \emph{summary descriptor} (SD for
short) of $\mathcal{A}$ is a structure $\mathcal{D}_\mathbb{A} =
(\mathbb{A}, \mathcal{G}, axpre, extent)$ that consists of:
\begin{itemize}
\item a set $\mathbb{A} = \{ \alpha_1, \ldots, \alpha_n \}$ of
AxPREs such that $\mathcal{P}_\mathbb{A}$ is a positive partition
of $\mathcal{A}$ by $\mathbb{A}$; \item an axis graph
$\mathcal{G}=(Sum, Axes^{\mathcal{D}}, Label,
\lambda^{\mathcal{D}})$, called \emph{SD graph}, representing axis
relationships between nodes in the sets (extents) of the positive
partition $\mathcal{P}_\mathbb{A}$ where:
\begin{itemize}
\item $Sum$ is a set of nodes; \item $Axes^{\mathcal{D}}$ is a set
of binary relations $\{ E_1^\mathcal{D}, \ldots, E_n^\mathcal{D}
\}$ in $Sum \times Sum$ such that there is a tuple $\langle
s_j,s_k \rangle$ in $E_i^\mathcal{D}$ iff $\exists E_i^\mathcal{A}
\in Axes,\exists v \in extent(s_j), \exists w \in extent(s_k)
\wedge \langle v,w \rangle \in E_i^\mathcal{A}$ (edges are labeled
by axis names); \item $Label$ is the set of node labels from
$\mathcal{A}$; \item $\lambda^{\mathcal{D}}$ is a function that
assigns labels in $Label$ to nodes in $Sum$.
\end{itemize}
\item a bijective function $axpre$ that assigns AxPREs from
$\mathbb{A}$ to nodes in $Sum$; \item a bijective function
$extent$ that assigns a set from the positive partition
$\mathcal{P}_\mathbb{A}$ to each node in $Sum$ (the set assigned
is called the extent of the node).
\end{itemize} \qed
\end{definition}

An SD has some particular characteristics. The set $\mathbb{A}$
uniquely defines the extents of the SD, and therefore its nodes,
for any particular axis graph instance. In other words, given an
axis graph $\mathcal{A}$ and the set $\mathbb{A}$ we can create
the SD of $\mathcal{A}$ by $\mathbb{A}$. On the other hand, not
any set of AxPREs define a positive partition and thus an SD. The
first SDs we can distinguish are those that are defined by a
unique AxPRE from those that have a multi-AxPRE definition. We
denote the former ones as \emph{homogeneous} SDs because all their
nodes are defined uniformly. Homogeneous SDs are the most common
in the summary literature (e.g., dataguides \cite{GW97}, 1-index
\cite{MS99}, ToXin \cite{RM01}, A(k)-index \cite{KSBG02},
F\&B-Index \cite{KBNK02}, Skeleton \cite{BCF+05}). SDs defined by
multiple AxPREs are called \emph{heterogeneous}.

\begin{definition}[Homogeneous and Heterogeneous SDs]
When the extents of all nodes in a SD $\mathcal{D}$ are defined
with the same AxPRE $\alpha$ (i.e., $|\mathbb{A}|=1$), we say that
the corresponding SD is \emph{homogeneous}. In this case we say
that $\mathcal{D}$ is an $\alpha$ SD. In contrast, if at least two
different nodes are defined with different AxPREs (i.e.,
$|\mathbb{A}| > 1$) we have a \emph{heterogeneous} SD. \qed
\end{definition}

\begin{proposition} \label{prop:labelSD}
Given an axis graph $\mathcal{A}$, and a set $\mathbb{A}$ of
AxPREs. If each $\alpha_i \in \mathbb{A}$ contains only AxPREs of
the form $[l],l\in Label$ different from each other, such that
there is an AxPRE for each label in $\mathcal{A},$ then
$\mathbb{A}$ defines an \emph{heterogeneous} SD. Such an SD is
denoted \emph{label SD}. \qed
\end{proposition}

\begin{proof}
It is easy to see that if $\mathbb{A}$ contains all the labels in
$\mathcal{A}$, each AxPRE $[l]$ will create a positive class
labeled $P_l$ associated to a different SD node $s_l$ such that
all nodes in $\mathcal{A}$ with label $l$ will belong to the
extent of $s_l.$ Since $\mathbb{A}$ contains all the labels in the
document, the set $P = \bigcup_i \mathcal{P}^+_{\alpha_i}$ will be
a partition of $\mathit{Inst}$. \qed
\end{proof}

Note that we need to know the instance in advance in order to
define the set $\mathbb{A}$ accordingly. However, the label SD can
also be defined by the AxPRE $\epsilon$, which makes the label SD
\emph{homogeneous} and its definition independent of the axis
graph. The $\epsilon$ SD will produce exactly the same equivalence
classes that the set $\mathbb{A}$ of Proposition
\ref{prop:labelSD}.

\begin{example}[Summary Descriptor]
Figure \ref{fig:PSIMI-label} shows a label SD for our running
example. Since there are ten different labels in the axis graph of
the instance, there are ten summary nodes in the label SD. Nodes
in the figure are labeled by their AxPREs, so we are considering a
heterogeneous label SD in which $\mathbb{A}$ contains an AxPRE per
label. The extent of each node is depicted below it. Edges
represent summary axis relations. For instance, there is an edge
from $s_2$ to $s_{10}$ labeled $c,$ because there is a $c$ edge in
the axis graph from node 14 (in the extent of $s_2$) to node 16
(in the extent of $s_{10}$).

There are three kinds of edges in the figure, depending on
properties of the sets that participate in the axis relation:
dashed, regular, and bold. Dashed edges, like $(s_2,s_{10})$ with
label $c$, mean that some element in the extent of $s_2$ has a
child in the extent of $s_{10}$. Regular edges, like ($s_6,s_7$)
with label $\mathit{fc}$, mean that every element in the extent of
$s_6$ has a first child in the extent of $s_7$. (Since $c$
includes $\mathit{fc}$, we do not draw a $c$ edge when an
$\mathit{fc}$ edge exists.) Finally, bold edges, like $(s_4,s_5)$
with label $\mathit{fc}$, mean that every element in the extent of
$s_5$ is a first child of some element in the extent of $s_4$ and
that every element in the extent of $s_4$ has a first child in the
extent of $s_5.$ The nodes and edges in the figure constitute the
SD graph of the label SD.

Figure \ref{fig:PSIMI-axpre} shows another heterogeneous SD with a
different set $\mathbb{A}$ where $[participant]$, $[expRoleList]$
and $[expRole]$ from Figure \ref{fig:PSIMI-label} have been
replaced by $[participant].c.\mathit{fc}.ns^*$,
$[expRoleList].\mathit{fc}.ns^*$ and $[expRole].ns^*$,
respectively. \qed
\end{example}

\begin{figure}
    \centering
        \includegraphics{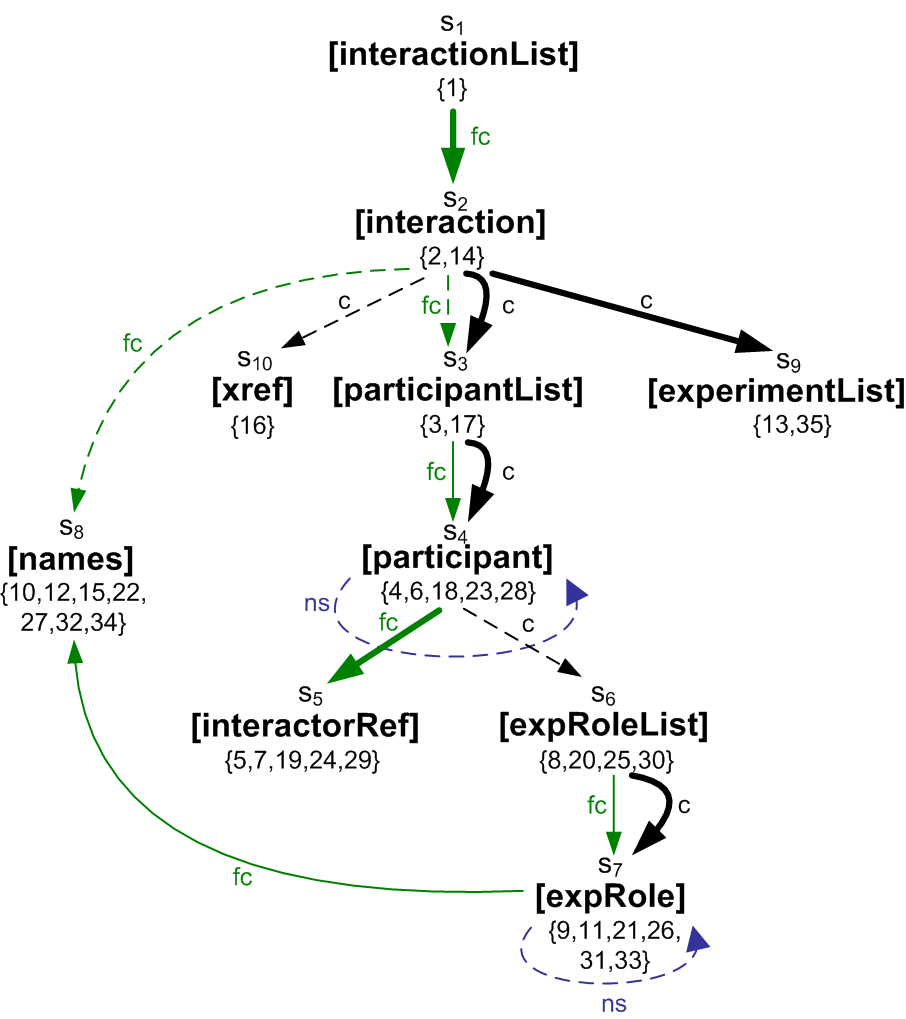}
    \caption{Label SD for the PSI-MI samples}
    \label{fig:PSIMI-label}
\end{figure}

\begin{figure}
    \centering
        \includegraphics{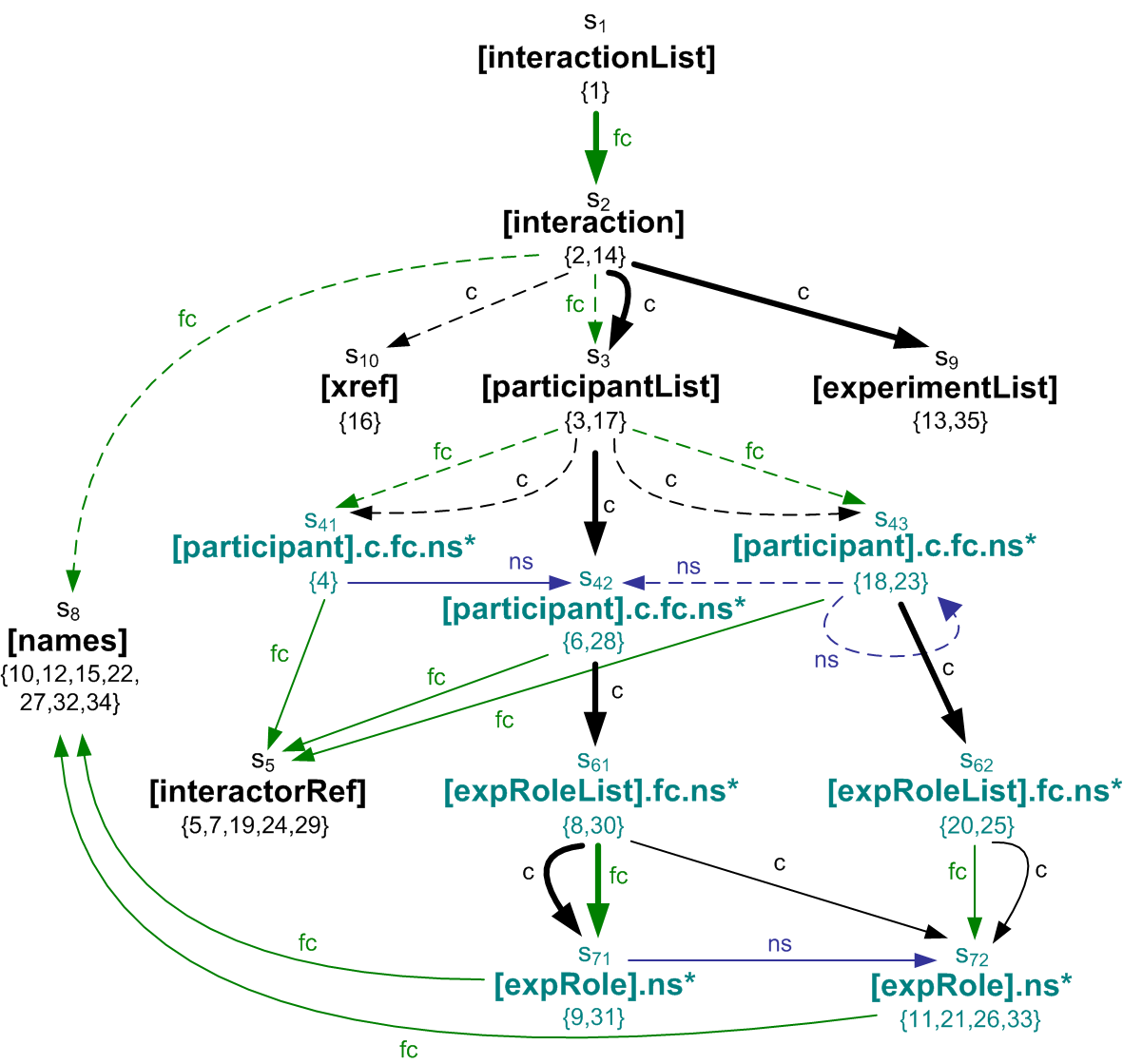}
    \caption{A refined SD for the PSI-MI samples}
    \label{fig:PSIMI-axpre}
\end{figure}

\begin{definition} [Summary Axis Stability]
\label{def:stability} Let $e = \langle s_i, s_j \rangle$ be an SD
graph edge with label $axis$. We say that $e$ is an
\emph{existential} edge if $\exists x \in extent(s_i), \exists y
\in extent(s_j) \wedge \langle x,y \rangle \in axis$, and a
\emph{forward-stable} edge if $\forall x \in extent(s_i), \exists
y \in extent(s_j) \wedge \langle x,y \rangle \in axis$. \qed
\end{definition}

Definition \ref{def:stability} captures the relationship between
edges in the SD graph and the axis graph, and generalizes to
several axes the edge stability representation in XSketch
\cite{PG06b}. Note that all forward-stable edges are also
existential. In Figures \ref{fig:PSIMI-label} and
\ref{fig:PSIMI-axpre}, existential edges are represented by dashed
lines and forward-stable edges by solid lines. A dashed line does
not necessarily mean that an edge is not forward-stable, it might
be that stability has not been checked on that edge (existential
edges in Figures \ref{fig:PSIMI-label} and \ref{fig:PSIMI-axpre}
have been checked and are not forward-stable). When an edge $e$
and its inverse are both forward-stable, $e$ is shown in bold
lines.

Algorithm~\ref{alg:createSD} computes an SD $D$ from an axis graph
$A$ and a set $X$ of AxPREs that define a positive partition of
$A$. Essentially, the algorithm creates the positive partition in
one pass over $A$ (outer loop spanning steps 2-18). Loop 3-18
computes the AxPRE neighbourhood of $v$ for each $\alpha$ in $X$
(step 5) in order to find the $\alpha$ for which the AxPRE
neighbourhood of $v$ is non-empty. Since $X$ defines a positive
partition as a precondition, then for every $v$ there is one and
only one $\alpha$ in $X$ such that $\mathcal{N}_{\alpha}(v) \neq
\emptyset$. This guarantees that condition in step 6 is true
exactly once for every $v$ in $A$.

The next task in the algorithm is to find the extent where $v$
belongs. Loop 7-11 compares by bisimulation
$\mathcal{N}_{\alpha}(v)$ with every node in $D$ that has the same
AxPRE $\alpha$. If there is a node $s$ in $D$ with $\alpha$ but
the  $\alpha$ neighbourhoods of $v$ and $s$ are not bisimilar
(step 10), then a new node $s$ is created and $v$ is added to its
extent (steps 12-16). The same happens if there is no $s$ in $D$
with $\alpha$ at all. Since each $v$ in $A$ may be in an $axis$
relationship with nodes in any extent, the final loop 17-18 checks
edge existence (for the input set of axes $Axes^{\mathcal{D}}$)
between the node $s$ such that $v \in extent(s)$ and every other
node in $D$. The result of the algorithm is an SD $D$ where each
$s$ in $D$ has associated a set in the positive partition of $A$
by $X$ and the axes in $Axes^{\mathcal{D}}$ satisfy the conditions
in Definition~\ref{def:stability}.

As shown, loop 2-18 performs $|\mathit{Inst}|$ iterations. At any
given moment, there is at most the same number of nodes in $D$ as
in $A$ (each extent having only one node) and all have the same
AxPRE. Therefore, loop 7-11 performs $|\mathit{Inst}|$ iterations
in the worst case. Each iteration computes an AxPRE bisimulation
(step 10) with time complexity $O(m . log|\mathit{Inst}|)$, where
$m$ is the total number of tuples (edges) in all axes in $Axis$.
The worst case for loop 17-18 is the same as that of loop 7-11, so
it also performs $|\mathit{Inst}|$ iterations. Thus, the total
time complexity of Algorithm~\ref{alg:createSD} is
$O(|\mathit{Inst}| . m . log|\mathit{Inst}|)$.

\vspace{.2in}

\begin{figure}
\begin{algorithm}
\label{alg:createSD}

    \mbox{}\\
    \parbox{\textwidth}{
    \noindent{\sl \underline{createSD$(A, X)$}}\\[0.5em]
    \noindent{\bf Input:} An axis graph $A$, a set $X$ of AxPREs that defines a positive partition of $A$, and a set $Axes^{\mathcal{D}}$ of SD axes where each axis contains only the empty tuple \\
    \noindent{\bf Output:} An SD $D$

    \begin{algorithmic} [1]

    \STATE create empty SD $D$
    \FOR {every $v$ in $A$}
        \STATE $candidate := \emptyset$
        \FOR {every $\alpha$ in $X$}
            \STATE compute the $\alpha$ neighbourhood of $v$: $\mathcal{N}_{\alpha}(v)$
            \IF {$\mathcal{N}_{\alpha}(v) \neq \emptyset$}
                \FOR {every node $s$ in $D$ such that $axpre(s) := \alpha$ }
                    \STATE let $w$ be a node in $extent(s)$
                    \STATE compute the $\alpha$ neighbourhood of $w$: $\mathcal{N}_{\alpha}(w)$
                    \IF {$v \sim^\alpha w$ (i.e., $\mathcal{N}_{\alpha}(v) \sim \mathcal{N}_{\alpha}(w)$) }
                        \STATE $candidate := s$
                    \ENDIF
                \ENDFOR
                \IF {$candidate = \emptyset$}
                    \STATE create a new node $candidate$ in $D$
                    \STATE $axpre(candidate) := \alpha$
                    \STATE $\lambda^D(candidate) := \lambda(v)$
                \ENDIF
                \STATE add $v$ to $extent(s)$
                \FOR {every node $s' \neq s$ in $D$}
                    \STATE add tuple $\langle s, s' \rangle$ and $\langle s', s \rangle$ to the corresponding $axis$ in $Axes^{\mathcal{D}}$ if conditions in Definition~\ref{def:stability} are satisfied
                \ENDFOR
            \ENDIF
        \ENDFOR
    \ENDFOR
    \end{algorithmic}
}

 \end{algorithm}
 \end{figure}

The notion of an AxPRE neighbourhood can also be defined for an SD
graph, and it is called \emph{summary AxPRE neighbourhood} of a
node. Since an SD Graph is in fact an axis graph
$\mathcal{G}=(Sum, Axes^{\mathcal{D}}, Label,
\lambda^{\mathcal{D}})$, for any given SD node $s$ and AxPRE
$\alpha$ we can define its SD graph automaton
$\mathcal{M}_\mathcal{G}(s)$ (Definition~\ref{def:axisAutomaton})
and intersect it with the AxPRE automaton $\mathcal{M}_{\alpha}$
(Definition~\ref{def:AxPREAutomaton}) in order to obtain an AxPRE
neighbourhood (Definition~\ref{def:neighbourhood}) of $s$.

\begin{definition} [Summary Neighbourhood] \label{def:SDneighbourhood}
Let $\mathcal{D}_\mathbb{A} = (\mathbb{A},$ $\mathcal{G},$
$axpre,$ $extent)$ be an SD, axis graph $\mathcal{G}=(Sum,$
$Axes^{\mathcal{D}},$ $Label,$ $\lambda^{\mathcal{D}})$ its SD
graph, $s$ a node in $\mathcal{G}$, $\alpha$ an AxPRE, and
$\mathcal{M}_\mathcal{G}(s) \cap \mathcal{M}_\alpha$ the
intersection automaton of $\mathcal{M}_\mathcal{G}(s)$ and
$\mathcal{M}_\alpha$. The \emph{summary neighbourhood} of $s$ by
$\alpha$, denoted $\mathcal{N}^\mathcal{G}_{\alpha}(s),$ is the
subgraph of $\mathcal{G}$ as in Definition
\ref{def:neighbourhood}. \qed
\end{definition}

\begin{figure}
    \centering
        \includegraphics{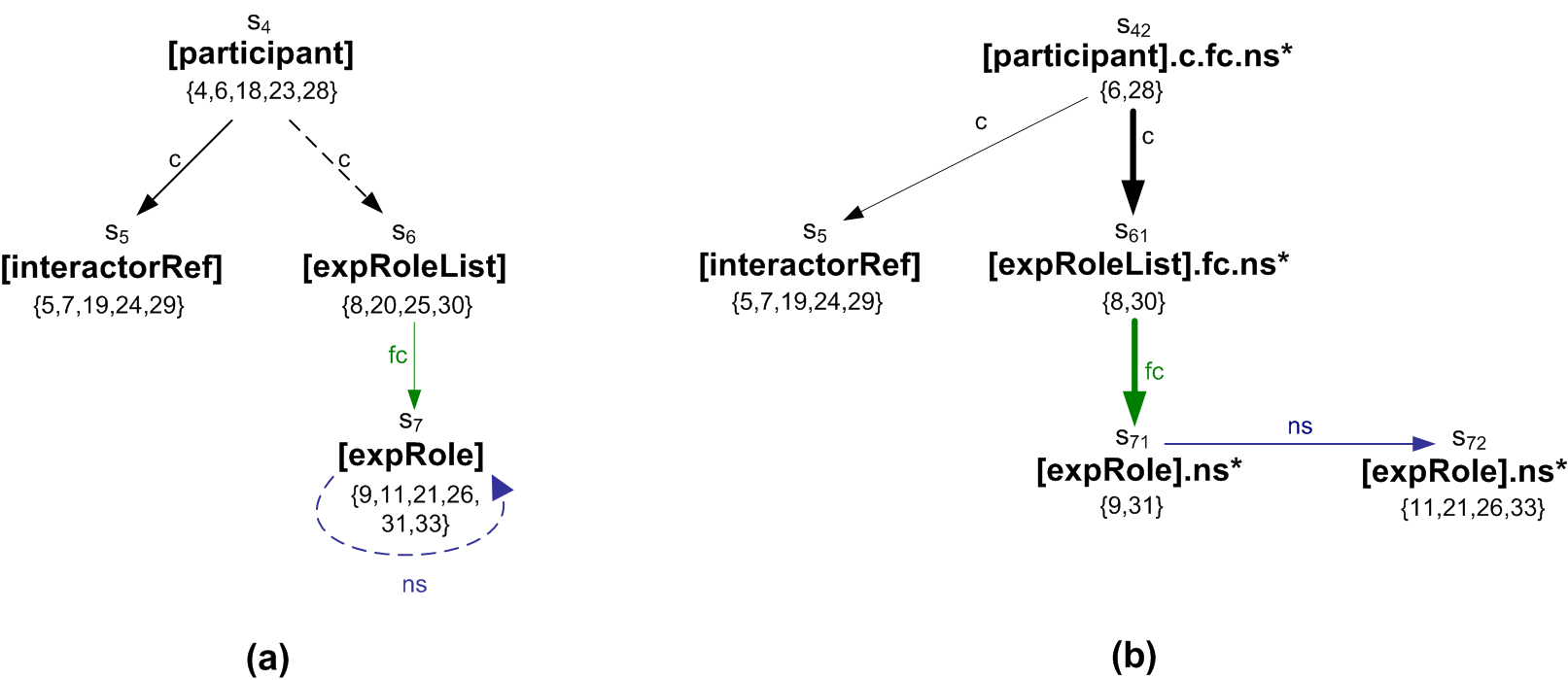}
    \caption{$[participant].c.\mathit{fc}.ns^*$ neighbourhoods of Figure \ref{fig:PSIMI-label} (a) and Figure \ref{fig:PSIMI-axpre} (b) SDs }
    \label{fig:PSIMI-SDNeighbourhoods}
\end{figure}

\begin{definition}[Partition Refinement]
\label{def:AxPREPartitionRefinement} Let $\mathcal{A} =
(\mathit{Inst},$ $Axes,$ $Label,$ $\lambda)$ be an axis graph. If
$\mathcal{P}_\mathbb{A}$ and $\mathcal{P}_\mathbb{B}$ are positive
partitions of $\mathcal{A}$, $\mathcal{P}_\mathbb{A}$ is a
\emph{partition refinement} of $\mathcal{P}_\mathbb{B}$ if every
set of $\mathcal{P}_\mathbb{A}$ is contained in a set of
$\mathcal{P}_\mathbb{B}$. \qed
\end{definition}

\begin{definition}[SD Refinement]
\label{def:SDRefinement} Let $\mathcal{A} = (\mathit{Inst},$
$Axes,$ $Label,$ $\lambda)$ be an axis graph and
$\mathcal{D}_\mathbb{A} = (\mathbb{A}, \mathcal{G}, extent)$ and
$\mathcal{D}_\mathbb{B} = (\mathbb{B}, \mathcal{G}', extent')$ be
two SDs of $\mathcal{A}$. $\mathcal{D}_\mathbb{A}$ is an \emph{SD
refinement} of $\mathcal{D}_\mathbb{B}$ if
$\mathcal{P}_\mathbb{A}$ is a partition refinement of
$\mathcal{P}_\mathbb{B}$. \qed
\end{definition}

\begin{proposition}
Let $\mathcal{A} = (\mathit{Inst},$ $Axes,$ $Label,$ $\lambda)$ be
an axis graph, $\alpha$ and $\beta$ be AxPREs, and
$\mathcal{P}_\alpha$ and $\mathcal{P}_\beta$ be AxPRE partitions
of $\mathcal{A}$. If $\alpha$ is contained in $\beta$ then
$\mathcal{P}_\beta$ is a refinement of $\mathcal{P}_\alpha$. \qed
\end{proposition}

\begin{proof} (sketch)
The proof follows from the notion of AxPRE neighbourhoods. If
$\alpha$ is contained in $\beta$ then for any given node $v$, its
$\alpha$ neighbourhood is contained in its $\beta$ neighbourhood.
Consequently, two nodes that are not distinguished by $\alpha$
(i.e., they are $\alpha$-bisimilar) may be distinguished by
$\beta$, but not the other way around. This guarantees that
$\beta$ creates either the same partition as $\alpha$ or a
refinement. \qed
\end{proof}

\begin{corollary}
Let $\mathcal{A} = (\mathit{Inst},$ $Axes,$ $Label,$ $\lambda)$ be
an axis graph and $\mathcal{D}_\mathbb{A} = (\mathbb{A},
\mathcal{G}, extent)$ and $\mathcal{D}_\mathbb{B} = (\mathbb{B},
G', extent')$ be two SDs of $\mathcal{A}$. If every $\beta \in
\mathbb{B}$ is contained in some $\alpha \in \mathbb{A}$ then
$\mathcal{D}_\mathbb{A}$ is an \emph{SD refinement} of
$\mathcal{D}_\mathbb{B}$. \qed
\end{corollary}

\vspace{.1in}

\begin{example} [SD Refinement] \label{ex:refinement}
Let us consider the label SD of Figure \ref{fig:PSIMI-axpre}.
Recall that in the label SD, $\mathbb{A}= \{[l_1],...,[l_n]\},$
where $l_i \in Label,$ $l_i \neq l_j~\forall~i,j,$ and
$\bigcup_i~{l_i}=Label.$ Suppose we want to refine node $s_4.$ For
this node, the partition represented in the figure was produced by
the AxPRE $[participant].$  If we replace this AxPRE by
$[participant].c$ in $\mathbb{A},$ and apply this set of  AxPREs
to $Inst$, two nodes will be produced, let us call these nodes
$s_{41}$ and  $s_{42},$ with extents $\{4\}$ and $\{6,18,23,28\}$,
respectively ($s_{4}$ will not appear because the AxPRE which
produced it was replaced by the new one). This occurs because node
4 in the axis graph has one child (namely interactorRef) while the
other four nodes have two children each (interactorRef and
expRoleList). Thus, applying $[participant]c$ we obtain two
different AxPRE neighbourhoods, plus the empty neighbourhood,
which is itself partitioned by the remaining AxPREs.

Analogously, if we want to refine the extent of $s_{42}$ further
using the AxPRE $c.ns,$  we will replace the AxPRE
$[participant].c$ by $[participant].c.c.ns.$ This will produce
three sets, with extents: $\{4\},$ $\{6,28\},$ $\{18,23\}.$

Finally, suppose now that the label SD is defined using
$\mathbb{A}= \epsilon,$ and we want to refine node $s_4$ with
$[participant].c.$ In this case, just adding the new AxPRE does
not suffice, because we would not obtain an SD: the union of
positive partitions will not be a partition of $Inst$ because
$\epsilon$ will still produce its own partitions. We solve this
adding the AxPRE $[\neg participant],$ which will produce the
remainder of the label SD and will send all nodes labeled
$participant$ to the empty neighbourhood class. \qed
\end{example}

\vspace{.1in}

The notions of partition and SD refinement, besides describing the
axis structure of an axis graph, allows us to define a
\emph{hierarchy} of SDs. This provides the basis for recognizing a
lattice among different SDs, where each node corresponds to a
different AxPRE definition. We will show that this lattice covers
all the summaries addressed in the literature, plus more complex
new ones. At the top of this hierarchy (i.e., the coarsest
partition), the empty AxPRE defines a SD where each node is
partitioned by label (as shown in Figure \ref{fig:prodlattice}), a
typical summary found in the literature \cite{CM94,NUWC97}. The
bottom of the lattice may vary, although the finest partition
granularity can be represented  by the expression
$(\mathit{fc}.ns^*)^*,$ that produces a partition in which each
node in the axis graph will belong to a different equivalence
class.

\begin{definition}[\dx\ Lattice] \label{def:lattice}
A \dx\ lattice with respect to a set of axes
$A=\{a_1,\ldots,a_n\}$ is defined as follows: each node
corresponds to an AxPRE generated by the grammar of Definition
\ref{def:axpre} when the terminal axis is one of $a_1,...,a_n.$
Also, there is an edge $(n_1,n_2)$ in the lattice if and only if
the AxPRE of $n_2$ is contained in the AxPRE of $n_1.$ \qed
\end{definition}

From Definition \ref{def:lattice} it follows that the coarsest
partition that the lattice may define is the label SD. The finest
partition depends on the chosen set of axes.

\vspace{.2in}

This chapter provided an overview of the \dx\ framework, including
the AxPRE language and some fundamental notions like
neighbourhood, bisimilarity, and summary descriptor (SD). In the
next chapter we will discuss how the \dx\ lattice captures and
generalizes many proposals in the literature.

\chapter{Capturing earlier literature proposals with \dx}
\label{sec:lattice}

\dx\ summaries can be classified in a lattice that describes a
\emph{refinement} relationship between entire summaries
(Definition \ref{def:lattice}). In this chapter we revisit some of
the related work discussed in Chapter \ref{sec:related} that can
be captured in such a lattice by the \dx\ framework.

Figure \ref{fig:prodlattice} shows a fragment of a \dx\ summary
lattice that captures earlier proposals based on the notion of
bisimilarity (in green) and ad-hoc constructions (in red). Each
node in the figure corresponds to a homogeneous SD defined by an
AxPRE. \dx\ not only captures most summary proposals but also
provides a declarative way of defining entirely new ones: nodes
and edges in blue are a sample of the richer SDs that were never
considered in the literature, like the one that appears in Figure
\ref{fig:PSIMI-axpre} ($c.\mathit{fc}.ns^*$) and in Chapter
\ref{Section:Experiments} ($p^*|c.\mathit{fs}$).

\begin{figure}
    \centering
    \includegraphics{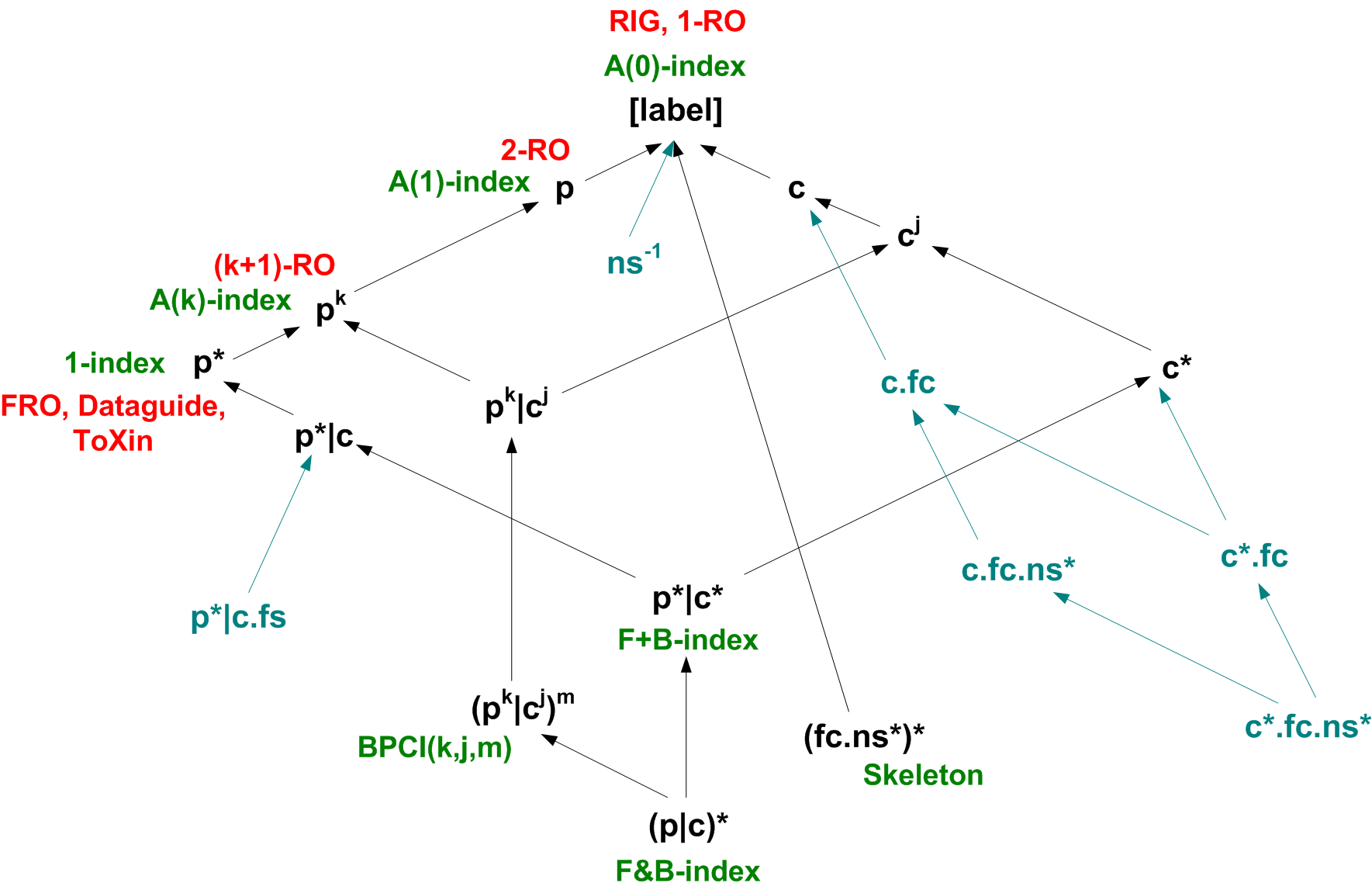}
    \caption{AxPRE summary lattice capturing earlier homogeneous proposals}
     \label{fig:prodlattice}
\end{figure}

\section{Bisimilarity-based proposals}

The earliest bisimilarity-based summary proposal is the family
presented in \cite{MS99}, which contains a $p^*$ summary: the
1-index. The 1-index partition is computed by using
\emph{bisimulation} as the equivalence relation. The F\&B-Index
\cite{KBNK02}, is an example of a $(p|c)^*$ SD. The F\&B-Index
construction uses bisimulation like the 1-index, but applied to
the edges and their inverses in a recursive procedure until a
fix-point. With this construction, the F\&B-Index's equivalence
classes are computed according to the incoming and outgoing label
paths of the nodes. The same work introduces the F+B-index (a
$p^*|c^*$ AxPRE summary constructed by applying bisimulation to
the edges and their inverses only once) and the BPCI(k,j,m) index
(a $(p^k|c^j)^m$ AxPRE summary, where $k$, and $j$ controls the
lengths of the paths  and $m$ the iterations of the bisimulation
on the edges and their inverses). The F+B-index and the F\&B-index
are BPCI($\infty, \infty, 1$) and BPCI($\infty, \infty, \infty$)
respectively. The A(k)-index \cite{KSBG02} is a $p^k$ AxPRE
summary based on $k$-bisimilarity (bisimilarity computed for paths
of length $k$). Thus, the A(0)-index is a label SD, the A(1)-index
is a $p$ SD, the A(2)-index is a $p.p$ SD, and the A(h)-index is
the $p^h$ SD. We discuss some of these proposals in more detail.

Unlike standard definitions in the bisimulation literature
\cite{PT87,DPP04}, 1-index, A(k)-index, F\&B-index, and
BPCI(k,j,m) use a bisimulation defined backwards in order to
capture incoming paths to the nodes. We provide next a definition
of backwards bisimulation and bisimilarity for completeness. In
the literature, the only axes considered are $c$ and
$\mathit{idref}$.

\vspace{.2in}

\begin{definition}[Backwards Bisimulation and Bisimilarity] \label{def:backwardsbisimulation}
Let $\mathcal{G}_1$ and $\mathcal{G}_2$ be two rooted subgraphs of
an axis graph $\mathcal{A} = (\mathit{Inst},$ $Axes,$ $Label,$
$\lambda)$, such that $Axes_{\mathcal{G}_1} \subseteq Axes$ and
$Axes_{\mathcal{G}_2} \subseteq Axes$, and let $r_1,r_2 \in
\mathit{Inst}$ be the roots of $\mathcal{G}_1$ and $\mathcal{G}_2$
respectively. A \emph{backwards bisimulation} between
$\mathcal{G}_1$ and $\mathcal{G}_2$ is a symmetric relation
$\approx_b$ such that for all $v \in \mathcal{G}_1$, $w \in
\mathcal{G}_2$, $E_i^{\mathcal{G}_1} \in Axes_{\mathcal{G}_1}$,
and $E_i^{\mathcal{G}_2} \in Axes_{\mathcal{G}_2}$: if $v
\approx_b w,$ then $\lambda(v)=\lambda(w)$; if $v \approx_b w,$
and $\langle v', v \rangle \in E_i^{\mathcal{G}_1}$, then $\langle
w', w \rangle \in E_i^{\mathcal{G}_2}$ and $v' \approx_b w'$. Two
nodes $v \in \mathcal{G}_1$, $w \in \mathcal{G}_2$ are
\emph{backward bisimilar}, in notation $v \sim_b w$, iff there
exist a backwards bisimulation $\approx_b$ between $\mathcal{G}_1$
and $\mathcal{G}_2$ such that $v \approx_b w$. \qed
\end{definition}

It is easy to see that the backwards bisimulation is an
equivalence relation. The F\&B-Index construction uses backwards
bisimulation like the 1-index, but applied to $c$ and
$\mathit{idref}$ edges and their inverses. Algorithm
\ref{alg:fbconstruction} computes the equivalence classes for the
F\&B-Index according to both incoming and outgoing label paths of
the nodes.

\vspace{.2in}

\begin{figure}
\begin{algorithm}
\label{alg:fbconstruction}

    \mbox{}\\
    \parbox{\textwidth}{
    \noindent{\sl \underline{F\&B$-construction(G)$}}\\[0.5em]
    \noindent{\bf Input:} Data graph $G$\\
    \noindent{\bf Output:} F\&B-index $I$

    \begin{algorithmic} [1]

    \STATE let $\mathcal{P}$ be a partition of the nodes in $G$
    \STATE $\mathcal{P}\leftarrow$ label SD partition of $G$
    \REPEAT{}
        \STATE reverse all edges in $G$
        \STATE $\mathcal{P}\leftarrow$ compute the \emph{backwards bisimilarity partition} of $G$
        initializing the computation with $\mathcal{P}$
        \STATE reverse all edges in $G$, obtaining the original $G$
        \STATE $\mathcal{P}\leftarrow$ compute the \emph{backwards bisimilarity partition} of $G$
        initializing the computation with $\mathcal{P}$
    \UNTIL{$\mathcal{P}$ does not change (fix point)}
    \FOR {each equivalence class $P_i \in \mathcal{P}$}
        \STATE create an index node $s \in I$
        \STATE $extent(s)\leftarrow P_i$
    \ENDFOR
    \FOR {each edge from $v$ to $w$ in $G$}
        \STATE let $s \in I$ be an index node such that $v \in
        extent(s)$
        \STATE let $s' \in I$ be an index node such that $w \in
        extent(s)$
        \IF {there is no edge from $s$ to $s'$}
            \STATE create an edge from $s$ to $s'$
        \ENDIF
    \ENDFOR
    \end{algorithmic}
}
 \end{algorithm}
 \end{figure}

\begin{proposition} \label{pro:fbindexEquivalence}
Let $G$ be an axis graph with $Axes=\{c\}$ (or $Axes=\{c,
\mathit{idref} \}$). The F\&B-index of $G$ is a $(p|c)^*$ SD (or a
$(p| c | \mathit{idref} | \mathit{idref}^{-1})^*$ SD). \qed
\end{proposition}

\begin{proof}
The input data graph used in the F\&B-index construction
(Algorithm \ref{alg:fbconstruction}) can be viewed as an axis
graph with the $c$ axis, in which the reversed edges of lines $4$
and $6$ correspond to the $c^{-1}$ axis (equivalent to a $p$
axis). Therefore, for simplicity, instead of reversing edges we
use an axis graph $G$ with $Axes=\{c\}$ and take its inverse when
necessary. If id-idrefs are considered, then $Axes=\{c,
\mathit{idref}\}$.

Let us consider first the case of $Axes=\{c\}$. We start with the
label SD in Line $2$, which is an $\epsilon$ SD. Lines $4$ and $5$
are equivalent to refining all nodes in the initial $\epsilon$ SD
by the $c^*$ AxPRE. This produces a $c^*$ SD. Then, lines $6$ and
$7$ produce a refinement of all $c^*$ nodes by the $p^*$ AxPRE,
thus obtaining a $c^*.p^*$ SD. The iterative process until the fix
point can be represented in our framework as a Kleene closure of
the $c^*.p^*$ AxPRE, which yields a $(c^*.p^*)^*$ SD. It is easy
to see that AxPRE $(c^*.p^*)^*$ produces the same SD as $(p |
c)^*$ (by identity of regular expressions). The remainder of the
algorithm (lines $9$-$16$) creates existential edges like in
Definition \ref{def:stability}.

When $Axes=\{c, \mathit{idref} \}$, the argument is similar but
with AxPREs $(c|\mathit{idref})^*$ and $(p|\mathit{idref}^{-1})^*$
instead of $c^*$ and $p^*$, respectively. In this case, the final
AxPRE for the SD is $(p| c | \mathit{idref} |
\mathit{idref}^{-1})^*$. \qed
\end{proof}

\vspace{.1in}

The notion of k-bisimilarity used in the A(k)-index was defined to
capture incoming paths on $c$ and $\mathit{idref}$ edges of length
up to $k$. We provide next a more general definition for axis
graphs that supports paths on all types of axes.

\begin{definition}[Backwards k-Bisimulation and k-Bisimilarity] \label{def:backwards-k-bisimulation}
Let $\mathcal{G}_1$ and $\mathcal{G}_2$ be two rooted subgraphs of
an axis graph $\mathcal{A} = (\mathit{Inst},$ $Axes,$ $Label,$
$\lambda)$, such that $Axes_{\mathcal{G}_1} \subseteq Axes$ and
$Axes_{\mathcal{G}_2} \subseteq Axes$, and let $r_1,r_2 \in
\mathit{Inst}$ be the roots of $\mathcal{G}_1$ and $\mathcal{G}_2$
respectively. A \emph{backwards k-bisimulation} between
$\mathcal{G}_1$ and $\mathcal{G}_2$ is a symmetric relation
$\approx^k_b$ such that for all $v \in \mathcal{G}_1$, $w \in
\mathcal{G}_2$, $E_i^{\mathcal{G}_1} \in Axes_{\mathcal{G}_1}$,
and $E_i^{\mathcal{G}_2} \in Axes_{\mathcal{G}_2}$: if $v
\approx^0_b w,$ then $\lambda(v)=\lambda(w)$; if $v \approx^k_b
w,$ and $\langle v', v \rangle \in E_i^{\mathcal{G}_1}$, then
$\langle w', w \rangle \in E_i^{\mathcal{G}_2}$ and $v'
\approx^{k-1}_b w'$. Two nodes $v \in \mathcal{G}_1$, $w \in
\mathcal{G}_2$ are \emph{backward k-bisimilar}, in notation $v
\sim^k_b w$, iff there exist a backwards k-bisimulation
$\approx^k_b$ between $\mathcal{G}_1$ and $\mathcal{G}_2$ such
that $v \approx^k_b w$. \qed
\end{definition}

Note that backwards k-bisimilarity defines an equivalence relation
on the nodes in the axis graph. The partition created by the
backwards k-bisimilarity corresponds to the A(k)-index, where $k$
is a parameter that represents the length of the incoming paths
summarized by the index.

\begin{figure}
\begin{algorithm}
\label{alg:BPCIconstruction}

    \mbox{}\\
    \parbox{\textwidth}{
    \noindent{\sl \underline{BPCI$-construction(G, k_{in}, k_{out},td)$}}\\[0.5em]
    \noindent{\bf Input:} Data graph $G$, local similarities $k_{in}$ and $k_{out}$, tree depth $td$\\
    \noindent{\bf Output:} BPCI$(k_{in}, k_{out},td)$ $I$

    \begin{algorithmic} [1]

    \STATE let $\mathcal{P}$ be a partition of the nodes in $G$
    \STATE $\mathcal{P}\leftarrow$ label SD partition of $G$
    \FOR{i=1 to $td$}
        \STATE reverse all edges in $G$
        \STATE $\mathcal{P}\leftarrow$ compute the \emph{backwards $k_{in}$-bisimilarity partition} of $G$
        initializing the computation with $\mathcal{P}$
        \STATE reverse all edges in $G$, obtaining the original $G$
        \STATE $\mathcal{P}\leftarrow$ compute the \emph{backwards $k_{out}$-bisimilarity partition} of $G$
        initializing the computation with $\mathcal{P}$
    \ENDFOR
    \FOR {each equivalence class $P_i \in \mathcal{P}$}
        \STATE create an index node $s \in I$
        \STATE $extent(s)\leftarrow P_i$
    \ENDFOR
    \FOR {each edge from $v$ to $w$ in $G$}
        \STATE let $s \in I$ be an index node such that $v \in
        extent(s)$
        \STATE let $s' \in I$ be an index node such that $w \in
        extent(s)$
        \IF {there is no edge from $s$ to $s'$}
            \STATE create an edge from $s$ to $s'$
        \ENDIF
    \ENDFOR
    \end{algorithmic}
}
 \end{algorithm}
 \end{figure}

\begin{proposition} \label{pro:a(k)indexEquivalence}
Let $G$ be an axis graph with $Axes=\{c\}$ (or $Axes=\{c,
\mathit{idref}\}$). The A(k)-index of $G$ is a $p^k$ SD (or a
$(p|\mathit{idref})^k$ SD).  \qed
\end{proposition}

\begin{proof}
Consider an axis graph $G$ with $Axes=\{c\}$. Two nodes $v,w$
belong to the same extent in the $p^k$ SD iff they are
$p^k$-bisimilar. In addition, we know that $v \sim_{p^k} w$ iff
there exists neighbourhoods $\mathcal{N}_{p^k}(v)$ and
$\mathcal{N}_{p^k}(w)$ such that $v \sim w$. This means we can
define a backwards k-bisimulation $\approx^k_b$ between
$\mathcal{N}_{p^k}(v)$ and $\mathcal{N}_{p^k}(w)$ such that $v
\approx^k_b w$ and thus $v \sim^k_b w$. \qed
\end{proof}

The BPCI($k_{in},k_{out},td$)-index is another proposal based on
the notion of backwards k-bisimulation. Algorithm
\ref{alg:BPCIconstruction} constructs a
BPCI($k_{in},k_{out},td$)-index. Algorithm
\ref{alg:BPCIconstruction} is similar to Algorithm
\ref{alg:fbconstruction} but uses $k_{in}$-bisimilarity for the
reversed edges (line 5), $k_{out}$-bisimilarity for the original
edges (line7), and a $td$ number of iterations instead of a fix
point (lines 3 to 7).

\begin{proposition} \label{pro:BPCIindexEquivalence}
Let $G$ be an axis graph with $Axes=\{c\}$ (or $Axes=\{c,
\mathit{\mathit{idref}}\}$). The
BPCI($k_{in},k_{\mathit{out}},td$)-index of $G$ is a
$(p^{k_{in}}|c^{k_{\mathit{out}}})^{td}$ SD (or a $(p^{k_{in}} |
c^{k_{\mathit{out}}} | \mathit{\mathit{idref}}^{k_{\mathit{out}}}
| (\mathit{\mathit{idref}}^{-1})^{k_{in}})^{td}$ SD). \qed
\end{proposition}

\begin{proof}
Like for F\&B-index construction (Algorithm
\ref{alg:fbconstruction}) the input data graph $G$ can be viewed
as an axis graph with the $c$ axis, in which the reversed edges
correspond to the $c^{-1}$ (or $p$) axis. If id-idrefs are
considered, then $Axes=\{c, \mathit{idref}\}$.

Let us consider first the case of $Axes=\{c\}$. Lines $4$ and $5$
are equivalent to refining all nodes in the initial $\epsilon$ SD
(line 2) by the $c^{k_{out}}$ AxPRE. This produces a $c^{k_{out}}$
SD. Then, lines $6$ and $7$ produce a refinement of all
$c^{k_{out}}$ nodes by the $p^{k_{in}}$ AxPRE, thus obtaining a
$c^{k_{out}}.p^{k_{in}}$ SD. The iterative process is repeated
$td$ times, which is equivalent to a
$(c^{k_{out}}.p^{k_{in}})^{td}$ SD. Again, by identity of regular
expressions $(c^{k_{out}}.p^{k_{in}})^{td}$ is equivalent to as
$(c^{k_{out}} | p^{k_{in}})^{td}$. The remaining of the algorithm
(lines $9$-$16$) creates existential edges like in Definition
\ref{def:stability}. When $Axes=\{c, \mathit{idref} \}$, the final
AxPRE for the SD is $(c^{k_{out}} | p^{k_{in}} |
\mathit{idref}^{k_{out}} | (\mathit{idref}^{-1})^{k_{in}} )^{td}$.
\qed
\end{proof}

The Skeleton summary \cite{BCF+05} clusters together nodes with the
same subtree structure, thus capturing node ordering in subtrees.
Skeleton uses an entirely different construction approach, but its
essence can be captured by the $(fc.ns^*)^*$ AxPRE.

The D(k)-index \cite{QLO03}, and M(k)-index \cite{HY04} are
heterogeneous SD proposals. All nodes $s_i$ are described by $p^k$
AxPREs with a different $k$ per $s_i$. They use different
construction strategies based on dynamic query workloads and local
similarity (i.e., the length of each path depends on its location
in the XML instance) to determine the subset of incoming paths to
be summarized.

XSketch \cite{PG06b} manages summaries capturing many (but not
all) heterogeneous SD's along the $p$ and $c$ axis, ranging from
the label summary to the F\&B-Index. However there is no control
over the refinements chosen, nor a description of the intermediate
summaries obtained. This makes sense given that XSketch objective
is to provide selectivity estimates. As such, its construction
algorithm is guided by heuristics to optimize the space/accuracy
trade-off.

\section{Ad-hoc construction proposals}

Region inclusion graphs (RIGs) \cite{CM94} and representative
objects of length 1 (1-RO) \cite{NUWC97} are label SDs, that is
$\epsilon$ SDs (because all their nodes $s_i$ are described by the
$\epsilon$ AxPRE). In general, representative objects are $p^k$
SDs for XML tree instances. Therefore, the 1-RO is a label SD, the
2-RO is a $p$ SD, the 3-RO is a $p.p$ SD, and the FRO (full
representative object) is the $p^*$ SD.

Dataguides \cite{GW97} group instance nodes into sets called
\emph{target sets} according to the label paths from the root they
belong to. The dataguide construction is basically a
nondeter\-mi\-nistic-to-deterministic automaton translation. When
the data instance is a tree, the data\-guide's target sets are
equivalent to the extents in our framework: a data\-guide of an
XML tree is a $p^*$ SD.

ToXin \cite{RM01} also has a component that can be viewed as an
$p^*$ SD. ToXin consists of three index structures: the ToXin
schema, the path index, and the value index. The ToXin schema is
defined only for tree instances, and it is equivalent to a $p^*$
SD graph.

\vspace{.2in}

In this chapter, we discussed how \dx\ uses AxPREs to capture many
summary proposals in the literature by providing a declarative
definition for them for the first time. In the next chapter, we
will show how SDs can be declaratively updated by means of two
basic operations, refinement and stabilization applied to
neighbourhoods.

\chapter{Describing extents and neighbourhoods} \label{sec:refinements}

We have seen that several SD nodes can share the same AxPRE
$\alpha$. The reason for this is that each SD node with the same
$\alpha$ corresponds to a different extent in the $\alpha$
partition. In the first section of this chapter, we provide
mechanisms for describing each extent in the partition based on
neighbourhoods, sets of axis label paths, and AxPREs.

The description provided by a node in the SD can be changed by an
operation that modifies its AxPRE and thus the AxPRE neighbourhood
of the nodes in its extent. When the new AxPRE partition thus
obtained constitutes a refinement of the old one, the operation is
called an \emph{AxPRE refinement}. The notion of refinement is
tightly related to that of \emph{stabilization}. An edge
stabilization determines the partition of an extent into two sets
based on the participation (total or partial) of the extent nodes
in the axis relation the edge represents. In the second section of
this chapter, we discuss in detail our approaches to refinement
and stabilization based on AxPREs.

\section{Concise descriptions} \label{sec:acyclic}

Since several SD nodes can share the same AxPRE, we need a
mechanism for uniquely describe each SD node and its extent. The
most straightforward way to do that would be just to list all
nodes that belong to the extent (extensional definition). A more
concise description is provided by the $\alpha$ neighbourhood of
any node in the extent. Since all nodes in an extent are
bisimilar, any $\alpha$ neighbourhood can be used to find all the
other nodes in the extent by bisimulation.

In order to get the most concise description, we need to find the
\emph{smallest} (in terms of number of nodes) neighbourhood in the
extent of $s$ that is bisimilar to all the others. We can do this
by computing a \emph{bisimulation contraction} over all
neighbourhoods in the extent of $s$. The bisimulation contraction
of a given graph is the smallest graph that is bisimilar to it,
which can be computed in time $O(m \log n)$ (where $m$ is the
number of edges and $n$ is the number of nodes) \cite{PT87}, or
even linearly for acyclic graphs \cite{DPP04}. Based on
bisimulation contraction we define the notion of
\emph{representative neighbourhood}.

\begin{figure}
    \centering
        \includegraphics{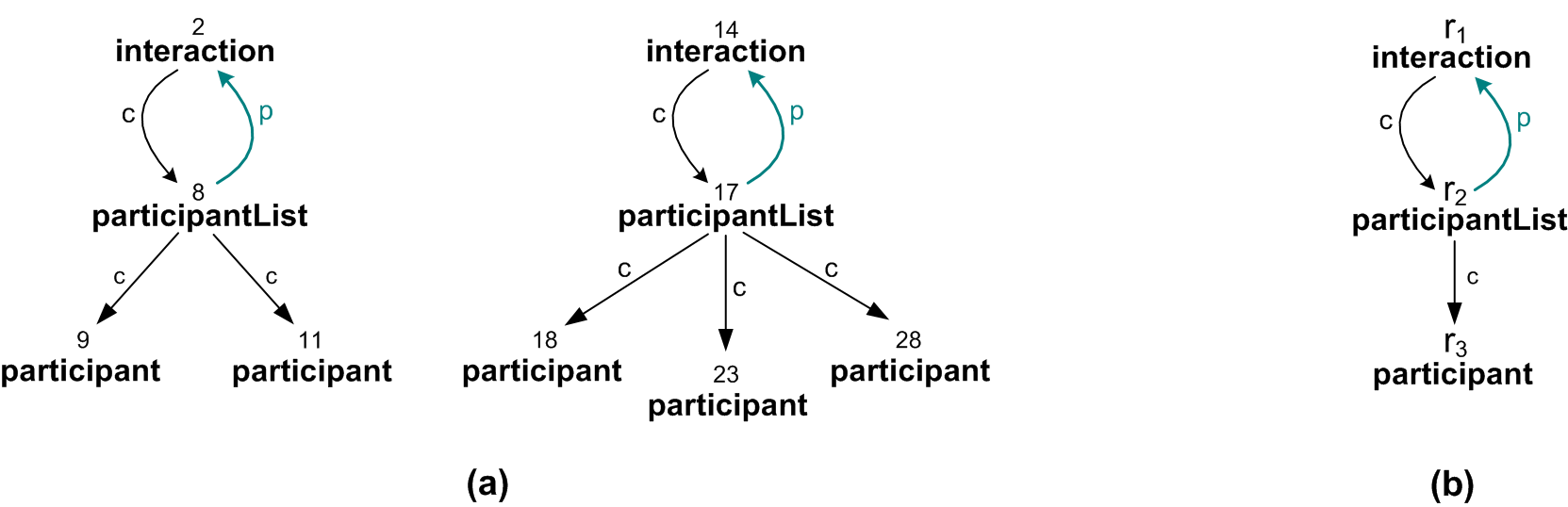}
    \caption{The two $[interaction].c[participantList].(c|p)$ neighbourhoods (a) and their representative neighbourhood (b) from our running example}
    \label{fig:representative}
\end{figure}

\begin{definition} [Representative Neighbourhood] \label{def:representativeNeighbourhood}
Let $\mathcal{D}$ be an SD and $s$ a node in $\mathcal{D}$ such
that $axpre(s)=\alpha$. The \emph{representative neighbourhood} of
$s$ for $\alpha,$ denoted $\mathcal{R}_{\alpha}(s),$ is an axis
graph that is the bisimulation contraction of all neighbourhoods
$\mathcal{N}_{\alpha}(v_i)$, where $v_i \in extent(s)$.
$\mathcal{R}_{\alpha}(s)$ has a single root node $v_0$ that is
bisimilar to all $v_i \in extent(s).$ \qed
\end{definition}

Note that the bisimulation contraction is not necessarily one of
the neighbourhoods in the extent -- it could be smaller than any
of them. Rather, a representative neighbourhood is an entirely new
axis graph that happens to be the smallest that is bisimilar to
all neighbourhoods in an extent.

\begin{example} [Representative Neighbourhood]
Consider the AxPRE partition of our running example described by
AxPRE $[interaction].c[participantList].(c|p)$. It has only one
set containing nodes $2$ and $14$, whose neighbourhoods are shown
in Figure~\ref{fig:representative}~(a). Its representative
neighbourhood
$\mathcal{R}_{[interaction].c[participantList].(c|p)}(s)$ is the
graph shown in Figure~\ref{fig:representative}~(b). Note that such
a neighbourhood does not belong to the extent of $s$ (there is no
\emph{participantList} in the axis graph with only one
\emph{participant} node). \qed
\end{example}

For some neighbourhoods, deciding bisimilarity is equivalent to
comparing the sets of simple label paths from their roots to their
leaves. (A path is \emph{simple} when it has no repeated edges.)
In those cases, neighbourhoods can be described by an \emph{extent
expression} (EE for short), which is capable of computing
precisely the set of elements in the extent of a given SD node and
functions like a virtual view. In
Chapter~\ref{sec:changingSDswithxpath} we provide a mechanism for
expressing EEs in XPath.

\begin{definition} [Path and LPath Sets] \label{def:lpathset}
Let $\mathcal{N}$ be a neighbourhood in an axis graph
$\mathcal{A}$, and $v$ a node in $\mathcal{N}$. We denote by
$Path(v)$ and $LPath(v)$ the set of simple axis paths and simple
axis label paths from $v$, respectively. \qed
\end{definition}

\begin{example}
Consider the neighbourhoods of Figure \ref{fig:representative}
(a). The $Path$ and $LPath$ sets are defined as follows: $Path(2)
= \{ c,c.c,c.p \} = Path(14)$, and $LPath(2) = \{ c[participant$
$List], c[participantList].c[participant],
c[participantList].p[interaction] \}=LPath(14)$. Note that both
sets include all the prefixes. \qed
\end{example}

If deciding bisimilarity between a given set of neighbourhoods is
equivalent to comparing their $LPath$ sets, we say that such
neighbourhoods are \emph{LPath distinguishable}.

\begin{figure}
    \centering
        \includegraphics{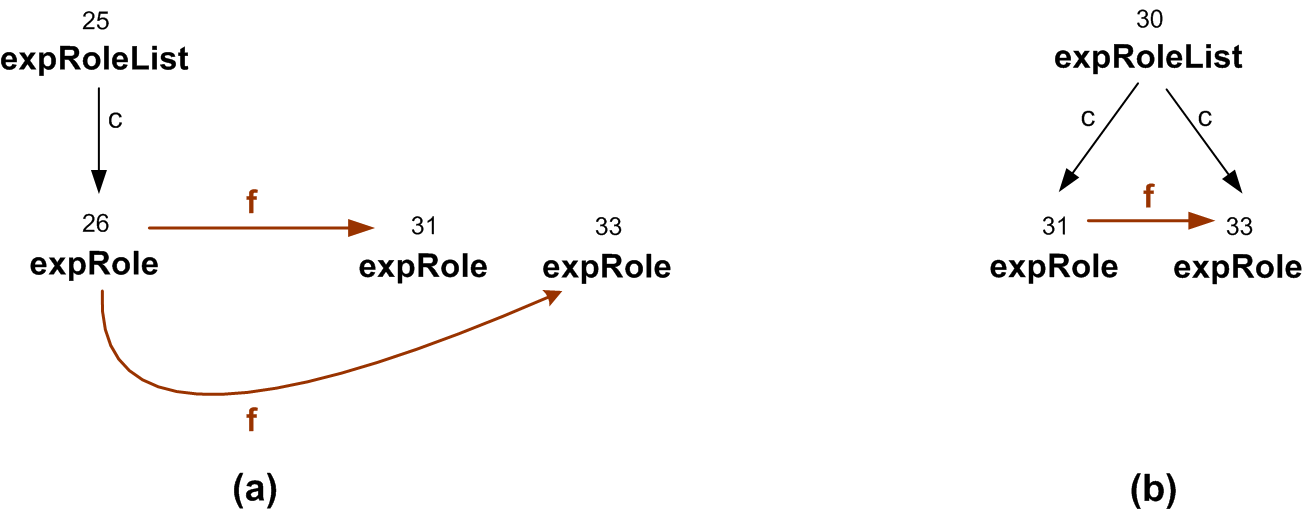}
    \caption{Two $[expRoleList].c.f[expRole]$ neighbourhoods from our running example}
    \label{fig:distinguishable}
\end{figure}

\begin{definition} [LPath Distinguishable] \label{def:lpathdisting}
Let $\mathcal{N}_1(v_1), \ldots, \mathcal{N}_m(v_m)$ be
neighbourhoods in an axis graph $\mathcal{A}$. We say that
$\mathcal{N}_1, \ldots, \mathcal{N}_m$ are \emph{LPath
distinguishable} when, for all $1 \leq i,j \leq m:$
$\mathcal{N}_i(v_i) \sim \mathcal{N}_j(v_j)$ iff
$LPath(v_i)=LPath(v_j)$. \qed
\end{definition}

Although the axis graph neighbourhoods we have considered so far
are all LPath distinguishable, some combination of axes may
produce neighbourhoods that are not, as illustrated by the next
example.

\begin{example} [LPath Distinguishable]
Consider the two acyclic neighbourhoods of Figure
\ref{fig:distinguishable}, which correspond to nodes 25 and 30 in
Figure \ref{fig:PSIMI}, respectively. Both neighbourhoods have the
same LPath set $\{ c[expRole], c[expRole].f[expRole] \}$. However,
it is easy to see they are \emph{not} bisimilar: node 33 in
neighbourhood (b) has $c$ and $f$ incoming edges, whereas all
$expRole$ nodes in neighbourhood (a) have either a $c$ or an $f$
edge, but not both. Thus, they are not LPath distinguishable.

In contrast, the three cyclic neighbourhoods of
Figure~\ref{fig:representative} are all bisimilar and have the
same LPath set $\{ c[participantList],
c[participantList].c[participant],  c[participantList].p$
$[interaction] \}$. Therefore, they are all LPath distinguishable.
\qed
\end{example}

We are interested in LPath distinguishable neighbourhoods because
they can be described by EEs. In general, determining whether a
given set of neighbourhoods is LPath distinguishable entails
computing the bisimulation between them and then comparing the
result to their LPath sets.

There is a class of neighbourhoods, however, that are guaranteed
to be always LPath distinguishable. For neighbourhoods in that
class, we can bypass the bisimulation computation and obtain the
EEs directly from the LPaths sets. Such is the class of the
\emph{tree neighbourhoods}. How to characterize other classes of
LPath distinguishable neighbourhoods without resorting to
bisimulation remains an open problem.

We will show below that tree neighbourhoods are in fact LPath
distinguishable (Proposition \ref{pro:bisimilarityAndLPaths}). In
order to do that, we need first some auxiliary results.

\begin{lemma} \label{lem:universalBisimulation}
If two neighbourhoods $\mathcal{N}_1$ and $\mathcal{N}_2$ are
bisimilar then there exists a labeled bisimulation $\approx$ such
that every node in both graphs is in $\approx$. \qed
\end{lemma}

\begin{proof}
By definition, in order for $\mathcal{N}_1$ and $\mathcal{N}_2$ to
be bisimilar $r_1$ and $r_2$ have to be bisimilar, where $r_1$ and
$r_2$ are the roots of $\mathcal{N}_1$ and $\mathcal{N}_2$
respectively. Thus, there has to be a labeled bisimulation
$\approx$ such that $r_1 \approx r_2$. In addition, all nodes in
$\mathcal{N}_1$ connected to $r_1$ by an edge with label $a$ have
to be in the labeled bisimulation with all nodes in
$\mathcal{N}_2$ connected to $r_2$ by an edge with label $a$ (also
by definition). This means that every node connected to either
$r_1$ or $r_2$ by an edge have to belong to $\approx$. Since every
node in $\mathcal{N}_1$ and $\mathcal{N}_2$ is reachable from
$r_1$ and $r_2$ respectively, we can prove inductively that every
node in both $\mathcal{N}_1$ and $\mathcal{N}_2$ belong to
$\approx$. \qed
\end{proof}

\begin{corollary} \label{cor:leavesBisimilar}
For all leaves $v \in \mathcal{N}_1$ and $w \in \mathcal{N}_2$: $v
\sim w$ iff $\lambda(v)=\lambda(w)$. \qed
\end{corollary}

\begin{proof}
By Definition \ref{def:bisimulation} if $v \sim w$ then there
exist a labeled bisimulation $\approx$ between $\mathcal{N}_1$ and
$\mathcal{N}_2$ such that $v \approx w$, which means that
$\lambda(v)=\lambda(w)$. We need to prove now that leaves having
the same label are bisimilar. It is easy to see from Definition
\ref{def:bisimulation} that there always exists a labeled
bisimulation between leaves in $\mathcal{N}_1$ and $\mathcal{N}_2$
when they have the same labels. Consequently, if
$\lambda(v)=\lambda(w)$ then $v \sim w$. \qed
\end{proof}

\begin{proposition} \label{pro:bisimilarityAndLPaths}
Let $\mathcal{N}_1$ and $\mathcal{N}_2$ be tree neighbourhoods in
an axis graph $\mathcal{A}$. Then, $\mathcal{N}_1(v) \sim
\mathcal{N}_2(w)$ iff $LPath(v)=LPath(w)$. \qed
\end{proposition}

\begin{proof}
We proceed by induction on the length of an arbitrary outgoing
path. For the base case, we have that $v$ and $w$ are leaves of
$\mathcal{N}_1$ and $\mathcal{N}_2$ respectively. By Corollary
\ref{cor:leavesBisimilar}, $v \sim w$ iff $\lambda(v)=\lambda(w)$.
Since they are leaves, $LPath(v)=LPath(w)=\emptyset$, so $v \sim
w$ iff $\lambda(v)=\lambda(w)$ and $LPath(v)=LPath(w)$.

For the induction step, consider nodes $v \in \mathcal{N}_1$ and
$w \in \mathcal{N}_2$ and all edges from them with label
``$axis$'': $\langle v, v_i \rangle$, $1 \leq i \leq n$ and
$\langle w, w_j \rangle$, $1 \leq j \leq m$. We know that, if
there is a $v_k$ that is not bisimilar to any $w_j$, i.e., $v_k
\not \sim w_1$, $\ldots$, $v_k \not \sim w_m$, then by Definition
\ref{def:bisimulation} $v \not \sim w$. We need to prove that the
latter statement is equivalent to the following: if there is a
$v_k$ whose label (or $LPath$ set) is different from the label (or
$LPath$ set) of every $w_j$, then $LPath(v) \neq LPath(w)$.

By inductive hypothesis, $v_k \not \sim w_j$ iff $\lambda(v_k)
\neq \lambda(w_j)$ or $LPath(v_k) \neq LPath(w_j)$. Note that,
edges $\langle v, v_i \rangle$ and $\langle w, w_j \rangle$ add
prefixes ``$axis[\lambda(v_i)]$'' and ``$axis[\lambda(w_j)]$'' to
each string in $LPath(v_i)$ and $LPath(w_j)$ respectively. For a
given node $v$, let us call $preLPath(v)$ the set of strings in
$LPath(v)$ prefixed with ``$axis[\lambda(v)]$''. It is easy to see
that, given any two nodes $v_k$ and $w_l$, if the original set of
string are different $(LPath(v_k) \neq LPath(w_l))$, then the
strings with the prefixes are going to be different
$(preLPath(v_k) \neq preLPath(w_l))$, no matter what the prefixes
are. In addition, if $\lambda(v_i) \neq \lambda(w_l)$, we have
that $preLPath(v_k) \neq preLPath(w_l)$ even when $LPath(v_k) =
LPath(w_l)$ (because the label of the nodes are included in the
prefixes).

Since $LPath(v)$ contains all label paths from $v$, in particular
it contains all those that begin with ``$axis$'' ($\bigcup_i
preLPath(v_i)) \subseteq LPath(v)$). Similarly, $\bigcup_j
preLPath(w_j) \subseteq LPath(w)$. However, if there is a $v_k$
such that either its label or its $LPath$ set is different from
those of every $w_j$, then $\bigcup_i preLPath(v_i) \neq \bigcup_j
preLPath(w_j)$. Since all label paths in $LPath(v)$ that are not
in $\bigcup_i preLPath(w_i)$ begin with a prefix different from
``$axis$'', we conclude $\bigcup_i preLPath(v_i) \neq \bigcup_j
preLPath(w_j) \Rightarrow LPath(v) \neq LPath(w)$. \qed
\end{proof}

\paragraph{Notation.} Let $s$ be a node in an SD $\mathcal{D}$ whose extent
contains only LPath distinguishable neighbourhoods. We denote by
$Path(s)$ and $LPath(s)$ the set of all different axis paths and
axis label path, respectively, from the nodes in the extent of
$s$. That is, $LPath(s) = \bigcup_i LPath(v_i)$, $v_i \in
extent(s)$.

When dealing with $LPath$ distinguishable neighbourhoods, the
$LPath$ set can be an alternative way of representing an extent:
just compute the representative neighbourhood
$\mathcal{R}_{\alpha}(s)$ of a given SD node $s$ and then take
$LPath(s)$. However, checking containment and equivalence from the
$LPath$ sets is cumbersome, so we would like to have a way of
obtaining an AxPRE from an $LPath$ set that provides a concise
description of the representative neighbourhood and thus of all
nodes in a given extent. We will denote this new expression
\emph{extent AxPRE}.

\begin{definition} [Extent AxPRE] \label{def:extentAxPRE}
Let $\mathcal{D}$ be an SD, $s$ a node in $\mathcal{D}$ and
$\alpha$ its AxPRE. An \emph{extent AxPRE} $\alpha'$ of $s$ is an
AxPRE such that all nodes in the extent of $s$ have $\alpha'$
neighbourhoods and $\alpha'$ is different from all other extent
AxPREs in $\mathcal{D}$. \qed
\end{definition}

It is important to note that extent AxPREs can only be defined
when representative neighbourhoods are not pairwise in an
inclusion relationship. Because of the prefix semantics we use, if
for any two representative neighbourhoods $\mathcal{R}_{\alpha}$
and $\mathcal{R}'_{\alpha}$ we have that $\mathcal{R}_{\alpha}
\subseteq \mathcal{R}'_{\alpha}$ then any possible AxPRE for
$\mathcal{R}'_{\alpha}$ will also return $\mathcal{R}_{\alpha}$,
and consequently it will not be an ``extent'' AxPRE.

The extent AxPRE of an SD node $s$ can be constructed from the
representative neighbourhood $\mathcal{R}_{\alpha}(s)$ by taking
the label of the root $v$ of $\mathcal{R}_{\alpha}(s)$ and
concatenating it with the disjunction of the axis label paths of
$v$. That is, the extent AxPRE $\alpha'$ of $s$ is one of the
followings:
\begin{itemize}
\item $[\lambda(v)].(lp_1 | lp_2 | \ldots | lp_n)$ if $LPath(v) =
\bigcup_i lp_i$ \item $[\lambda(v)].lp$ if $LPath(v) = \{lp\}$
\item $[\lambda(v)]$ if $LPath(v)=\emptyset$
\end{itemize}

It easy to see from the construction that all nodes in the extent
of $s$ will have $\alpha'$ neighbourhoods.

\begin{figure}
    \centering
        \includegraphics{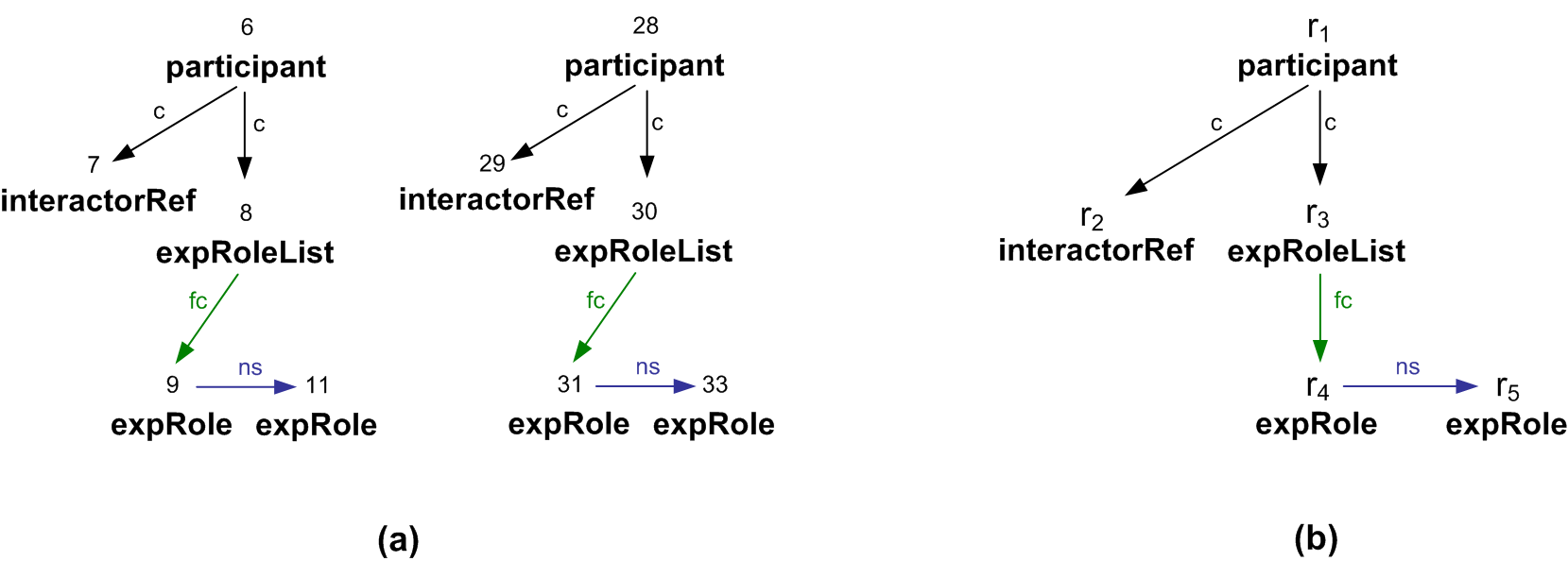}
    \caption{Two $[participant].c.fc.ns^*$ neighbourhoods (a) and their representative neighbourhood (b) from our running example}
    \label{fig:extentAxPRE}
\end{figure}

\begin{example} [Extent AxPRE]
Consider the two neighbourhoods of Figure~\ref{fig:extentAxPRE}
(a) from our running example. They are tree
$[participant].c.fc.ns^*$ neighbourhoods of elements $6$ and $28$,
respectively. In this case, the bisimulation contraction of both
neighbourhoods is an axis graph isomorphic to them and appears in
Figure~\ref{fig:extentAxPRE} (b). Since the label of both nodes
$6$ and $28$ is $participant$, the extent AxPRE begins with the
prefix $[participant]$. In addition, $LPath(6) = \{
c[interactorRef], c[expRoleList],
c[expRoleList].\mathit{fc}[expRole],$
$c[expRoleList].\mathit{fc}[expRole].ns[expRole] \} = LPath(28)$,
which means that the AxPRE contains a conjunction of four
subAxPREs, resulting in $[participant].(c[interactorRef] | c$
$[expRoleList] | c[expRoleList].\mathit{fc}[expRole] |
c[expRoleList].\mathit{fc}[expRole].ns[expRole]).$ \qed
\end{example}

According to Definition \ref{def:stability} (summary axis
stability), forward-stable edges provide stronger information on
the axis relationship that nodes in their extents satisfy: from a
forward-stable edge $\langle s_i, s_j \rangle$ labeled $axis$, we
know that all nodes in the extent of $s_i$ are related by $axis$
to some nodes in the extent of $s_j$. Thus, we are particulary
interested in neighbourhoods in which all edges are forward-stable
for their descriptive capabilities.

\begin{definition} [Forward-stable Neighbourhood] \label{def:stableNeigh}
A \emph{forward-stable neighbourhood} of an SD node $s$ is a
neighbourhood of $s$ with all its edges forward-stable. \qed
\end{definition}

An AxPRE always describes some neighbourhood in an axis graph,
either of an instance or an SD. When an AxPRE describes a
forward-stable neighbourhood in the SD graph, it is called a
\emph{neighbourhood AxPRE}. If all edges in the $\alpha$
neighbourhood of SD node $s$ are forward-stable, the extent AxPRE
of $s$ can be computed from them rather than from the axis graph
of the instance.

\begin{figure}
    \centering
        \includegraphics{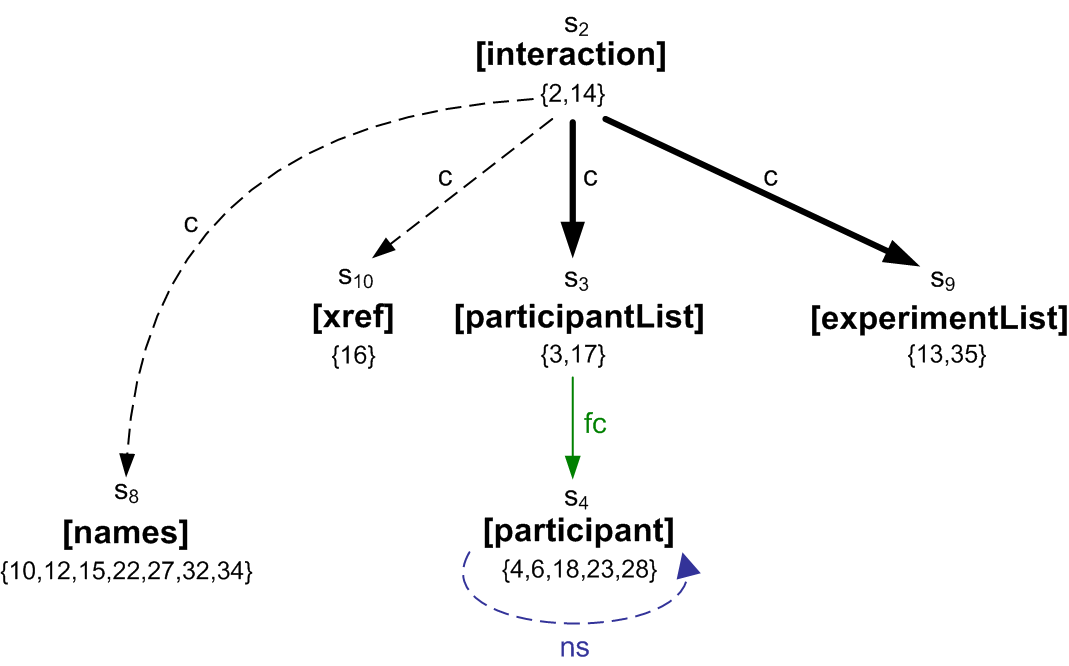}
    \caption{The $c.fc.ns^*$ neighbourhood of node $s_2$ of Figure \ref{fig:PSIMI-label}}
    \label{fig:neighbourhood}
\end{figure}

\begin{example} [Neighbourhood AxPRE] \label{exa:detailed}
Consider node $s_2$ in Figure \ref{fig:neighbourhood}. Its current
AxPRE is $[interaction]$, which means that its extent contains
only interaction elements. We can infer from the SD graph an
neighbourhood AxPRE as follows. Since edges $\langle s_2, s_3
\rangle$, $\langle s_2, s_9 \rangle$, and $\langle s_3, s_4
\rangle$ are forward-stable, we could write an AxPRE that
expresses those relations, which is $
[interaction].(c[participantList].\mathit{fc}[participant]|c[experimentList]).$
Such an AxPRE tells us that not only the extent of $s_2$ contains
interaction elements, but more precisely they also have nested
elements such as a participantList with a nested participant, and
an experimentList. \qed
\end{example}

\section{Refinement and stabilization}
\label{sec:NeighbourhoodStabilization}

The description provided by a node in the SD can be changed by an
operation that modifies its AxPRE and thus its AxPRE
neighbourhood. This operation is called a \emph{refinement} of an
SD node. The refinement of an SD node can be computed directly by
changing the AxPRE of the node (Algorithm \ref{alg:refineNode}) or
by stabilizing a summary neighbourhood for a given AxPRE
(Algorithm \ref{alg:n-stabilize}). Note that Algorithm
\ref{alg:refineNode} in fact changes one of the AxPREs in the
definition of the SD, so all nodes that share the modified AxPRE
will be affected.

\begin{figure}
\begin{algorithm}
\label{alg:refineNode}

    \mbox{}\\
    \parbox{\textwidth}{
    \noindent{\sl \underline{refineNode$(D,s,\alpha)$}}\\[0.5em]
    \noindent{\bf Input:} An SD $D$, a node $s$ in $D$, and an AxPRE $\alpha \subseteq axpre(s)$ \\
    \noindent{\bf Output:} An SD $D$ where $s$ has been refined by $\alpha$

    \begin{algorithmic} [1]

    \FOR {every $v$ in $extent(s)$}
        \STATE $candidate := \emptyset$
        \STATE compute the $\alpha$ neighbourhood of $v$: $\mathcal{N}_{\alpha}(v)$
        \FOR {every node $s$ in $D$ such that $axpre(s) := \alpha$ }
            \STATE let $w$ be a node in $extent(s)$
            \STATE compute the $\alpha$ neighbourhood of $w$: $\mathcal{N}_{\alpha}(w)$
            \IF {$v \sim^\alpha w$ (i.e., $\mathcal{N}_{\alpha}(v) \sim \mathcal{N}_{\alpha}(w)$) }
                \STATE $candidate := s$
            \ENDIF
        \ENDFOR
        \IF {$candidate = \emptyset$}
            \STATE create a new node $candidate$ in $D$
            \STATE $axpre(candidate) := \alpha$
            \STATE $\lambda^D(candidate) := \lambda(v)$
        \ENDIF
        \STATE move $v$ from $extent(s)$ to $extent(candidate)$
    \ENDFOR
    \STATE let $S$ be the set of nodes connected to $s$
    \FOR {every node $s'$ in $S$}
        \STATE add edges $\langle candidate, s' \rangle$ and $\langle s', candidate \rangle$ if conditions in Definition \ref{def:stability} are satisfied
    \ENDFOR
    \STATE delete $s$ and all its incoming and outgoing edges from $D$

    \end{algorithmic}
}

 \end{algorithm}
 \end{figure}

\begin{figure}
\begin{algorithm}
\label{alg:edgeStabilization}

    \mbox{}\\
    \parbox{\textwidth}{
    \noindent{\sl \underline{stabilizeEdge$(D,s_i,s_j)$}}\\[0.5em]
    \noindent{\bf Input:} An SD $D$ containing  a non forward-stable edge
$e=\langle s_i, s_j \rangle$ with label $axis$\\
    \noindent{\bf Output:} An SD $D$ where $e$ has been replaced
    by forward-stable $e'=\langle s_i', s_j \rangle$

    \begin{algorithmic} [1]

    \STATE $\alpha:=axpre(s_i) | axis \; axpre(s_j)$
    \FOR {every node $s$ in $D$ such that $axpre(s)=axpre(s_i)$}
        \STATE refineNode$(D,s,\alpha)$
    \ENDFOR

    \end{algorithmic}
}

 \end{algorithm}
 \end{figure}

\begin{figure}
\begin{algorithm}
\label{alg:axisStabilization}

    \mbox{}\\
    \parbox{\textwidth}{
    \noindent{\sl \underline{stabilizeAxis$(D,s_i,axis)$}}\\[0.5em]
    \noindent{\bf Input:} An SD $D$ containing a non forward-stable edge from $s_i$ with label $axis$\\
    \noindent{\bf Output:} An SD $D$ where all $axis$ edges from $s_i$ are forward-stable

    \begin{algorithmic} [1]

    \STATE $\alpha:=axpre(s_i) | axis$
    \FOR {every node $s$ in $D$ such that $axpre(s)=axpre(s_i)$}
        \STATE refineNode$(D,s,\alpha)$
    \ENDFOR

    \end{algorithmic}
}

 \end{algorithm}
 \end{figure}

\begin{figure}
\begin{algorithm}
\label{alg:edgeUnfolding}
    \mbox{}\\
    \parbox{\textwidth}{
    \noindent{\sl \underline{unfoldEdge$(D,s_i,axis)$}}\\[0.5em]
    \noindent{\bf Input:} An SD $D,$ a node $s_i$ such that there
    exists a non forward-stable $e=\langle s_i, s_i \rangle$ with label $axis$\\
    \noindent{\bf Output:} The SD $D$ where any edge $e=\langle s_i, s_i \rangle$ with label
    $axis$ is forward-stable

    \begin{algorithmic} [1]

    \STATE $\alpha:=axpre(s_i) | axis^*$
    \FOR {every node $s$ in $D$ such that $axpre(s)=axpre(s_i)$}
        \STATE refineNode$(D,s,\alpha)$
    \ENDFOR

    \end{algorithmic}
}
 \end{algorithm}
\end{figure}

\vspace{.2in}

Previous proposals perform global refinements on the entire SD
graph \cite{KBNK02,KSBG02} or local refinements based on
statistics or workload \cite{QLO03,HY04,PG06b}, without the
ability to refine a declaratively defined neighbourhood. In
contrast, using \dx\ we can precisely characterize the
neighbourhood considered for the refinement with an AxPRE.

\vspace{.2in}

\dx\ refinements can also be based on the notion of summary axis
stability (Definition~\ref{def:stability}). The goal of this
particular refinement operation is to make all edges of a
neighbourhood, given by an AxPRE in the SD graph, forward-stable.
Edges can be stabilized one at a time or by groups with the same
axis. For the former approach, \dx\ implements two different
strategies. If the edge links two different nodes, then
Algorithm~\ref{alg:edgeStabilization} is invoked. In contrast, if
the edge forms a loop, then Algorithm~\ref{alg:edgeUnfolding} is
used. For stabilizing a group of edges with the same axis from a
given node, \dx\ invokes Algorithm~\ref{alg:axisStabilization}.
All algorithms mentioned above reduce edge stabilization to
refinement: step 1 in each algorithm composes a new AxPRE and step
3 refines the affected nodes by calling
Algorithm~\ref{alg:refineNode}.

\vspace{.2in}

The next two examples illustrate how a non forward-stable edge is
stabilized by Algorithms~\ref{alg:edgeStabilization} and
\ref{alg:edgeUnfolding}, respectively.

\begin{example} [Edge Stabilization] \label{exa:stable}
Consider edge $\langle s_4, s_6 \rangle$ from Figure
\ref{fig:stable} (a). This edge is not forward-stable because
elements 4 is not related to any node in $extent(s_6)$ via the $c$
axis (i.e. there is no $c$ edge from 4 to a expRoleList element in
Figure \ref{fig:PSIMI}). Edge stabilization (Algorithm
\ref{alg:edgeStabilization}) creates two nodes, $s_{41}$ and
$s_{42}$, such that $extent(s_{41}) = \{ 4 \}$ and $extent(s_{42})
= \{ 6 , 18 , 23 , 28 \}$. Since $axpre(s_4)=[participant]$ and
$axpre(s_6)=[expRoleList]$ (the original AxPREs), then line 1 of
Algorithm~\ref{alg:edgeStabilization} creates the new AxPRE
$[participant]|c[expRoleList]$, which will be used to refine all
nodes with $[participant]$ AxPRE (lines 2 and 3). The new edge
$\langle s_{41}, s_5 \rangle$ is forward-stable. The result of
stabilizing edge $\langle s_4, s_6 \rangle$ is shown in Figure
\ref{fig:stable} (b). \qed
\end{example}


\begin{example} [Edge Unfolding] \label{exa:unfold}
Consider the $ns$ loop on node $s_{42}$ from Figure
\ref{fig:stable} (b). The edge is not forward-stable because some
element in $extent(s_{42})$ is not in a $ns$ relation with
elements in the same extent (for instance, there is no element
that is the next sibling of 28 in Figure \ref{fig:PSIMI}). Since
$axpre(s_{42})=[participant]|c[expRoleList]$ (the result of the
stabilization performed in Example \ref{exa:stable}), then line 1
of Algorithm~\ref{alg:edgeUnfolding} creates the new AxPRE
$[participant]|c[expRoleList]|ns^*$, which will be used to refine
all nodes with $[participant]|$ $c[expRoleList]$ AxPRE (lines 2
and 3). The new edges are forward-stable. The result of unfolding
$ns$ loop on $s_{42}$ is shown in Figure~\ref{fig:stable2}. \qed
\end{example}

\begin{figure}
    \centering
        \includegraphics{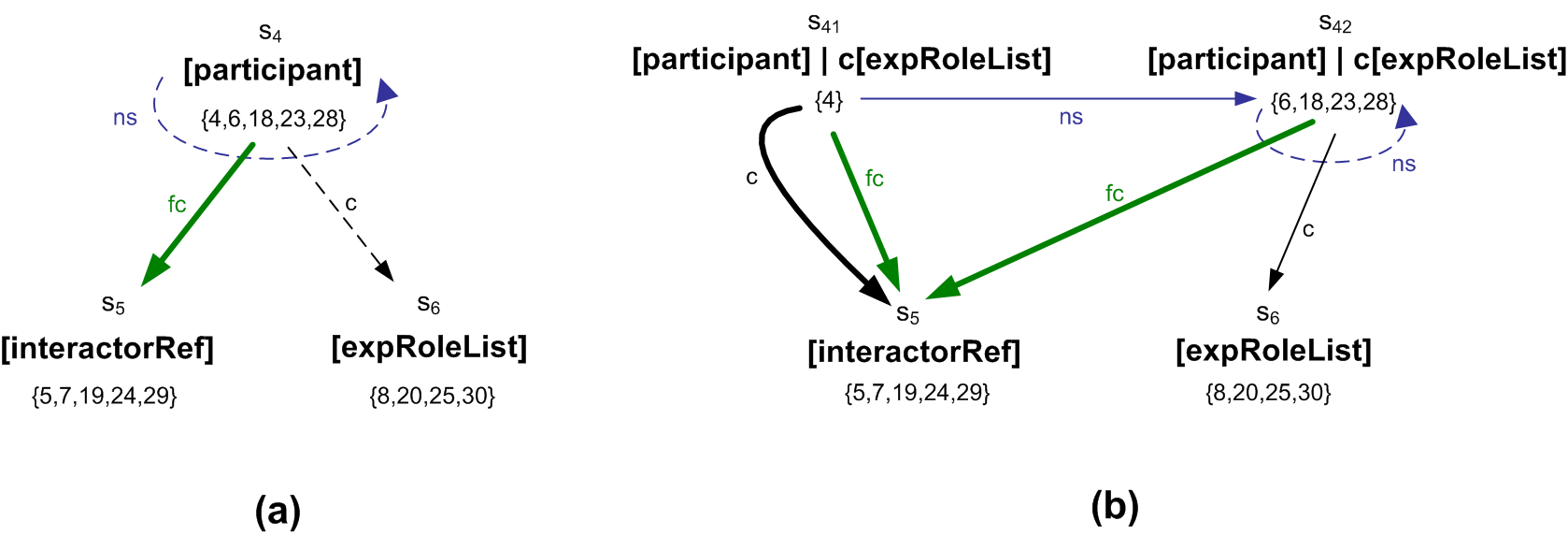}
    \caption{The $\mathit{fc}|c|ns$ neighbourhood of $s_4$ from Figure \ref{fig:PSIMI-label}: (a) before stabilizing $c$ edge to $s_6$, (b) after stabilization}
    \label{fig:stable}
\end{figure}

\begin{figure}
    \centering
        \includegraphics{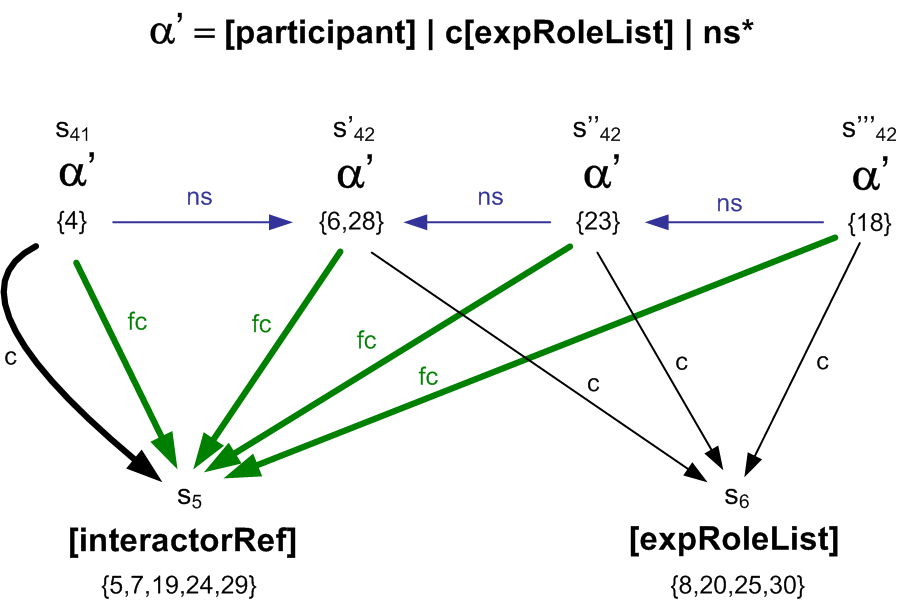}
    \caption{The neighbourhood from Figure \ref{fig:stable} (b) after stabilizing $ns$ loop on $s_{42}$}
    \label{fig:stable2}
\end{figure}

\begin{figure}
\begin{algorithm}
\label{alg:n-stabilize}
    \mbox{}\\
    \parbox{\textwidth}{
    \noindent{\sl \underline{StabilizeNeighbourhood$(D,\alpha,s)$}}\\[0.5em]
    \noindent{\bf Input:} An SD $D,$ an AxPRE $\alpha,$ and a node $s$\\
    \noindent{\bf Output:} An SD $D$ where all the edges in the $\alpha$ neighborhood of
    $s$ are forward-stable

    \begin{algorithmic} [1]

    \STATE compute the $\alpha$ neighbourhood of $s$
    \STATE $S=\{s' \, | \, s'$ is in the  $\alpha$ neighbourhood of $s \}$
    \WHILE {$S \neq \emptyset$}
        \STATE pick a node $s'$ in $S$ such that $s'$ is at the end of the longest simple path from $s$
        \FOR {each edge $e=\langle s', s' \rangle$}
            \STATE unfoldEdge$(D,s',axis)$
        \ENDFOR
        \FOR {each edge $e=\langle s'', s' \rangle$}
            \STATE stabilizeEdge$(D,s',s'')$
        \ENDFOR
        \STATE remove $s'$ from $S$
    \ENDWHILE

    \end{algorithmic}
 }

\end{algorithm}
\end{figure}

We have now all the building blocks for introducing the
neighbourhood stabilization algorithm, Algorithm
\ref{alg:n-stabilize}, which computes a refinement of the extent
of an SD node $s$ for an AxPRE $\alpha$ that results in a stable
$\alpha$ neighbourhood of $s$. Given an SD node $s$ and an AxPRE
$\alpha$, Algorithm \ref{alg:n-stabilize} computes an AxPRE
partition of the extent of $s$ for $\alpha$ that is a refinement
of the extent of $s$. This is achieved by stabilizing the $\alpha$
neighbourhood of $s$. In order to stabilize a single edge,
Algorithm~\ref{alg:n-stabilize} invokes
Algorithm~\ref{alg:edgeStabilization}, for different nodes, and
Algorithm~\ref{alg:edgeUnfolding}, for the same node (loop).
Algorithm~\ref{alg:axisStabilization} is a variation of
Algorithm~\ref{alg:edgeStabilization} in which all edges labeled
with the same axis are stabilized.

Most of the execution of the neighbourhood stabilization algorithm
is covered by Examples~\ref{exa:stable} and \ref{exa:unfold}. For
instance, if we want to stabilize the $[participant].(c|fc|ns^*)$
neighbourhood of node $s_4$ in Figure~\ref{fig:stable} (a), then
Algorithm~\ref{alg:n-stabilize} stabilizes edge $\langle s_4, s_6
\rangle$, as described in Example \ref{exa:stable}, and unfolds
edge $\langle s_7, s_7 \rangle$ labeled $ns$, as described in
Example \ref{exa:unfold}. The resulting stable
$[participant].(c|fc|ns^*)$ neighbourhood is shown in
Figure~\ref{fig:stable2}.

\vspace{.2in}

In this chapter, we discussed how an SD description can be changed
by operations that modify its AxPREs and thus the AxPRE
neighbourhoods of the nodes in their extents. We introduced the
two basic \dx\ operations, AxPRE refinement and stabilization, and
provided algorithms for them. We also gave, for LPath
distinguishable neighbourhoods, a characterization of the extent
of an SD node with an EE. In the next chapter, we discuss the
XPath syntax and data model, together with a novel mechanism for
expressing EEs in XPath.

\chapter{Changing descriptions with XPath}
\label{sec:changingSDswithxpath}

We have discussed how to characterize SD nodes and their extents
using different approaches based on neighbourhoods, sets of axis
label paths, and AxPREs. In this chapter, we propose a novel
mechanism to characterize an SD node with an XPath expression
\cite{W3C:XPath/02} whose evaluation returns exactly the elements
in the extent. This expression, which effectively represents the
extent of a given SD node $s$, is called \emph{extent expression}
(EE) and is denoted $ee(s)$.

In \dx, the extents of any SD node can be precomputed and stored
in a data structure. This approach, which we call
\emph{materialized extents}, uses a pointer to every XML element
in the collection and therefore it can be very space consuming.
Since the evaluation of an $ee(s)$ of a node $s$ returns the
actual extents of $s$, a more space-efficient approach is to keep
only the EEs. These \emph{virtual extents} are a compact
representation of the extents, similar to the concept of virtual
views.

Since EEs are expressed in XPath, we give first an introduction to
the XPath syntax and data model. The formal semantics definition
of the full language is provided for completeness in Appendix
\ref{Section:XPathLanguage}.

\section{XPath syntax and data model}

XPath is a compositional language for selecting element nodes in
XML documents. It is also the dialect that most XML manipulation
languages (e.g., XSLT\footnote{\small \tt{
http://www.w3.org/TR/xslt}}, XPointer\footnote{\small \tt{
http://www.w3.org/TR/xptr/}}, XQuery\footnote{\small \tt{
http://www.w3.org/TR/xquery/}}, etc.) have in common. In this
section we introduce the language expression grammar and its data
model based on axes.

\begin{definition}[XPath Expression Grammar] \label{def:expressions}
Let $e, e_1 \ldots e_m$ be expressions, $locpath$, $locpath_1$,
$\ldots$, $locpath_m$ be location paths, $l$ be a node name from
the label alphabet $Label$ of the axis graph, $axis$ be a relation
in $Axes$, and $op$ be a place holder for any of the XPath
functions and operators such as $+, -, *, div, =, \neq, \leq, <,
\geq$ and $>$, as well as for context accessing functions
$position()$ and $last()$. The following is the grammar for XPath
1.0 expressions:
\[ e := disj \; | \; op(e_1, \ldots, e_m) \]
\[ disj := locpath_1 \; | \ldots | \; locpath_m \]
\[ locpath := par \; | \; comp \; | \; abs \; | \; step \]
\[ par := ( \; disj \; ) \; [ e_1 ] \ldots [ e_m ] \]
\[ comp := locpath_1 \; / \; locpath_2  \]
\[ abs := / \; locpath \]
\[ step := axis \; :: \; l \; [ e_1 ] \; \ldots \; [ e_m ] \] \qed
\end{definition}

The XPath data model includes atomic values, sequences, and a
predefined set of axes for navigating the instance. Like in an
axis graph, which is an abstract representation of the XPath data
model, axes define relationships between nodes in the instance. We
provide next a definition of the XPath axes in terms of
$firstchild$, $nextsibling$, their inverses and $\mathit{self}$.

\begin{definition} [XPath Axes] \label{def:XPathAxes}
Given an axis graph $\mathcal{A} = (\mathit{Inst},$ $Axes,$
$Label,$ $\lambda)$, the XPath axes in $\mathcal{A}$ are defined
as follows:

\begin{itemize}
\item $\mathit{self}:= \{\langle v,v \rangle \; | \; v \in Inst
\}$ \item $child := firstchild.nextsibling^*$ \item $parent :=
(nextsibling^{-1})^*.firstchild^{-1}$ \item $descendant :=
firstchild.(firstchild \bigcup nextsibling)^*$ \item $ancestor :=
(firstchild^{-1} \bigcup nextsibling^{-1})^*.firstchild^{-1}$
\item $descendant$-$or$-$\mathit{self} := descendant \bigcup
\mathit{self} $ \item $ancestor$-$or$-$\mathit{self} := ancestor
\bigcup \mathit{self} $ \item $following :=
ancestor$-$or$-$\mathit{self}.nextsibling.nextsibling^*.descendant$-$or$-$\mathit{self}$
\item $preceding :=
ancestor$-$or$-$\mathit{self}.nextsibling^{-1}.(nextsibling^{-1})^*.descendant$-$or$-$\mathit{self}$
\item $following$-$sibling := nextsibling.nextsibling^*$ \item
$preceding$-$sibling := (nextsibling^{-1})^*.nextsibling^{-1}$
\qed
\end{itemize}

Whenever it is clear from the context, we use $s$, $c$, $p$, $d$,
$a$, $ds$, $as$, $f$, $pc$, $\mathit{fs}$ and $ps$ as
abbreviations of $\mathit{self}$, $child$, $parent$, $descendant$,
$ancestor$, $descendant$-$or$-$\mathit{self}$,
$ancestor$-$or$-$\mathit{self}$, $following$, $preceding$,
$following$-$sibling$ and $preceding$-$sibling$, respectively.
\qed

\end{definition}

Note that the $\mathit{self}$, $ancestor$, $descendant$,
$preceding$, and $following$ axes from a given node $v$ partition
the nodes in the XML tree. This is represented graphically by the
schema in Figure \ref{fig:partition}.

\begin{figure}
    \centering
    \includegraphics{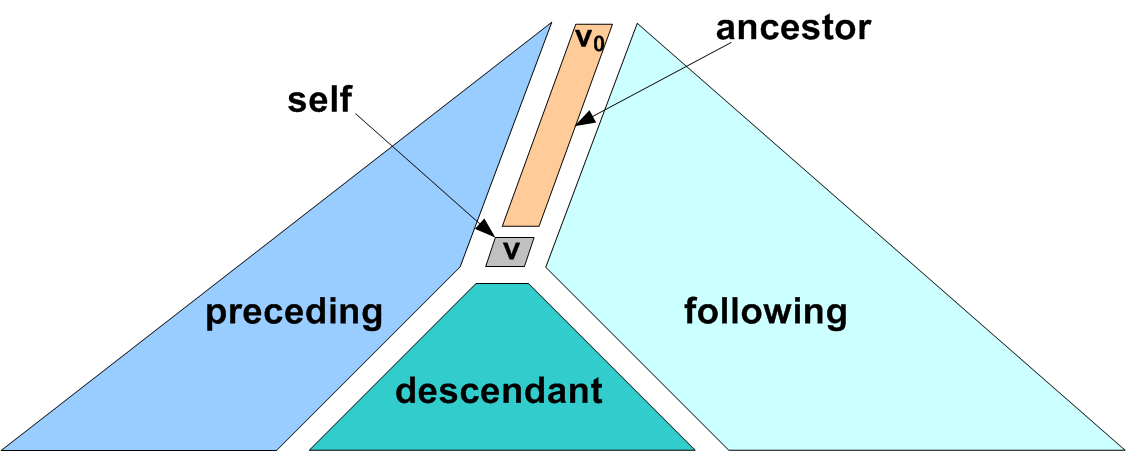}
    \caption{Partition of the nodes in the XML tree by axis relations}
\label{fig:partition}
\end{figure}

Since XML documents are ordered, we need to define a
\emph{document order} relation on the nodes of an axis graph
$\mathcal{A}$.

\begin{definition} [Document Order] \label{def:doc_order}
The \emph{document order} relation $\prec_{doc}$ on an axis graph
$\mathcal{A} = (\mathit{Inst},$ $Axes,$ $Label,$ $\lambda)$ is the
total order relation given by $d \bigcup f$, where $d$ and $f$ are
the XPath axes in $Axes$ from Definition~\ref{def:XPathAxes}. \qed
\end{definition}

Based on the document order relation and its inverse we define
next \emph{axis order} and \emph{axis position}.

\begin{definition}[Axis Order] \label{def:order}
Let axis graph $\mathcal{A} = (\mathit{Inst},$ $Axes,$ $Label,$
$\lambda)$ be an axis graph. We define the binary \emph{axis
order} relation $\prec_{axis}$ in $\mathit{Inst} \times
\mathit{Inst}$ as $\prec_{doc}$ if $axis \in \{s$, $c$, $d$, $ds$,
$f$, $fs \}$ and as $\prec_{doc}^{-1}$ otherwise. \qed
\end{definition}

Having introduced the XPath syntax and data model, we discuss next
how descriptions are changed in \dx\ using XPath.

\section{Refinement with XPath} \label{sec:xpathRefinement}

Whenever the representative neighbourhood of an SD node $s$ is
LPath distinguishable, it is possible to precisely characterize
the extent of $s$ in terms of the axis label paths in its $LPath$
set (see Chapter \ref{sec:acyclic}). For this class of
neighbourhood, nodes with the same LPath set are bisimilar
(Proposition \ref{pro:bisimilarityAndLPaths}). Therefore, we
propose a mechanism capable of computing the extent of $s$ based
on its $LPath$ set.

First, we need a few auxiliary results that show how an axis label
path in a given LPath set can be captured by a single XPath
expression. We will show later how to derive EEs from these axis
label path expressions. In order to prove our results, we use the
XPath formal semantics given in Appendix
\ref{Section:XPathLanguage}.

\begin{lemma} \label{thm:expressionLabelPath}
Let $e$ be an XPath expression of the form $axis_1::l_1/ \ldots
/axis_n::l_n$. If $\mathcal{D}[\![e]\!](v)\neq \emptyset$ then
there exists an axis label path $lp = axis_1[l_1]. \ldots
.axis_n[l_n]$ from $v$. \qed
\end{lemma}

\begin{proof}
If $\mathcal{D}[\![e]\!](v)\neq \emptyset$, then by semantic rule
(\ref{sem:locPath}) in Figure \ref{sem:expressions} there must
exist  $v_1, \ldots, v_n$ such that $\langle v, v_1 \rangle \in
axis_1 \wedge$ $\langle v_1, v_2 \rangle \in axis_2 \wedge \ldots
\wedge$ $\langle v_{n-1}, v_n \rangle \in axis_n,$ and
$\lambda(v_i) = l_i$, $1 \leq i \leq n$. This means that there is
a path from $v$ to $v_n$ going through edges $axis_1, \ldots,
axis_n$ and nodes $v_1, \ldots, v_n$ such that $axis_1[l_1].
\ldots .axis_n[l_n]$ is its axis label path. \qed
\end{proof}

\begin{lemma} \label{thm:notExpressionLabelPath}
Let $e$ be an XPath expression of the form $axis_1::l_1/ \ldots
/axis_n::l_n$. If $\mathcal{D}[\![e]\!](v) = \emptyset$ then there
is no axis label path $lp= axis_1[l_1]. \ldots .axis_n[l_n]$ from
$v$. \qed
\end{lemma}

\begin{proof}
If $\mathcal{D}[\![e]\!](v) = \emptyset$, then by semantic rule
(\ref{sem:locPath}) in Figure \ref{sem:expressions} there are no
$w_1, \ldots, w_m$ such that $\langle v, w_1 \rangle \in axis_1,
\langle w_1, w_2 \rangle \in axis_2, \ldots, \langle w_{m-1}, w_m
\rangle \in axis_m$ and $\lambda(w_i) = l_i$, $1 \leq i \leq m$.
This means that there is no path from $v$ to $w_n$ going through
edges $axis_1, \ldots, axis_m$ and nodes $w_1, \ldots, w_m$, and
thus there is no axis label path $lp = axis_1[l_1]. \ldots
.axis_m[l_m]$ from $v$. \qed
\end{proof}

Consider SD node $s$ with AxPRE $\alpha$. In order to compute the
extent of $s$ we need to get all nodes that have the same $LPath$
set and label as $s$. Therefore, we need to write an XPath
expression $e_i$ as defined in Lemma \ref{thm:expressionLabelPath}
for each different axis label path $lp_i$ in $LPath$. Then, all
$e_i$ expressions have to be combined in one EE as follows:
$exp=/ds::\lambda(s)[e_1]\ldots[e_n]$. However, such an expression
does not guarantee that the returned nodes have the exact $\{
lp_1, \ldots, lp_n \}$ LPath set: it only guarantees containment.
That is, $exp$ will return all nodes $v$ such that $LPath(v)
\supseteq \{ lp_1, \ldots, lp_n \}$. The reason for that is that
$exp$ says that all $[e_1]\ldots[e_n]$ have to be satisfied, but
it does not say they have to be the only ones, which would be
required for equality. The way of circumventing this problem is by
explicitly adding a $[not(e_i)]$ predicate for each $lp_i$ that is
not in $LPath(v)$.

The problem with this approach is that it would require the
explicit negation of a large number of axis label paths. However,
we can drastically reduce that number by considering only SD nodes
that have an AxPRE $\alpha'$ such that $\alpha' \subseteq \alpha$.
The intuition is that, if two AxPREs of SD nodes $s_1$ and $s_2$
are not in a containment relationship, then nodes in their extents
cannot have $LPath$ sets in a containment relationship either and
we do not need to have a $not()$ predicate for them. The following
example illustrates how an EE is composed from axis label paths
expressions and $not()$ predicates.

\begin{example} [Extent Expressions] \label{ex:extentExpressions}
Consider SD nodes $s_{41}$, $s_{42}$, and $s_{43}$ from Figure
\ref{fig:SDNeighbourhood2}. For the EE of $s_{41}$, we need all
axis label paths that are in $LPath(s_{41})$ but not in
$LPath(s_{42}) \cup LPath(s_{43})$. The required LPath sets are
the following: $LPath(s_{41}) = \{ c[interactorRef] \}$,
$LPath(s_{43}) =
\{c[interactorRef],c[expRoleList].\mathit{fc}[expRole] \}$, and
$LPath(s_{42}) = \{
c[interactorRef],c[expRoleList].\mathit{fc}[expRole].ns[expRole]\}$.
The final EE will have a positive predicate for each string in
$LPath(s_{41})$ and a negative one for each string in
$(LPath(s_{42}) \cup LPath(s_{43})) - LPath(s_{41})$. The
resulting expression is $ee(s_{41}) =
/ds::participant[c::interactorRef][not(c::expRoleList/c::*[1][s::expRole])]
[not(c::expRoleList/c::*[1][s::expRole]/\mathit{fs}::*[1][expRole])]$,
where $c::*[1][s::expRole]$ and $\mathit{fs}::*[1][expRole]$ are
the XPath expressions of $\mathit{fc}[expRole]$ and $ns[expRole]$,
respectively. \qed
\end{example}

\begin{figure}
    \centering
        \includegraphics{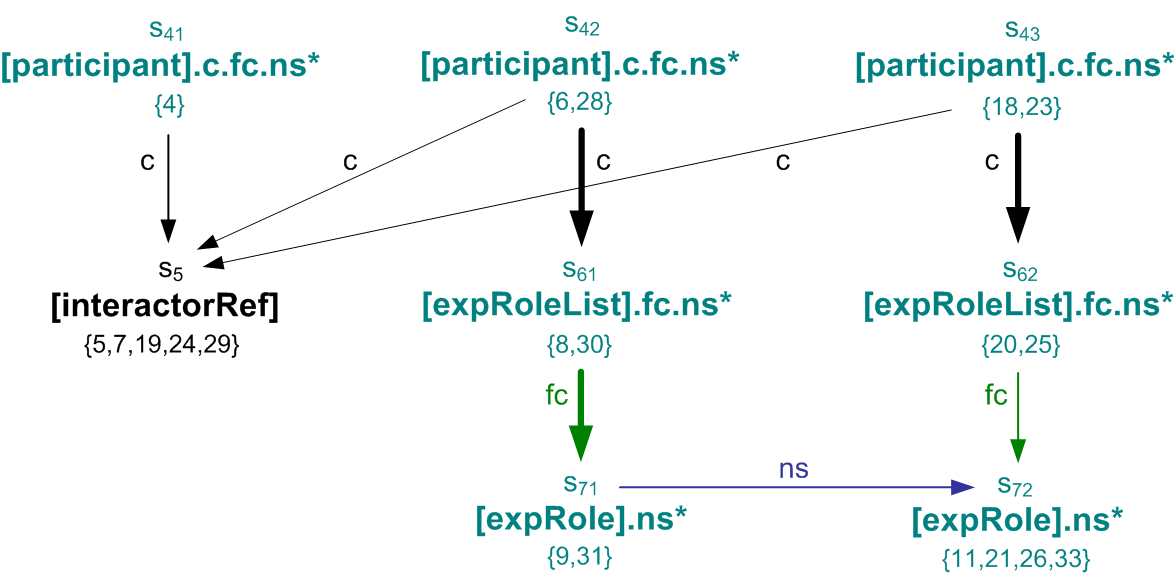}
    \caption{The $[participant].c.\mathit{fc}.ns^*$ neighbourhood from Figure \ref{fig:PSIMI-axpre}}
    \label{fig:SDNeighbourhood2}
\end{figure}

Note that the EEs resulting from this approach might have
redundant predicates that can be simplified. Consider Example
\ref{ex:extentExpressions} for instance: if a node does not
exists, neither does a following sibling for that node, then the
last predicate for $ee(s_{41})$ can be removed safely. There are
many other useful simplifications that can be applied to EEs, but
a broad theory of equivalence is beyond the scope of this thesis.

\vspace{.2in}

The next proposition shows that the EEs thus constructed return
all nodes that do have the axis label paths specified in the
positive predicates and do not have those in the negative
predicates.

\vspace{.1in}

\begin{proposition} \label{pro:extentXPath}
Let $\mathcal{A}$ be an axis graph and $e_s=ds::l[e_1] \ldots
[e_n][not(e_{n+1})] \ldots [not(e_m)]$ an XPath expression where
$e_i=axis_{i_1}::l_{i_1}/ \ldots /axis_{i_{k_i}}::l_{i_{k_i}}$, $1
\leq i \leq m$. Then, $e_s$ returns all nodes $v$ such that there
exists $lp_1, \ldots, lp_n$ axis label paths from $v$ and there
are no $lp_{n+1}, \ldots, lp_m$ axis label paths from $v$, where
$lp_i= axis_{i_1}[l_{i_1}]. \ldots .axis_{i_{k_i}}[l_{i_{k_i}}].$
\qed
\end{proposition}

\vspace{.1in}

\begin{proof}
\begin{tabbing}
\hspace{0in} \= \hspace{.1in} \= \hspace{.1in} \= \kill \\
\> \>   $\mathcal{D}[\![ds::l/[e_1] \ldots
[e_n][not(e_1')] \ldots [not(e_m')] ]\!](v_0) $\\
\> = \> \> by semantic rules (\ref{sem:union}) and (\ref{sem:axis}) in Figure \ref{sem:locationPaths} \\
\> \>   $ \{v \; | \; \lambda(v)=l \; \wedge \; \langle v_0,v \rangle \in ds \; \bigwedge^n_{i=1}\mathcal{E}[\![e_i]\!](v,pos_{ds}(v,S),|S|) = true \; $ \\
\> \> \>  \`  $ \bigwedge^m_{i=n+1}\mathcal{E}[\![not(e_i)]\!](v,pos_{ds}(v,S),|S|) = true \} $  \\
\> = \> \> since all nodes $v$ are reachable from $v_0$, $\langle v_0,v \rangle \in ds$ is always true \\
\> \>   $ \{v \; | \; \lambda(v)=l \; \bigwedge^n_{i=1}\mathcal{E}[\![e_i]\!](v,pos_{ds}(v,S),|S|) = true \; $ \\
\> \> \>  \`  $ \bigwedge^m_{i=n+1}\mathcal{E}[\![not(e_i)]\!](v,pos_{ds}(v,S),|S|) = true \} $  \\
\> = \> \> by semantic rule (\ref{sem:op}) in Figure \ref{sem:expressions} with $Op=not(boolean(S)$ from Figure \ref{fig:BasicOperators} \\
\> \>   $ \{v \; | \; \lambda(v)=l \; \bigwedge^n_{i=1}\mathcal{E}[\![e_i]\!](v,pos_{ds}(v,S),|S|) = true \; $ \\
\> \> \>  \`  $ \bigwedge^m_{i=n+1}\mathcal{E}[\![e_i]\!](v,pos_{ds}(v,S),|S|) = false \} $ \\
\> = \> \> by semantic rule (\ref{sem:locPath}) in Figure \ref{sem:expressions}  \\
\> \>   $ \{v \; | \; \lambda(v)=l \; \bigwedge^n_{i=1}\mathcal{D}[\![e_i]\!](v) = true \; \bigwedge^m_{i=n+1}\mathcal{D}[\![e_i]\!](v) = false \} $ \\
\> = \> \> by Lemmas \ref{thm:expressionLabelPath} and \ref{thm:notExpressionLabelPath}  \\
\> \>   $ \{v \; | \; \lambda(v)=l \; \wedge \; \exists lp_1, \ldots, \exists lp_n \; \wedge \; \not \exists lp_{n+1}, \ldots, \not \exists lp_m \} $ \\
\qed
\end{tabbing}
\end{proof}

In some special cases, a more compact XPath expression can be
obtained. For instance, for an expression containing the closure
of an axis, like $c^*$, we can enforce that the $lp_i$'s expressed
by the EE are the only ones by using the $count()$ XPath function.
Since the XPath expression of each $lp_i$ for $c^*$ contains only
compositions of the $child$ axis, the set of nodes reached by all
$lp_i$'s and all their substrings are exactly all the descendants.

\begin{figure}
    \centering
        \includegraphics{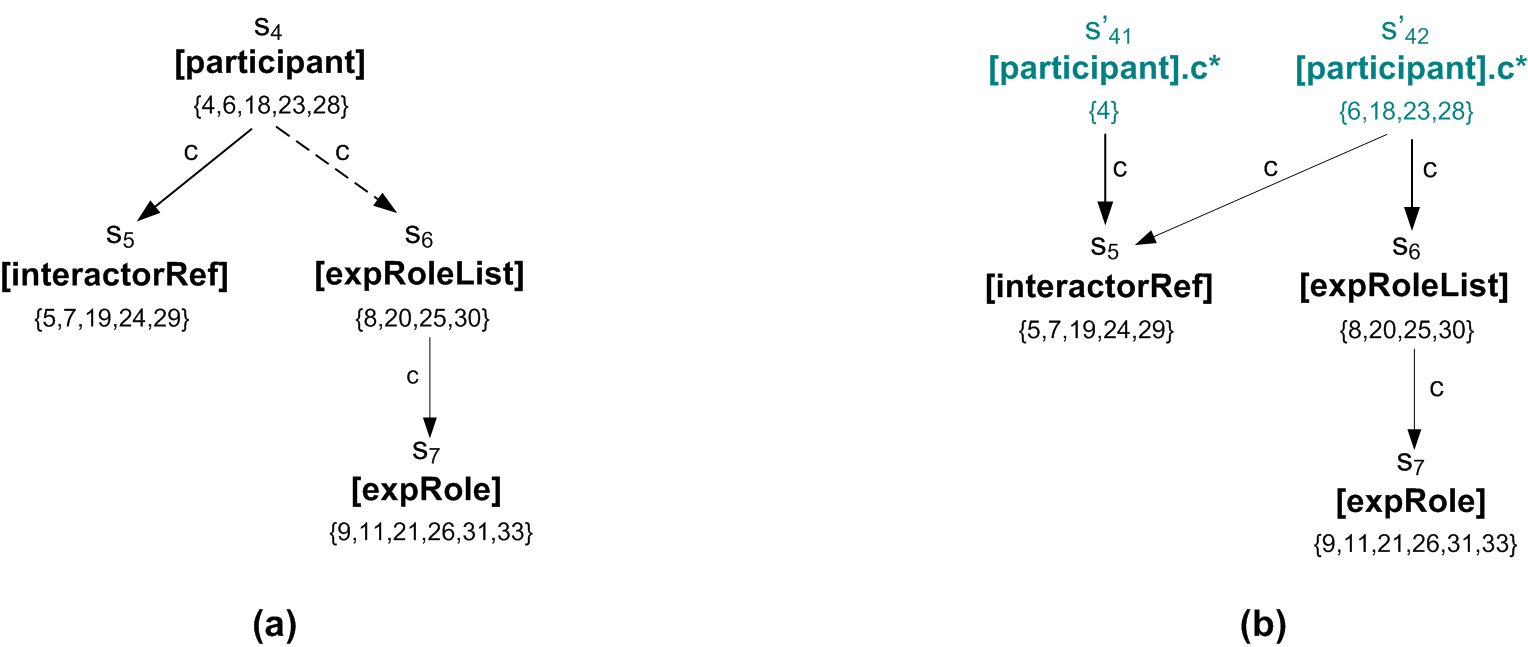}
    \caption{The $[participant].c^*$ neighbourhood of $s_4$ from Figure \ref{fig:PSIMI-label}: (a) before a $c^*$ refinement, (b) after the refinement}
    \label{fig:SDNeighbourhood}
\end{figure}

\begin{example}
Consider the $[participant].c^*$ neighbourhood of nodes $s'_{41}$
and $s'_{42}$ in Figure \ref{fig:SDNeighbourhood}. The extents of
$s'_{41}$ and $s'_{42}$ are $\{4\}$ and $\{6,18,23,28\}$,
respectively. The LPath sets of the nodes are $LPath(s'_{41})= \{
c[interactorRef] \}$ and $LPath(s'_{42})= \{ c[interactorRef],$
$c[expRoleList].c[expRole] \}$, whereas the EEs are
$e_1=ds::participant[c::interactorRef]$
$[count(ds::*)=count(c::interactorRef)]$ and
$e_2=ds::participant[c::interactorRef]$
$[c::expRoleList][c::expRoleList/c::expRole]
[count(d::*)=count(c::interactorRef)+count(c::expRoleList)+count(c::expRoleList/c::expRole].$
\qed
\end{example}

The next proposition shows, for the special case of $c^*$, that
the EEs thus constructed return a set of nodes that do have the
axis label paths specified in the predicates.

\begin{proposition} \label{pro:extentXPath2}
Let $\mathcal{A}$ be an axis graph and $e_s=ds::l[e_1] \ldots
[e_n][count(d::*) = count(e_1)+ \ldots + count(e_n)]$ an XPath
expression where $e_i=c::l_{i_1}/ \ldots /c::l_{i_{k_i}}$, $1 \leq
i \leq m$. Then, $e_s$ returns all nodes $v$ such that there
exists only $lp_1, \ldots, lp_n$ axis label paths from $v$ of the
form $lp_i= c[l_{i_1}]. \ldots .c[l_{i_{k_i}}].$ \qed
\end{proposition}

For proving Proposition \ref{pro:extentXPath2} we need the
following Lemmas.

\begin{lemma} \label{thm:inclusionP[e]P[ds]}
Let $e=child::l_1/ \ldots /child::l_m$ be an XPath expression. For
every node $v \in Inst: \mathcal{D}[\![e]\!](v) \subseteq
\mathcal{D}[\![d::*]\!](v).$ \qed
\end{lemma}

\begin{lemma} \label{thm:pigeonhole}
Let $S$ and $S_1, \ldots, S_n$ be sets such that $S_i \subseteq S$
and $S_i \neq \emptyset$ for $1 \leq i \leq n$. Then $| \; S \; |
= | \; \sum^n_{i=1} S_i \; | \Leftrightarrow S = \bigcup^n_{i=1}
S_i.$ \qed
\end{lemma}

Using Lemmas \ref{thm:inclusionP[e]P[ds]} and \ref{thm:pigeonhole}
we can prove Proposition \ref{pro:extentXPath2} as follows.

\begin{proof}
\begin{tabbing}
\hspace{0in} \= \hspace{.1in} \= \hspace{.1in} \= \kill \\
\> \>   $\mathcal{D}[\![d::*/[e_1] \ldots
[e_n][count(d::*) = count(e_1)+ \ldots + count(e_n)]]\!](v_0) $\\
\> = \> \> by semantic rules (\ref{sem:union}) and (\ref{sem:axis}) in Figure \ref{sem:locationPaths} \\
\> \>   $ \{v \; | \; \langle v_0,v \rangle \in d \; \bigwedge^n_{i=1}\mathcal{E}[\![e_i]\!](v,pos_{d}(v,S),|S|) = true \; \wedge \; $ \\
\> \> \>  \`  $ \mathcal{E}[\![count(d::*)= count(e_1)+ \ldots + count(e_n)]\!](v,pos_{d}(v,S),|S|) = true \} $  \\
\> = \> \> since all nodes $v$ are reachable from $v_0$, $\langle v_0,v \rangle \in d$ is always true  \\
\> \>   $ \{v \; | \; \bigwedge^n_{i=1}\mathcal{E}[\![e_i]\!](v,pos_{d}(v,S),|S|) = true \; \wedge \; $ \\
\> \> \>  \`  $ \mathcal{E}[\![count(d::*)= count(e_1)+ \ldots + count(e_n)]\!](v,pos_{d}(v,S),|S|) = true \} $  \\
\> = \> \> by semantic rule (\ref{sem:op}) in Figure \ref{sem:expressions} with $Op$'s $count$, $+$, and $=$ from Figure \ref{fig:BasicOperators} \\
\> \>   $ \{v \; | \; \bigwedge^n_{i=1}\mathcal{E}[\![e_i]\!](v,pos_{d}(v,S),|S|) = true \; \wedge \; | \; \mathcal{D}[\![d::*]\!](v) \; | \; = \sum^n_{i=1} \; | \; \mathcal{D}[\![e_i]\!](v) \; | \; \} $  \\
\> = \> \> by semantic rule (\ref{sem:locPath}) in Figure \ref{sem:expressions}  \\
\> \>   $ \{v \; | \; \bigwedge^n_{i=1}\mathcal{D}[\![e_i]\!](v) \neq \emptyset \; \wedge \; | \; \mathcal{D}[\![d::*]\!](v) \; | \; = \sum^n_{i=1} \; | \; \mathcal{D}[\![e_i]\!](v) \; | \; \} $  \\
\> = \> \> by Lemmas \ref{thm:inclusionP[e]P[ds]} and \ref{thm:pigeonhole} \\
\> \>   $ \{v \; | \; \bigwedge^n_{i=1}\mathcal{D}[\![e_i]\!](v) \neq \emptyset \; \wedge \; \mathcal{D}[\![d::*]\!](v) = \bigcup^n_{i=1} \mathcal{D}[\![e_i]\!](v)  \}  $  \\
\> = \> \> by semantic rule (\ref{sem:axis}) in Figure \ref{sem:expressions} \\
\> \>   $ \{v \; | \; \bigwedge^n_{i=1}\mathcal{D}[\![e_i]\!](v) \neq \emptyset \; \wedge \; \{ w \; | \; \langle v,w \rangle \in d \} = \bigcup^n_{i=1} \mathcal{D}[\![e_i]\!](v) \}  $  \\
\> = \> \> since the $lp_i$'s are of the form $lp_i= c[l_{i_1}]. \ldots .c[l_{i_{k_i}}]$ \\
\> \>   $ \{v \; | \; \{ w \; | \; \langle v,w \rangle \in d \} = \{ w \; | \; p=(v, c, \ldots,c,w)$ is an axis path $\; \wedge \; \lambda(p)=lp_i  \}  $  \\
\qed
\end{tabbing}
\end{proof}

\section{Stabilization with XPath} \label{sec:virtual}

As we have seen in Chapter \ref{sec:NeighbourhoodStabilization},
edge stabilization can be reduced to node refinement. However,
when the EEs of the nodes in an edge are available, we can use the
description provided by the EEs and compute the stabilization
directly from them. The idea is to express the condition for
forward-stability (i.e., $\forall x \in extent(s_i), \exists y \in
extent(s_j) \wedge \langle x,y \rangle \in axis$) of an edge
$\langle s_i, s_j \rangle$ in XPath using $ee(s_i)$ and $ee(s_j)$.

Algorithm~\ref{alg:edgeStabilizationWithXPath} computes the
stabilization of a single edge by updating the EEs of the nodes in
the edge and their extents. The algorithm replaces node $s_i$ by
two new nodes: $s_i'$ and $s_i''$. The extent of $s_i'$ contains
all nodes in the extent of the original $s_i$ that are in an
$axis$ relation with nodes in the extent of $s_j$ (line 2). The
extent of $s_i''$ contains the complement of $s_i'$ with respect
to $s_i$, i.e., it contains all nodes that do not have such an
$axis$ relation with nodes in the extent of $s_j$ (line 3).
Consequently, after the new edge is created (line 6), $s_i'$ has a
forward-stable $axis$ edge to $s_j$ whereas $s_i''$ does not have
any $axis$ edge to $s_j$. The EEs obtained in lines 4 and 5 are
the EEs for the new nodes.

Note that we do not need additional $count()$ nor $not()$
predicates in the new expressions because all the required ones
are already in $ee(s_i)$ and $ee(s_j)$.

\begin{figure}
\begin{algorithm}
\label{alg:edgeStabilizationWithXPath}

    \mbox{}\\
    \parbox{\textwidth}{
    \noindent{\sl \underline{stabilizeEdgeXPath$(D,s_i,s_j)$}}\\[0.5em]
    \noindent{\bf Input:} An SD $D$ containing  a non forward-stable edge
$e=\langle s_i, s_j \rangle$ with label $axis$\\
    \noindent{\bf Output:} An SD $D$ where $e$ has been replaced
    by forward-stable $e'=\langle s_i', s_j \rangle.$

    \begin{algorithmic} [1]

    \STATE create new nodes $s_i'$ and $s_i''$
    \STATE $extent(s_i') := \{ x \in extent(s_i) \, | \, \exists y \in extent(s_j) \wedge \langle x,y \rangle \in axis \}$
    \STATE $extent(s_i''):= extent(s_i) - extent(s_i')$
    \STATE $ee(s_i')=ee(s_i)[\mathit{axis::\lambda(s_j)} \; intersect \; ee(s_j)]$
    \STATE $ee(s_i'')=ee(s_i)[not(\mathit{axis::\lambda(s_j)} \; intersect \; ee(s_j))]$
    \STATE create an edge $e'=\langle s_i', s_j \rangle$
    \STATE let $S$ be the set of nodes connected to $s_i$
    \FOR {every node $s$ in $S$}
        \STATE add edges $\langle s_i', s \rangle$ and $\langle s, s_i' \rangle$ if conditions in Definition \ref{def:stability} are satisfied
    \ENDFOR
    \STATE delete node $s_i,$ and all its incoming and outgoing edges
    \end{algorithmic}
}

 \end{algorithm}
 \end{figure}

\begin{example}
Consider edge $\langle s_4, s_6 \rangle$ from Figure
\ref{fig:stable} (a), which is not forward-stable. Edge
stabilization will create two nodes, $s_{41}$ and $s_{42}$ as
shown in Figure \ref{fig:stable} (b). Given that
$ee(s_4)=\mathit{/ds::participant}$,
$ee(s_6)=\mathit{/ds::expRoleList}$, and the stabilized edge
corresponds to a $c$ axis, the resulting expressions are the
following: $ee(s_{42}) = /ds::participant$ $[child::expRoleList \;
intersect \; /ds::expRoleList]$ and $ee(s_{41})
=/ds::participant[not(child::expRoleList \; intersect \;
/ds::expRoleList)]$. \qed
\end{example}

\section{Adapting SDs to XPath queries} \label{sec:candidates}

Previously in this chapter, we used XPath to express EEs and to
manipulate them for refinement and edge stabilization operations.
In this section we show how XPath queries are used to guide a
refinement operation in the process of adapting an SD to a query.

In order to evaluate a query using an SD, we need to find the SD
nodes that participate in the answer. \dx's approach is to find
the SD nodes that contain a \emph{superset} of the answer and then
evaluate the entire expression on them to get the \emph{exact}
answer.

One of the central problems for finding a superset of the answer
is how to decide what SD nodes can be used to answer an XPath
query. This requires some sort of XPath \emph{matching} algorithm
and the ability to decide whether there exists an exact rewriting
of a query using an SD. The matching algorithm will transform the
\emph{structural subquery} of the XPath expression (the expression
that results from removing all non-structural predicates such as
those containing functions) to be evaluated into an equivalent
AxPRE $\alpha$. Then, we need to find the SD node (or nodes) whose
AxPRE is contained in $\alpha$. The union of the extents of such
nodes are a superset of the answer. If the query is purely
structural (i.e. the query is equal to its structural subquery)
and $\alpha$ is equivalent to some SD node AxPRE, then the answer
to the query is exactly the union of the extents. Otherwise, we
need to run the entire query on the union of the extents to find
the exact answer.

We begin by discussing in the next section how to derive an AxPRE
from an XPath expression.

\subsection{Deriving AxPREs from queries}

\dx\ can adapt an SD node to an XPath query $Q$, as we have
illustrated in our RRS feeds motivating example in Chapter
\ref{sec:workloadExample}. This section formalizes how an AxPRE is
obtained from $Q$ by using the two derivation functions $L$ and
$P$ we provide in Figure~\ref{axpre:construction2}. We begin by
illustrating the XPath-to-AxPRE derivation with a concrete
example.

\begin{figure}
\begin{equation}
P(Op(e_1, \ldots, e_m)) := \epsilon \label{p:op}
\end{equation}
\begin{equation}
P(axis\!::\!l[e_1]\ldots[e_m] / rlocpath) :=
\mathit{Ax}(axis).(P(e_1) | \! \ldots \! | P(e_m) | P(rlocpath))
\label{p:composition}
\end{equation}
\begin{equation}
P((locpath)[e_1]\ldots[e_m] / rlocpath) := P(locpath).(P(e_1) | \!
\ldots \! | P(e_m) | P(rlocpath)) \label{p:parenthesis}
\end{equation}
\begin{equation}
P(locpath_1 | \! \ldots \! | locpath_m) := ( P(locpath_1) | \!
\ldots \! |  P(locpath_m) ) \label{p:union}
\end{equation}
\begin{equation}
L(rlocpath/axis\!::\!l[e_1]\ldots[e_m]) :=
\mathit{Ax}(axis^{-1}).(L(rlocpath)) | P(e_1) | \! \ldots \! |
P(e_m) \label{axpre:composition2}
\end{equation}
\begin{equation}
L(rlocpath/(locpath)[e_1]\ldots[e_m]) := L(locpath).(L(rlocpath))
| P(e_1) | \! \ldots \! | P(e_m) \label{axpre:parenthesis2}
\end{equation}
\begin{equation}
L(locpath_1 | \! \ldots \! | locpath_m) := ( L(locpath_1) | \!
\ldots \! | L(locpath_m)) \label{axpre:union2}
\end{equation}
\caption{AxPRE derivation functions $L$ and $P$ }
\label{axpre:construction2}
\end{figure}

\begin{example}
Consider the following query

\begin{center}
{\small \texttt{Q3 =
/ds::participant[c::expRoleList/fc::expRole/ns::expRole]
 \textcolor[rgb]{0.50,0.50,0.50}{[not(ds::expRole/names=``prey'')]}}}
\end{center}

Q3 returns all participants that have expRoleLists whose first two
children are expRole elements and that are not playing the
``prey'' role in the experiments. Note that the structural
subquery appears in black (the last predicate in grey is not part
of the structural subquery).

The first rule of Figure \ref{axpre:construction2} that applies is
(\ref{axpre:composition2}) with the following variables: $rlocpath
= \emptyset$, $axis=ds$, $l=participant$,
$e_1=c::expRoleList/fc::expRole/ns::expRole$, and
$e_2=not(ds::expRole/names=``prey")$, resulting in
\[\mathit{Ax}(ds^{-1})|P(e_1)|P(e_2)\] where $\mathit{Ax}$ is a function that translates the XPath axis
into its AxPRE axis counterpart. In particular,
$\mathit{Ax}(axis^{-1})$ returns the actual AxPRE inverse (e.g.,
$child^{-1}$ is converted into $p$) and recursive axes are
translated to an equivalent Kleene closure of non-recursive axes
(e.g., $descendant$ translates into $c^*$).

The expansion of $P(e_2)$ is very simple. The predicate is
basically a function, so it matches rule (\ref{p:op}) and the
result of $P(not(ds::expRole/names=``prey"))$ is $\epsilon$
(Remember that this predicate is not part of the structural
subquery). This results in the following intermediate expression
\[ as|P(e_1)|\epsilon \]

For expanding $P(e_1)$, the first rule invoked is
(\ref{p:composition}) with $axis = c$, $l=expRoleList$,
$rlocpath=fc::expRole/ns::expRole$ and empty predicates. The
intermediate expression is now \[
as|Ax(c).(P(fc::expRole/ns::expRole)) \]

For expanding $P(fc::expRole/ns::expRole)$, the rule that applies
is (\ref{axpre:parenthesis2}) with $axis = fc$, $l=expRole$,
$rlocpath=ns::expRole$ and no predicates, which results in
\[ as|c.Ax(fc).(P(ns::expRole)) \]

Similarly, we can expand $P(ns::expRole)$ and obtain \[ as|c.fc.ns
\]

Finally, the node test of the step corresponding to the answer
($participant$ in this case) is prefixed as a label predicate to
the AxPRE. Therefore, the resulting AxPRE of query $Q3$ is
\[\alpha_{Q3} = [participant].(as|c.fc.ns) \]

\end{example}

Once the query AxPRE $\alpha$ of a given XPath query $Q$ is
computed, the next step in adapting the SD to $Q$ is finding the
SD node (or nodes) whose AxPRE $\alpha'$ is contained in $\alpha$.
Since the problem of AxPRE containment is related to that of
regular expression containment, any regular expression containment
algorithm can be used here. After finding the node, \dx\ proceeds
to change $\alpha'$ to $\alpha$, which in fact modifies the
description of the node and thus the neighbourhood it summarizes.
This entails performing a \emph{refinement} of the extent of the
node.

\subsection{Finding candidates} \label{sec:findingcandidates}

If an extent contains a superset of the answer of a query, then we
say that the elements in such an extent are \emph{candidate
elements}. Note that, by adapting the SD to the structural
subquery, \dx\ has found a restricted superset of the answer and
hence has considerably reduced the search space for computing the
entire query.

The \dx\ architecture is tailored to process XML collections one
file-at-a-time, the prevalent data processing model for the Web.
Each file is parsed and processed independently of the other files
in the collection. In this context, after adapting the SD to a
given query $Q$, \dx\ can restrict the evaluation of $Q$ to those
documents (called \emph{candidate documents}) that are guaranteed
to provide a non-empty answer for the structural subquery of $Q$.
Those candidate documents that do contain an answer for the entire
query are called \emph{answer documents}.

Once \dx\ has computed the query AxPRE $\alpha$ of a given XPath
query $Q$ as described above, it needs to find the SD node whose
AxPRE contains $\alpha$ in order to get the candidate documents
for evaluating $Q$. If there is an SD node $s$ with AxPRE
$\alpha$, then all documents in the extent of $s$ are in fact
candidate documents. In contrast, if $s$ has an AxPRE $\alpha'$
containing $\alpha$, \dx\ has two alternatives. One, it can adapt
the SD by refining $s$ from $\alpha'$ to $\alpha$ and then get the
candidate documents as in the previous case. Two, it can get all
documents in the extent of $s$ and run the structural subquery of
$Q$ on them in order to get the candidates. Once the candidate
documents are found, finding the answer documents entails running
$Q$ on all candidates.

\begin{example}
Consider query $Q3$ for our running example. We could evaluate
$Q3$ using the label SD from Figure \ref{fig:PSIMI-label}. In that
case, the only node whose AxPRE is contained in $\alpha_{Q3}$ is
$s_4$ with AxPRE $[participant]$. For simplicity, let us assume
that every node in the extent of $s_4$ belongs to a different
document, so there will be as many elements as documents, both
candidates and answer. From the SD graph we know that not all
participants in the extent of $s_4$ contain an $expRoleList$
element because the edge $\langle s_4, s_6 \rangle$ is not
forward-stable. So we conclude that $s_4$ contains only a superset
of the answer, so we get the six documents from the extent and
evaluate the query in all of them in order to get the answer.

Alternatively, we could use the refined SD from Figure
\ref{fig:PSIMI-axpre}. This SD could have been obtained from the
label SD after adapting it to $Q3$. Regardless of how the SD was
created, we found that three nodes have AxPREs contained in
$\alpha_{Q3}$: $s_{41}$, $s_{42}$, and $s_{43}$. However, we
notice from the SD graph that only node $s_{42}$ has a
\emph{forward-stable} neighbourhood for $\alpha_{Q3}$. (Note that
it is the only $[participant]$ node with an edge $c$, followed by
a $fc$ and an $ns$, all forward-stable.) That means that both
nodes (and thus documents) in the extent of $s_{42}$ are
candidates, and thus we need to run $Q3$ only in those two
documents. If $Q3$ did not have the second predicate (in grey),
the extent of $s_{42}$ would be the exact answer of $Q3$. \qed

\end{example}

The process of exploring candidates is not unidirectional: a
developer can move back and forth between the query explanations
described here and the structural exploration described in Chapter
\ref{sec:workloadExample}. For instance, she may create an SD and
run a query on some candidate documents. Next, she might decide to
relax the query in order to further investigate its impact on the
collection. Then, she may want to get a more or less refined
description of the collection by changing the SD using AxPRE
refinements, and then start the process again. \dx\ provides the
developer with this interactive functionality for describing and
evaluating XPath queries on large XML collections.

\chapter{\dx\ engine} \label{Section:Implementation}

In previous chapters we introduced \dx, a powerful framework
capable of declaratively describing complex structural summaries
of XML collections that captures and generalizes many proposals in
the literature. We also showed how summary descriptors (SDs) are
created and refined to selectively produce more or less detailed
descriptions of the data. In this chapter, we discuss how the \dx\
framework is implemented in the summarization engine and present
two strategies for refining an SD: one is based on materializing
the SD partitions, the other is a virtual approach that relies on
constructing XPath expressions that compute extents.

The \dx\ architecture is tailored to process XML collections one
file at a time, the prevalent data processing model for the Web.
Each file is parsed, processed and stored before continuing with
the next file in the collection. Such an approach supports the
interactive creation and refinement of AxPRE SDs for large
collections of XML documents.

The \dx\ engine is implemented in Java using Berkeley DB Java
Edition\footnote{{\small
\tt{http://www.oracle.com/technology/products/berkeley-db/je/}}}
to store and manage indexed collections (tables). The
implementation can invoke an arbitrary JAXP 1.3\footnote{{\small
\tt{http://jaxp.dev.java.net/1.3/}}} XPath processor for the
evaluation of XPath expressions. JAXP is an implementation
independent portable API for processing XML with Java. For the
experiments reported later in this paper, the
Saxon\footnote{{\small \tt{http://saxon.sourceforge.net/}}} XPath
processor was employed. The Saxon implementation conforms to the
XPath 1.0 standard set by the W3C \cite{W3C:XPath/01} and
therefore satisfies the semantic characterization formalized in
Appendix \ref{Section:XPathLanguage}.

The \dx\ implementation stores the extents in an indexed table
named {\tt elemDB} that has schema {\tt elemDB}(\underline{\tt
SID}, \underline{\tt docID}, \underline{\tt endPos}, {\tt
startPos}, \texttt{SID2}), where the underlined attributes are the
key (also used for indexing). The {\tt elemDB} table contains a
tuple for each XML element in the collection. Each SD node is
identified by a unique id called SID. Each element belongs to the
extent of a unique SD node, whose SID is stored in the {\tt SID}
attribute. The attribute {\tt docID} holds the identifier of the
document in which the element appears. The {\tt startPos} and {\tt
endPos} are the positions, in the document, where the element
starts and ends, respectively. \texttt{SID2} allows us to maintain
an SID for a second SD.

Alternatively, the user can decide to keep the extents virtual and
thus make the \dx\ engine store a {\tt docDB} table instead of the
{\tt elemDB} table described above. The schema of the {\tt docDB}
table is {\tt docDB}(\underline{\tt SID}, \underline{\tt docID}),
which contains for each sid $s$ the docIDs of all XML documents
containing elements in the extent of $s$. This can be used to
efficiently locate the XML documents to be evaluated by the EE of
$s$ in order to get the extent of $s$. The EEs are stored in a
separate XML file.

A third scenario in which both {\tt elemDB} and {\tt docDB} tables
coexist is also possible. In such a case, some SIDs would be kept
in the {\tt elemDB} table (with their extents materialized) and
some others would be stored without extents in the {\tt docDB}
table. In this thesis we have not studied the trade offs emerging
from this scenario.

The \dx\ engine keeps the SD graph in main memory in separate hash
tables for each axis relation in the SD, e.g. the {\tt parentsMap}
and {\tt childrenMap} maps contain the edge definitions for the
$p$ and $c$ SD axes respectively. In other words, each binary
$axis$ relation is stored as a map between a key SID $s$ and a set
of SIDs $s_1, \ldots, s_n$ such that $\langle s, s_i \rangle \in
axis$, $1 \leq i \leq n$. In addition, there is a label map,
\texttt{labelMap}, that contains the label of each SD node.

\section{Initial SD construction} \label{sec:initial}

Some SDs can be constructed in one pass over the collection. This
is possible when the parsing information collected at either the
start tag or the end tag of an element $v$ is enough to construct
the AxPRE neighbourhood $\mathcal{N}_{\alpha}(v)$ of the element,
compute the AxPRE partition and thus decide to what extent $v$
belongs. For instance, the \texttt{start} tag itself is enough to
classify an element $v$ when constructing the $\epsilon$ SD (the
$\mathcal{N}_{\epsilon}(v)$ contains just node $v$). For the $p^k$
and $p^*$ SDs, it suffices to keep the sequence of the last $k$
open elements (for the $p^k$) or all of them (for the $p^*$) for
creating $\mathcal{N}_{p^k}(v)$ and $\mathcal{N}_{p^*}(v)$. Thus,
$p^k$ and $p^*$ SDs can also be constructed in one pass over the
collection.

\begin{figure}
\begin{algorithm}
\label{alg:buildP(K)}

    \mbox{}\\
    \parbox{\textwidth}{
    \noindent{\sl \underline{buildP$(k)$}}\\[0.5em]
    \noindent{\bf Input:} Collection $C$ of XML documents \\
    \noindent{\bf Output:} $p^k$ SD

    \begin{algorithmic} [1]

    \FOR {each XML document $doc$ in collection $C$}
        \STATE assign a new docID $d$ to $doc$
        \STATE create a new XOM tree $t$
        \WHILE {parsing $doc$}
        \IF {element \texttt{start} tag is found in $doc$}
            \STATE create a new $e$ in $t$ with XML attributes $sid$, $startPos$, and $endPos$ set to empty
            \IF {the $p^k$ neighbourhood of $e$ is not in the SD graph}
                \STATE create a new SID $s'$
                \STATE update {\tt labelMap}, {\tt parentsMap}, and {\tt childrenMap}
                \STATE store the $p^k$ XPath expression of $s'$ in the EE XML file
            \ENDIF
            \STATE get the sid $s$ of $e$ from the SD
            \STATE set $e.sid$ to $s$ and $e.startPos$ to the offset position of the \texttt{start} tag of the element
        \ENDIF
        \IF {element \texttt{end} tag is found in $doc$}
            \STATE set $e.endPos$ to the offset position of the end tag of the element
            \STATE append tuple $(e.s, d, e.endPos, e.startPos)$ to {\tt elemDB}
        \ENDIF
        \ENDWHILE
    \ENDFOR
    \end{algorithmic}
    }
\end{algorithm}
\end{figure}

Algorithm~\ref{alg:buildP(K)} (buildP(k)) illustrates the use of
the \dx\ data structures. The algorithm computes the $\epsilon$,
$p^k$ and $p^*$ SDs. The parameter $k$ encodes the SD as follows:
$k=0$ corresponds to $\epsilon$, $k=maxint$ to $p^*$, and all
other values represent $p^k$. For each XML document in the
collection, the algorithm parses the document and creates a
XOM\footnote{{\small \tt{http://www.xom.nu/}}} tree (a lightweight
XML object model). The algorithm uses the XOM tree created for
composing the {\tt elemDB} tuple of each element in the document
containing SID, docID, and its beginning and end offset position.
Both the XOM tree and the SD are constructed simultaneously during
parsing time.

Once an SD has been constructed from scratch, the user can refine
any SD node or set of nodes by changing the node's AxPRE, as
described in Chapter \ref{sec:refinements}. In the next section we
provide algorithms for computing such refinements.

\section{Computing refinements} \label{sec:algorithms}

\begin{figure}
\begin{algorithm}
\label{alg:refineMaterialized}

    \mbox{}\\
    \parbox{\textwidth}{
    \noindent{\sl \underline{refineMaterialized$(sd, s, r_1, \ldots, r_n)$}}\\[0.5em]
    \noindent{\bf Input:} $sd$ is the SD, $s$ is the sid to be refined, $r_1 \ldots r_n$ is a family of refining XPath EEs \\
    \noindent{\bf Output:} Updated $sd$
    \begin{algorithmic} [1]

    \STATE get the XPath EE $e_s$ of $s$
    \FOR {each input $r_i$}
        \STATE create a new sid $s_i$
        \FOR {each $d$ s.t. there is a tuple $t_d$ in {\tt elemDB} with $t_d.SID=s$ and $t_d.docID=d$}
            \STATE create a XOM tree $t$ of $d$ in which each element has an $endPos$ attribute with the offset position of the \texttt{end} tag of the element
            \STATE assign to $extent$ the answer of $/e_s/r_i$
            \FOR {each element $n_j$ in $extent$}
                \STATE locate the tuple $t_j$ in the {\tt elemDB} table corresponding to $n_j$ by using $(s,$ $d,$
                $n_j.endPos)$ as a key
                \STATE assign $s_i$ to tuple $t_j$ by setting $t_j.SID=s_i$
            \ENDFOR
            \STATE update {\tt labelMap} by assigning the label of $s$ to the new $s_i$
            \STATE store the $r_i$ EE of $s_i$ in the EE XML file
            \FOR {each $axis$ in the $SD$}
                \STATE call computeEdgeByMerge$(sd, axis, s_i, extent, s)$ to test the existence of an $axis$ edge from $s_i$
            \ENDFOR
        \ENDFOR
    \ENDFOR

\end{algorithmic}
}
\end{algorithm}
\end{figure}

Following the materialized extents approach, a refinement can be
evaluated with Algorithm~\ref{alg:refineMaterialized}
(refineMaterialized), whereas virtual extents can be refined by
Algorithm~\ref{alg:refineVirtual} (refineVirtual). Both algorithms
are invoked with sid $s$ to be refined, its current EE $e_s$, and
a family $r_1 \ldots r_n$ of refining EEs, constructed as
described in Chapter \ref{sec:virtual}.

\vspace{.1in}

Suppose that SD node $s_i$ with EE $r_i$ is one of the refinements
of SD node $s$ with EE $e_s$. The extent of $s_i$ is computed by
evaluating $r_i$ on the set of documents that contain elements in
the extent of $s$, which entails evaluating the expression
$/e_s/r_i$ (line 6 of Algorithms \ref{alg:refineVirtual}
(refineVirtual) and \ref{alg:refineMaterialized}
(refineMaterialized). This set of documents are obtained from {\tt
ElemDB} (if the extent of $s$ is materialized) or from {\tt docDB}
(if the extent of $s$ is virtual). Once we have the extent of
$s_i$, the edges in the SD graph can be constructed either from
the EE when the extent is virtual (by computeEdgeByXPath, line 10
of Algorithm refineVirtual) or from {\tt ElemDB} when the extent
is materialized (by computeEdgeByMerge, line 13 of Algorithm
refineMaterialized).

\vspace{.1in}

In order to update the edges, we need to check whether there is an
$axis$ edge between $s_i$ and a set of candidate SD nodes $c_1,
\ldots, c_n$ such that $\langle s, c_j \rangle \in axis$. This is
performed by Algorithm \ref{alg:computeEdgeByXPath}
(computeEdgeByXPath) by computing the expression $e_{s_r}/axis::*
\cap e_{c_j}$, where $e_{c_j}$ is the EE of candidate $c_j$ (line
4). If the evaluation of the expression is not empty, then there
exists an edge from $s_i$ to $c_j$, otherwise there is no edge
(lines 5 and 6).

\vspace{.1in}

Algorithm computeEdgeByMerge (not shown), in contrast, simply
computes a merge of the {\tt ElemDB} using the {\tt startPos} and
{\tt endPos} attributes to check for containment (in case of
$\mathit{fc}$, $c$, $p$, $a$, and $d$ axes) or precedence (for
$ns$, $\mathit{fs}$, $f$, and $p$ axes).

\begin{figure}
\begin{algorithm}
\label{alg:refineVirtual}

    \mbox{}\\
    \parbox{\textwidth}{
    \noindent{\sl \underline{refineVirtual$(sd, s, r_1, \ldots, r_n, extent)$}}\\[0.5em]
    \noindent{\bf Input:} $sd$ is the SD, $s$ is the sid to be refined, $r_1 \ldots r_n$ is a family of refining XPath EEs \\
    \noindent{\bf Output:} Updated $sd$, $extent$ with the element in the extent of $s_i$

    \begin{algorithmic} [1]

    \STATE get the XPath EE $e_s$ of $s$
    \FOR {each input $r_i$}
        \STATE create a new sid $s_i$
        \FOR {each $d$ s.t. there is a tuple $t_d$ in {\tt docDB} with $t_d.SID=s$ and $t_d.docID=d$}
            \STATE create a XOM tree $t$ of $d$ in which each element has an $endPos$ attribute with the offset position of the \texttt{end} tag of the element
            \STATE assign to $extent$ the answer of $/e_s/r_i$
            \STATE update {\tt labelMap} by assigning the label of $s$ to the new $s_i$
            \STATE store the $r_i$ XPath expression of $s_i$ in the EE XML file
            \FOR {each $axis$ in $sd$}
                \STATE call computeEdgeByXPath$(sd, axis, s_i, extent, s)$ to test the existence of an $axis$ edge from $s_i$
            \ENDFOR
        \ENDFOR
    \ENDFOR

\end{algorithmic}
}
\end{algorithm}
\end{figure}

\begin{figure}
\begin{algorithm}
\label{alg:computeEdgeByXPath}

    \mbox{}\\
    \parbox{\textwidth}{
    \noindent{\sl \underline{computeEdgeByXPath$(sd, axis, s_i, extent, s)$}}\\[0.5em]
    \noindent{\bf Input:} $sd$ is the SD, $axis$ is the axis edge to be computed, $s_i$ is the new sid, $extent$ is
the extent of $s_i$, and $s$ is the sid being refined. \\
    \noindent{\bf Output:} Updated $sd$
    \begin{algorithmic} [1]

    \STATE assign to $candidates$ the set of sids $\{c_1, \ldots, c_n \}$ mapped to $s$ in {\tt axisMap}
    \FOR {each $c_j$ in $candidates$}
        \STATE get the EE $e_j$ of $c_j$ from the EE XML file
        \STATE evaluate the intersection expression $e=axis::* \cap e_j$ from $extent$
        \IF {the evaluation of $e$ is not empty}
            \STATE add an axis edge between $s_i$ and $c_j$ to the corresponding {\tt axisMap}
        \ENDIF
    \ENDFOR

\end{algorithmic}
}
\end{algorithm}
\end{figure}

\newpage

In this chapter, we presented an implementation of the \dx\
framework that supports the interactive creation and
refinement/stabilization of AxPRE SDs for XML collections. We
introduced two strategies for locally updating an SD: one based on
materializing the SD partitions (extents), the other relies on a
novel virtual approach based on XPath expressions. The next
chapter presents experimental results that demonstrate the
scalability of our strategies, even to multi gigabyte web
collections.

\chapter{Experimental results} \label{Section:Experiments}

We present here the results of an extensive empirical study we
conducted using the \dx\ framework introduced in this thesis.

The first part of our study evaluates the performance of the
initial SD construction and the feasibility of the different
approaches (materialized, virtual, edges, etc.) to \dx\ main
exploration operations: refinement and stabilization. The
objective here is twofold. First, to understand how key parameters
(e.g., extent size, number of documents involved, and number of SD
nodes and edges affected) impact on each operation. Second, to
determine what method performs better under what kind of
conditions.

The goal of the second part of our experimental evaluation is to
study the impact of various summaries on XPath query processing
performance. This part also provides a comparison with variations
of incoming and outgoing path summaries capturing existing
proposals like 1-index, APEX, A(k)-index, D(k)-index, and
F+B-Index. We want to emphasize that query evaluation times on
collections the size of Wikipedia are rarely reported in the
literature. In fact, XML DB systems (and not just research
prototypes) become challenged when working with collections at
this scale. The experiments demonstrate that \dx\ easily scales up
to gigabyte sized XML collections with important performance
results.

\begin{table}
\centering \caption{Test collections} {\small
\begin{tabular}{|l|r|r|r|r|r|r|}
\hline
\textbf{Collection} & \multicolumn{1}{c|}{\textbf{Size}} & \textbf{\#Docs} & \multicolumn{2}{c|}{\textbf{\#Nodes}} & \multicolumn{2}{c|}{\textbf{Load Time (s)}} \\
\cline{4-7}
                    & \multicolumn{1}{c|}{\textbf{(MB)}}    &                 & \multicolumn{1}{c|}{$p^*$ SD} & \multicolumn{1}{c|}{label SD} & \multicolumn{1}{c|}{$p^*$ SD} & \multicolumn{1}{c|}{label SD} \\
\hline \hline
RSS2                & 210                & 9600            & 1058  & 301  & 64.2   & 41.4    \\
PSIMI2              & 234                & 156             & 199   & 54   & 93.1   & 81.7    \\
Wiki5              & 545                & 30000           & 15602 & 259  & 438.6  & 175.7   \\
Wiki45             & 4520               & 659388          & 66073 & 1245 & 8089.1 & 6201.2  \\
\hline
\end{tabular}
} \label{tab:TestCollections}
\end{table}

\section{Initial SD construction}

Our experiments were conducted over four collections of documents.
Table~\ref{tab:TestCollections} summarizes the size and number of
documents in each collection, and the number of nodes and load
times for the $p^*$ and label SDs, which includes computing the SD
graph and the partitions, and storing the extents in the {\tt
ElemDB} table.

For measuring times, we conducted five separate runs starting with
a cold Java Virtual Machine (JVM) for each query. The best and
worst times were ignored and the reported runtime is the average
of the remaining three times. The experiments were carried out on
a Windows XP machine with a 2.4GHz Intel Core 2 Quad processor,
and the JVM was allocated 1~GB of RAM.

The selected collections have different characteristics, namely
total size, size and number of individual documents, and document
heterogeneity. The first collection (RSS2) was obtained by
collecting RSS feeds from thousands of different sites. The second
collection (PSIMI2) is a fragment of the IntAct PSI-MI
dataset\footnote{\small
\tt{http://psidev.sourceforge.net/mi/xml/doc/user/} }. The third
and fourth collections (Wiki5 and Wiki45, respectively) were
created from the Wikipedia XML Corpus provided in INEX 2006
\cite{wikipediaxml:2006}. PSIMI2 is a very small collection in
terms of number of documents (only 156 in total) but a
medium-sized collection with respect to total size (about 234 MB).
In contrast, Wiki5 is about twice the PSIMI2 size but has almost
200 times the number of documents. Consequently, the average
document size in both collections ranges from 1.5 MB in PSIMI2 to
18 KB in Wiki 5. Documents in RSS2 are similar in size to Wiki5.
The largest collection (Wiki45 with 4.5GB spanning 660 thousand
files) is also the one with the smallest average document size
(only 6.8 KB).

The number of nodes in both $p^*$ and label SDs provide a measure
of heterogeneity and structural complexity. PSIMI2 is the most
homogeneous of our collections, with only 54 different element
names and 199 different label paths from the root. In contrasts,
the most heterogenous one is Wiki45 with over one thousand
different labels and over 66 thousand different label paths from
the root.

\section{Refinements} \label{sec:refResults}

We tested the performance of two types of SD updates: refinements
and stabilization. In this section we discuss the results for
refinements and we provide the stabilization results in the next
one.

Tables~\ref{tab:RSSEEs}, \ref{tab:PSIMIEEs} and \ref{tab:WikiEEs}
show the SIDs and EEs of the selected $p^*$ SD nodes in our test
collections. These are the nodes we use for refinements and edge
stabilization in our experiments reported below. For instance,
$r_{468}$ corresponds to the $p^*$ SD node that has
/rss/channel/image as its EE in the RSS2 collection. Our benchmark
refinements were selected with scalability in mind: smallest and
largest extents and number of documents involved are three orders
of magnitude apart, ranging from 4 documents in the $p_{193}$
refinement (Table \ref{tab:PSIMI2RefTimes}) to 6509 documents in
the $r_{449}$ refinement (Table \ref{tab:RSS2RefTimes}).

\begin{table}
\centering \caption{Selected $p^*$ SD nodes and EEs from RSS2}
\small{
\begin{tabular}{|l|l|}
\hline

\textbf{Node} & \textbf{Extent Expression (EE)} \\
\hline \hline
$r_{468}$ & /rss/channel/image \\
$r_{449}$ & /rss/channel/item \\
$r_{653}$ & /rss/channel/item/body \\
$r_{452}$ & /rss/channel/item/description \\
\hline
\end{tabular}
} \label{tab:RSSEEs}
\end{table}

\begin{table}
\centering \caption{Selected $p^*$ SD nodes and EEs from PSIMI2}
\small{
\begin{tabular}{|l|l|}
\hline

\textbf{Node} & \textbf{Extent Expression (EE)} \\
\hline \hline
$p_{59}$ & /entrySet/entry/interactorList/interactor/organism \\
$p_{18}$ & /entrySet/entry/experimentList/experimentDescription/bibref/xref \\
$p_{24}$ & /entrySet/entry/experimentList/experimentDescription\\
         & /hostOrganismList/hostOrganism \\
$p_{193}$ & /entrySet/entry/interactorList/interactor/organism/cellType \\
\hline
\end{tabular}
} \label{tab:PSIMIEEs}
\end{table}

\begin{table}
\centering \caption{Selected $p^*$ SD nodes and EEs from Wiki5 and
Wiki45} \small{
\begin{tabular}{|l|l|}
\hline

\textbf{Node} & \textbf{Extent Expression (EE)} \\
\hline \hline
$w_{372}$ & /article/body/section/section/section/figure \\
$w_{199}$ & /article/body/section/p/sub \\
$w_{333}$ & /article/body/section/section/section/section \\
$w_{967}$ & /article/body/template/template/wikipedialink \\
\hline
\end{tabular}
} \label{tab:WikiEEs}
\end{table}

We evaluated two different types of refinements, one given by a
generic AxPRE ($p^*|c^*$) and the other defined by a very specific
one. Tables~\ref{tab:RSS2RefTimes} through
\ref{tab:Wiki45RefTimes} report $p^*|c^*$ refinement times for the
selected SD nodes. We choose the $p^*|c^*$ refinement to show the
performance with AxPREs involving common axes used throughout the
summary literature.

\begin{table}
\centering  \caption{RSS2 $p^*|c^*$ refinements} {\small
\begin{tabular}{|l|r|r||r|r|r|r|r|}
\hline
\multicolumn{3}{|c||}{\textbf{$p^*$ SD}} &  \multicolumn{5}{c|}{\textbf{$p^*|c^*$ Refinement}}  \\
\hline
\textbf{Node}  &  \multicolumn{2}{c||}{\textbf{Extent Size}}    &  \multicolumn{1}{|c|}{\textbf{\#}}     & \multicolumn{4}{c|}{\textbf{Times (s)}} \\
\cline{2-3} \cline{5-8}
               &  \textbf{\#Docs}  & \textbf{\#Elems}       &  \textbf{EEs}    & \textbf{V} &  \textbf{M} & \textbf{P} & \textbf{X} \\
\hline \hline
$r_{468}$      &  3296   &     3296    &  7  &    219.2    &  101.8   &  100.1 & 185.7  \\
$r_{449}$      &  6509   &    90583    & 201 &    14786.2  &  598.2   &  575.2 & 14235.2  \\
$r_{653}$      &    18   &      320    &  42 &    51.5     &    4.5   &    3.7 & 145.2  \\
$r_{452}$      &  6253   &    82022    &   3 &    358.4    &  189.6   &  185.1 & 332.7  \\
\hline
\end{tabular}
} \label{tab:RSS2RefTimes}
\end{table}

\begin{table}
\centering  \caption{PSIMI2 $p^*|c^*$ refinements} {\small
\begin{tabular}{|l|r|r||r|r|r|r|r|}
\hline
\multicolumn{3}{|c||}{\textbf{$p^*$ SD}} &  \multicolumn{5}{c|}{\textbf{$p^*|c^*$ Refinement}}  \\
\hline
\textbf{Node}  &  \multicolumn{2}{c||}{\textbf{Extent Size}}    &  \multicolumn{1}{|c|}{\textbf{\#}}     & \multicolumn{4}{c|}{\textbf{Times (s)}} \\
\cline{2-3} \cline{5-8}
               &  \textbf{\#Docs}  & \textbf{\#Elems}       &  \textbf{EEs}    & \textbf{V} &  \textbf{M} & \textbf{P} & \textbf{X}   \\
\hline \hline
$p_{59}$      &  156   &   24256    &  3  &     42.8   &  29.8   &   28.2 &  41.1 \\
$p_{18}$      &  156   &    2072    &  2  &     32.1   &  26.1   &   23.7 &  29.7 \\
$p_{24}$      &  156   &    2072    &  8  &    732.4   &  157.6  &  149.8 &  603.4 \\
$p_{193}$     &  4     &     28     &  1  &      3.9   &  3.5    &    2.0 &  2.8 \\
\hline
\end{tabular}
} \label{tab:PSIMI2RefTimes}
\end{table}

\begin{table}
\centering  \caption{Wiki5 $p^*|c^*$ refinements} {\small
\begin{tabular}{|l|r|r||r|r|r|r|r|}
\hline
\multicolumn{3}{|c||}{\textbf{$p^*$ SD}} &  \multicolumn{5}{c|}{\textbf{$p^*|c^*$ Refinement}}  \\
\hline
\textbf{Node}  &  \multicolumn{2}{c||}{\textbf{Extent Size}}    &  \multicolumn{1}{|c|}{\textbf{\#}}     & \multicolumn{4}{c|}{\textbf{Times (s)}} \\
\cline{2-3} \cline{5-8}
               &  \textbf{\#Docs}  & \textbf{\#Elems}       &  \textbf{EEs}    & \textbf{V} &  \textbf{M} & \textbf{P} & \textbf{X}   \\
\hline \hline
$w_{372}$      &  252   &     522    &  16  &     295.8    &  26.3   &  25.7 &  10.1 \\
$w_{199}$      &  463   &    2194    &   4  &     448.4    &  33.8   &  29.9 &  3.4 \\
$w_{333}$      &  128   &     500    &  61  &    2138.9    &  87.8   &  79.1 &  308.2 \\
$w_{967}$      &  155   &     241    &   6  &     235.9    &  12.9   &  10.4 &  2.3 \\
\hline
\end{tabular}
} \label{tab:Wiki5RefTimes}
\end{table}

\begin{table}
\centering  \caption{Wiki45 $p^*|c^*$ refinements} {\small
\begin{tabular}{|l|r|r||r|r|r|r|r|}
\hline
\multicolumn{3}{|c||}{\textbf{$p^*$ SD}} &  \multicolumn{5}{c|}{\textbf{$p^*|c^*$ Refinement}}  \\
\hline
\textbf{Node}  &  \multicolumn{2}{c||}{\textbf{Extent Size}}    &  \multicolumn{1}{|c|}{\textbf{\#}}     & \multicolumn{4}{c|}{\textbf{Times (s)}} \\
\cline{2-3} \cline{5-8}
               &  \textbf{\#Docs}  & \textbf{\#Elems}       &  \textbf{EEs}    & \textbf{V} &  \textbf{M} & \textbf{P} & \textbf{X}   \\
\hline \hline
$w_{372}$      &  898    &   2166    &   37  &     1449.1   &  455.2  &  446.1  & 144.1   \\
$w_{199}$      &  1479   &   6963    &   14  &     2493.5   &  773.2  &  748.7  & 64.8 \\
$w_{333}$      &  736    &   3714    &  203  &    12813.4   &  574.7  &  573.5  & 6602.3 \\
$w_{967}$      &  2330   &   3662    &    8  &     1835.9   &  569.4  &  552.3  & 20.6  \\
\hline
\end{tabular}
} \label{tab:Wiki45RefTimes}
\end{table}

\newpage

Tables~\ref{tab:RSS2RefTimes} through \ref{tab:Wiki45RefTimes} are
divided into two parts, the first half provides information on the
number of documents and elements in the extent of the $p^*$ SD
nodes being refined (\textbf{\# Docs} and \textbf{\# Elems}
columns, respectively) , and the second half contains numbers
relative to the $p^*|c^*$ refinement itself. The numbers under
\textbf{\# Docs} indicate how many documents need to be opened to
evaluate the refinement. The number of new SD nodes created by the
refinements (which is the same as the number of EEs evaluated) are
reported in the \textbf{\# EEs} columns. For instance, the
$p^*|c^*$ refinement partitions node $r_{449}$ into 201 new SD
nodes, which means that 201 XPath expressions have to be evaluated
in 6509 documents in order to obtain the $p^*|c^*$ of node
$r_{449}$. In general, refinement times increase proportionally to
the number of documents that need to be opened for computing the
refinement.

We consider two scenarios, one in which extents are materialized
in the {\tt ElemDB} table (reported under columns \textbf{V},
\textbf{M} and \textbf{P}), and another in which the extents are
virtual and are thus represented only by the EEs (reported under
columns \textbf{X}). Times reported in \textbf{V}, \textbf{M} and
\textbf{P} columns comprise locating the affected files using the
SD, opening them and evaluating the EE in order to update the
materialized extent information in the {\tt ElemDB} table. In
addition to extent updates, columns \textbf{V} and \textbf{M}
times include edge computations using Algorithms
\ref{alg:refineVirtual} and \ref{alg:refineMaterialized},
respectively. (The labels \textbf{V} and \textbf{M}, which stand
for ``virtual'' and ``materialized'', refers only to the different
approaches to \emph{edge} computation). In contrast, times under
the \textbf{P} column correspond to extent computation only
(without edges). Comparing column \textbf{P} against columns
\textbf{V} and \textbf{M} gives us an idea of how much overhead
\dx\ incurs on the edges. Finally, the \textbf{X} column displays
how long it takes just to obtain the expressions for both edges
and extents under the virtual approach. Thus, the \textbf{X}
column corresponds to a ``purely virtual'' approach in which no
materialization is used for either edges nor extents. Since edges
are computed from the EEs, the SD graph is still maintained.

The time differences between the \textbf{V} and \textbf{M} columns
come from the fact that computing the edges between the new SD
nodes using XPath is usually more costly than computing them from
the information stored in the {\tt ElemDB} table. However, we are
not aware of any technique for computing general XPath expressions
from the region encodings in the {\tt ElemDB} table, so using just
the materialized extents is not always possible.

\begin{table}
\centering  \caption{RSS2 AxPRE refinements} {\small
\begin{tabular}{|l||l|r|r|r|}
\hline
\multicolumn{1}{|c||}{\textbf{$p^*$ SD}}&  \multicolumn{1}{|c|}{\textbf{Refining}} &  \multicolumn{2}{c|}{\textbf{Resulting Extent}}   &  \multicolumn{1}{c|}{\textbf{Times (s)}} \\
\cline{3-5}
\multicolumn{1}{|c||}{\textbf{Node}}    &   \multicolumn{1}{|c|}{\textbf{AxPRE}}    &  \textbf{\#Docs}  & \textbf{\#Elems}      &  \multicolumn{1}{|c|}{\textbf{P}}   \\
\hline \hline
$r_{468}$      & c[title].fs[url].fs[link].fs[width]      &     172    &  172   &    3.9     \\
               & .fs[height].fs[description]              &            &        &            \\
\hline
$r_{449}$      & c[enclosure].fs[enclosure].fs[enclosure] &       9    &  37    &    10.8    \\
\hline
$r_{653}$      & fc[p].ns[p].ns[img]                      &       6    &  26    &     0.4    \\
\hline
$r_{452}$      & fs[link]                                 &     688    &  13885 &    10.1    \\
\hline
\end{tabular}
} \label{tab:RSS2AxPRERefTimes}
\end{table}

\begin{table}
\centering  \caption{PSIMI2 AxPRE refinements} {\small
\begin{tabular}{|l||l|r|r|r|}
\hline
\multicolumn{1}{|c||}{\textbf{$p^*$ SD}}&  \multicolumn{1}{|c|}{\textbf{Refining}} &  \multicolumn{2}{c|}{\textbf{Resulting Extent}}   &  \multicolumn{1}{c|}{\textbf{Times (s)}} \\
\cline{3-5}
\multicolumn{1}{|c||}{\textbf{Node}}    &   \multicolumn{1}{|c|}{\textbf{AxPRE}}    &  \textbf{\#Docs}  & \textbf{\#Elems}      &  \multicolumn{1}{|c|}{\textbf{P}}   \\
\hline \hline
$p_{59}$      & c[name].fs[cellType]           &      2    &  14   &    27.3    \\
\hline
$p_{18}$      & c[primaryRes].fs[secondaryRef] &     12    &  20   &    29.9    \\
\hline
$p_{24}$      & c[name].ns[cellType].ns[tissue]&      4    &  4    &    27.1    \\
\hline
$p_{193}$     & c[name].ns[xref]               &      4    &  28   &     3.4    \\
\hline
\end{tabular}
} \label{tab:PSIMI2AxPRERefTimes}
\end{table}

\begin{table}
\centering  \caption{Wiki5 AxPRE refinements} {\small
\begin{tabular}{|l||l|r|r|r|}
\hline
\multicolumn{1}{|c||}{\textbf{$p^*$ SD}}&  \multicolumn{1}{|c|}{\textbf{Refining}} &  \multicolumn{2}{c|}{\textbf{Resulting Extent}}   &  \multicolumn{1}{c|}{\textbf{Times (s)}} \\
\cline{3-5}
\multicolumn{1}{|c||}{\textbf{Node}}    &   \multicolumn{1}{|c|}{\textbf{AxPRE}}    &  \textbf{\#Docs}  & \textbf{\#Elems}      &  \multicolumn{1}{|c|}{\textbf{P}}   \\
\hline \hline
$w_{372}$      & c[caption].c[collectionLink]   &     2      &  2     &    1.9     \\
               &   .fs[br].fs[collectionLink]   &            &        &            \\
\hline
$w_{199}$      & c[sub].c[sub].fs[sub]          &     1      &  1     &    2.1    \\
\hline
$w_{333}$      & c[title].fs[p].fs[p].fs[p]     &     39     &  79    &    1.1    \\
\hline
$w_{967}$      & c[br]$|$fs[collectionLink]     &     4      &  6     &    1.2    \\
               &   .fs[collectionLink]          &            &        &            \\
\hline
\end{tabular}
} \label{tab:Wiki5AxPRERefTimes}
\end{table}

\begin{table}
\centering  \caption{Wiki45 AxPRE refinements} {\small
\begin{tabular}{|l||l|r|r|r|}
\hline
\multicolumn{1}{|c||}{\textbf{$p^*$ SD}}&  \multicolumn{1}{|c|}{\textbf{Refining}} &  \multicolumn{2}{c|}{\textbf{Resulting Extent}}   &  \multicolumn{1}{c|}{\textbf{Times (s)}} \\
\cline{3-5}
\multicolumn{1}{|c||}{\textbf{Node}}    &   \multicolumn{1}{|c|}{\textbf{AxPRE}}    &  \textbf{\#Docs}  & \textbf{\#Elems}      &  \multicolumn{1}{|c|}{\textbf{P}}   \\
\hline \hline
$w_{372}$      & c[caption].c[collectionLink]   &     3      &  3     &    33.1    \\
               &   .fs[br].fs[collectionLink]   &            &        &            \\
\hline
$w_{199}$      & c[sub].c[sub].fs[sub]          &     3      &  3     &    39.0    \\
\hline
$w_{333}$      & c[title].fs[p].fs[p].fs[p]     &     155    &  320   &    28.2    \\
\hline
$w_{967}$      & c[br]$|$fs[collectionLink]     &     9      &  11    &    57.3    \\
               &   .fs[collectionLink]          &            &        &            \\
\hline
\end{tabular}
} \label{tab:Wiki45AxPRERefTimes}
\end{table}

Tables~\ref{tab:RSS2AxPRERefTimes} through
\ref{tab:Wiki45AxPRERefTimes} report refinements that were chosen
to study SDs involving novel axes (e.g., $\mathit{fc}$,
$\mathit{fs}$, $ns$) and more expressive AxPREs with label
predicates. The tables show the refining AxPRE for each $p^*$ SD
node, the number of documents and elements that contain
neighbourhoods matching the entire AxPRE (\textbf{\# Docs} and
\textbf{\# Elems} columns, respectively), together with how long
it takes to compute the extent (\textbf{Times} column). For any
given expression, the number of elements with either empty
neighbourhoods or matching prefixes of the AxPRE is the complement
of the number reported under \textbf{\# Elems}. For instance, the
$r_{449}$ row of Table \ref{tab:RSS2AxPRERefTimes} indicates that
37 elements in 9 documents have exact
$c[enclosure].fs[enclosure].fs[enclosure]$ neighbourhoods and
obtaining them from the $r_{449}$ extent takes 10.8 seconds. In
addition, we know that the number of elements either matching
prefixes or with empty neighbourhoods is 90546, which comes from
the number in column \textbf{\# Elems} and row $r_{449}$ in Table
\ref{tab:RSS2RefTimes} (90583) minus the number in column
\textbf{\# Elems} and row $r_{449}$ in Table
\ref{tab:RSS2AxPRERefTimes} (37). Such subtraction would not be
meaningful for the \textbf{\# Docs} columns because the same
document may contain elements in different extents (remember that
an SD contains a partition of elements, not documents, so document
extents may overlap).

These results suggest that, even though computing generic
refinements like $p^*|c^*$ may be expensive, more specific
refinements can be performed in less than a minute and many of
them in just a few seconds for the smaller test collections.

\section{Edge stabilization} \label{sec:stabResults}

In this section, we report experimental results for stabilization
of SD edges from our selected $p^*$ nodes.

Tables~\ref{tab:RSS2StabilizationTimes} through
\ref{tab:Wiki45StabilizationTimes} report edge stabilization times
and extent sizes for the selected SD nodes. The edge stabilized is
indicated in the tables by an AxPRE containing the axis and the
label of the target node. The four \textbf{Resulting Extents}
columns show the number of document and elements that do contain
the edge and the number of those that do not. The times reported
under columns \textbf{V }and \textbf{M} correspond to the
materialized extent approach with edge computation using EEs (the
former) and the {\tt ElemDB} table (the latter), as explained in
the previous section for refinements.

\begin{table}
\centering  \caption{RSS2 edge stabilization} {\small
\begin{tabular}{|l|l||r|r|r|r|r|r|}
\hline
\multicolumn{1}{|c|}{\textbf{$p^*$ SD}}  &  \multicolumn{1}{|c||}{\textbf{Edge}}    &  \multicolumn{4}{c|}{\textbf{Resulting Extents}}                  &  \multicolumn{2}{c|}{\textbf{Times}} \\
\cline{3-6}
\multicolumn{1}{|c|}{\textbf{Node}}  &  \multicolumn{1}{|c||}{\textbf{Stabilized}}  &  \multicolumn{2}{c|}{\textbf{With Edge}} &  \multicolumn{2}{c|}{\textbf{Without Edge}}   & \multicolumn{2}{|c|}{\textbf{(s)}}                    \\
\cline{3-8}
                                      &                                             &  \textbf{\#Docs}  & \textbf{\#Elems}     &  \textbf{\#Docs}  & \textbf{\#Elems}           & \multicolumn{1}{|c|}{\textbf{V}}  & \multicolumn{1}{|c|}{\textbf{M}}   \\
\hline \hline
$r_{468}$      & c[description] &     492  &  492   &  2804  &  2804  &  4.1    &  0.5    \\
$r_{468}'$     & c[link]        &    2792  &  2792  &  12    &  12    &  4.5    &  0.2    \\
\hline
$r_{449}$      & ps[item]       &    6263  &  84063 &  6509  &  6520  &  12.5   &  2.9    \\
$r_{449}'$     & c[body]        &    15    &  15    &  6494  &  6505  &  10.9   &  3.7    \\
\hline
$r_{653}$      & d[img]         &    12    &  12    &  10    &  201   &  0.5    &  0.3    \\
$r_{653}'$     & d[table]       &     7    &   7    &  3     &  14    &  0.4    &  0.2    \\
\hline
$r_{452}$      & c[br]          &     12   &  12    &  6249  &  81968 &  12.4   &  5.9    \\
\hline
\end{tabular}
} \label{tab:RSS2StabilizationTimes}
\end{table}

\begin{table}
\centering  \caption{PSIMI2 edge stabilization} {\small
\begin{tabular}{|l|l||r|r|r|r|r|r|}
\hline
\multicolumn{1}{|c|}{\textbf{$p^*$ SD}}  &  \multicolumn{1}{|c||}{\textbf{Edge}}    &  \multicolumn{4}{c|}{\textbf{Resulting Extents}}                  &  \multicolumn{2}{c|}{\textbf{Times}} \\
\cline{3-6}
\multicolumn{1}{|c|}{\textbf{Node}}  &  \multicolumn{1}{|c||}{\textbf{Stabilized}}  &  \multicolumn{2}{c|}{\textbf{With Edge}} &  \multicolumn{2}{c|}{\textbf{Without Edge}}   & \multicolumn{2}{|c|}{\textbf{(s)}}                    \\
\cline{3-8}
                                      &                                             &  \textbf{\#Docs}  & \textbf{\#Elems}     &  \textbf{\#Docs}  & \textbf{\#Elems}           & \multicolumn{1}{|c|}{\textbf{V}}  & \multicolumn{1}{|c|}{\textbf{M}}   \\
\hline \hline
$p_{59}$       & c[cellType]    &     4    &  28  &  156    &  24228  &  38.2    &  9.8    \\
\hline
$p_{18}$       & c[secondaryRef]&    12    &  20  &  148    &  2052   &  25.6    &  1.2    \\
\hline
$p_{24}$       & c[tissue]      &     8    &  84  &  156    &  1988   &  25.4    &  1.3    \\
$p_{24}'$      & c[cellType]    &     8    & 548  &  156    &  1440   &  23.4    &  1.2    \\
\hline
\end{tabular}
} \label{tab:PSIMI2StabilizationTimes}
\end{table}

We stabilize two different edges for some $p^*$ SD nodes. After
one edge stabilization, the resulting SD node that does not have
the stabilized edge is indicated by the SID with an apostrophe.
The second edge stabilized always corresponds to a node with an
apostrophe from the previous stabilization. For instance, the
first edge stabilized from node $r_{449}$ (Table
\ref{tab:RSS2StabilizationTimes}) was the $ps$ edge to an $item$
node, which resulted in two SD nodes: one containing a
\emph{stable} $ps$ edge with 84063 elements in its extent, and
another one ($r_{449}'$) with \emph{no edge} and 6520 elements.
From node $r_{449}'$ we stabilize then the $c$ edge to a $body$
node obtaining again two nodes: one with a stable $c$ edge with 15
elements in its extent, and the other one with 6505 elements and
no edge. The time for computing the $ps$ edge stabilization is
12.5 seconds when computing the edges with EEs, and 2.9 seconds
when using the {\tt ElemDB} table. The times for the $c$ edge
stabilization are 10.9 and 3.7 seconds respectively.

\begin{table}
\centering  \caption{Wiki5 edge stabilization} {\small
\begin{tabular}{|l|l||r|r|r|r|r|r|}
\hline
\multicolumn{1}{|c|}{\textbf{$p^*$ SD}}  &  \multicolumn{1}{|c||}{\textbf{Edge}}    &  \multicolumn{4}{c|}{\textbf{Resulting Extents}}                  &  \multicolumn{2}{c|}{\textbf{Times}} \\
\cline{3-6}
\multicolumn{1}{|c|}{\textbf{Node}}  &  \multicolumn{1}{|c||}{\textbf{Stabilized}}  &  \multicolumn{2}{c|}{\textbf{With Edge}} &  \multicolumn{2}{c|}{\textbf{Without Edge}}   & \multicolumn{2}{|c|}{\textbf{(s)}}                    \\
\cline{3-8}
                                      &                                             &  \textbf{\#Docs}  & \textbf{\#Elems}     &  \textbf{\#Docs}  & \textbf{\#Elems}           & \multicolumn{1}{|c|}{\textbf{V}}  & \multicolumn{1}{|c|}{\textbf{M}}   \\
\hline \hline
$w_{372}$      & d[collectionLink] &  335  &  592   &  695  &  1574  &  34.7  &  2.3    \\
$w_{372}'$     & d[small]          &    3  &  5     &  694  &  1569  &  33.9  &  2.2    \\
\hline
$w_{199}$      & c[sub]            &    28  &  33    &  1469  &  6930 &  41.7  &  5.4    \\
$w_{199}'$     & c[small]          &    18  &  83    &  1454  &  6847 &  43.7  &  5.4    \\
\hline
$w_{333}$      & c[outsideLink]    &    34  &  83    &  724  &  3631  &  31.9  &  3.6    \\
$w_{333}'$     & c[unknownLink]    &    68 &  131    &  705  &  3500  &  29.6 &  3.8    \\
\hline
$w_{967}$      & c[template]       &    26  &  27     &  2304  &  3635  &  61.2  &  3.5    \\
$w_{967}'$     & c[sup]            &    174 &  246   &   2130  &  3389  &  60.2  &  3.4    \\
\hline
\end{tabular}
} \label{tab:Wiki5StabilizationTimes}
\end{table}

\begin{table}
\centering  \caption{Wiki45 edge stabilization} {\small
\begin{tabular}{|l|l||r|r|r|r|r|r|}
\hline
\multicolumn{1}{|c|}{\textbf{$p^*$ SD}}  &  \multicolumn{1}{|c||}{\textbf{Edge}}    &  \multicolumn{4}{c|}{\textbf{Resulting Extents}}                  &  \multicolumn{2}{c|}{\textbf{Times}} \\
\cline{3-6}
\multicolumn{1}{|c|}{\textbf{Node}}  &  \multicolumn{1}{|c||}{\textbf{Stabilized}}  &  \multicolumn{2}{c|}{\textbf{With Edge}} &  \multicolumn{2}{c|}{\textbf{Without Edge}}   & \multicolumn{2}{|c|}{\textbf{(s)}}                    \\
\cline{3-8}
                                      &                                             &  \textbf{\#Docs}  & \textbf{\#Elems}     &  \textbf{\#Docs}  & \textbf{\#Elems}           & \multicolumn{1}{|c|}{\textbf{V}}  & \multicolumn{1}{|c|}{\textbf{M}}   \\
\hline \hline
$r_{372}$      & d[collectionLink] &   125 &  207   &  169  &  315  &  2.4  &  0.6    \\
$r_{372}'$     & d[small]          &    2  &  4     &  169  &  311  &  2.2  &  0.5    \\
\hline
$r_{199}$      & c[sub]            &    3  &  3     &  462  &  2191 &  2.7  &  0.8    \\
$r_{199}'$     & c[small]          &    5  &  35    &  458  &  2156 &  2.6  &  0.8    \\
\hline
$r_{333}$      & c[outsideLink]    &    7  &  12    &  126  &  488  &  1.7  &  0.7    \\
$r_{333}'$     & c[unknownLink]    &    10 &  14    &  123  &  474  &  1.4  &  0.6    \\
\hline
$r_{967}$      & c[template]       &    4  &  5     &  151  &  236  &  1.2  &  0.6    \\
$r_{967}'$     & c[sup]            &    66 &  123   &   85  &  113  &  1.4  &  0.5    \\
\hline
\end{tabular}
} \label{tab:Wiki45StabilizationTimes}
\end{table}

Our results show that \dx\ can provide interactive response times
(from sub second to just a few seconds) for all edge
stabilizations tested when using the materialized approach for
both extents and edges. Moreover, when using the more expensive
EE-based approach for finding the SD edges, we still obtain
response times in the order of a minute in the vast majority of
test cases. This is compelling evidence that \dx\ can be used in
scenarios in which SDs need to be manipulated interactively in
order to selectively explore the structure of an XML collection
(e.g., aggregating thousands of RSS feed from dozens of content
providers).

\section{XPath query evaluation using SDs}
\label{sec:queryevalwithSDs}

In this section, we provide performance results for obtaining
answer documents for several XPath queries using a variety of SDs.
These results considerably expand the preliminary study presented
in \cite{CR07}.

Tables~\ref{tab:RSSQueries}, \ref{tab:PSIMIQueries}, and
\ref{tab:WikiQueries} show the twelve queries in our benchmark
(the structural subqueries appear in black, the non-structural
predicates are in grey). These queries were selected to show how
the system scales with respect to key query parameters like answer
size and number of candidate documents (those that provide a
non-empty answer for the structural subquery). Our benchmark
focuses on the navigational features of XPath, following the
approach of the MemBeR XQuery Micro-Benchmark \cite{AMM05}, which
provides a form of standardization for studying XQuery evaluation.

\begin{table}
\centering \caption{RSS collection queries} \small{
\begin{tabular}{|c|l|l|}
\hline

\textbf{Query} &  \textbf{XPath Expression} \\
\hline \hline
R1 & \texttt{/rss/channel/image[title/following-sibling::url/following-sibling::} \\
   & \texttt{link/following-sibling::width/following-sibling::height} \\
   & \texttt{/following-sibling::description]\textcolor[rgb]{0.50,0.50,0.50}{[width $<$ height]}} \\
\hline
R2 &
\texttt{/rss/channel/item[enclosure][enclosure/following-sibling::enclosure} \\
   & \texttt{/following-sibling::enclosure]\textcolor[rgb]{0.50,0.50,0.50}{[enclosure/@type='audio/mpeg']}} \\
\hline
R3 & \texttt{/rss/channel/item/body[child::*[1][self::p]/following-sibling::*[1]} \\
   & \texttt{[self::p]/following-sibling::*[1][self::img]]\textcolor[rgb]{0.50,0.50,0.50}{[img[width=height]]}} \\
\hline
R4 & \texttt{/rss/channel/item/description[following-sibling::link]} \\
   & \texttt{\textcolor[rgb]{0.50,0.50,0.50}{[contains(.,'2005')]} }\\
\hline
\end{tabular} } \label{tab:RSSQueries}
\end{table}

\begin{table}
\centering \caption{PSIMI collection queries} \small{
\begin{tabular}{|c|l|l|}
\hline
\textbf{Query} &  \textbf{XPath Expression} \\
\hline \hline
P1 & \texttt{/entrySet/entry/interactorList/interactor/organism[names} \\
   & \texttt{/following-sibling::cellType]\textcolor[rgb]{0.50,0.50,0.50}{[contains(.,'Cercopithecus')]}} \\
\hline P2 &
\texttt{/entrySet/entry/experimentList/experimentDescription/bibref/xref} \\
   & \texttt{[primaryRef/following-sibling::secondaryRef]} \\
   & \texttt{\textcolor[rgb]{0.50,0.50,0.50}{[secondaryRef/@refType='method reference']}} \\
\hline
P3 & \texttt{/entrySet/entry/experimentList/experimentDescription} \\
   & \texttt{/hostOrganismList/hostOrganism[child::names/following-sibling::*[1]}  \\
   & \texttt{[self::cellType]/following-sibling::*[1][self::tissue]]} \\
   & \texttt{\textcolor[rgb]{0.50,0.50,0.50}{[tissue[contains(.,'endothelium')]]}} \\
\hline
P4 & \texttt{/entrySet/entry/interactorList/interactor/organism/cellType[names} \\
   & \texttt{/following-sibling::*[1][self::xref]]\textcolor[rgb]{0.50,0.50,0.50}{[contains(.,'Cercopithecus')]}} \\
\hline
\end{tabular}
} \label{tab:PSIMIQueries}
\end{table}

\begin{table}
\centering \caption{Wikipedia collections queries} \small{
\begin{tabular}{|c|l|l|}
\hline

\textbf{Query} &  \textbf{XPath Expression} \\
\hline \hline
W1 & \texttt{/article/body/section/section/section/figure[caption/collectionlink} \\
   & \texttt{/following-sibling::br/following-sibling::collectionlink]} \\
   & \texttt{\textcolor[rgb]{0.50,0.50,0.50}{[contains(.,'Loutherbourg')]}} \\
\hline
W2 & \texttt{/article/body/section/p/sub[child::sub/child::sub} \\
   & \texttt{/following-sibling::sub]\textcolor[rgb]{0.50,0.50,0.50}{[sub/sub='2']}} \\
\hline
W3 & \texttt{/article/body/section/section/section/section[child::title}  \\
   & \texttt{/following-sibling::p/following-sibling::p/following-sibling::p]} \\
   & \texttt{\textcolor[rgb]{0.50,0.50,0.50}{[contains(.,'extinction')]}} \\
\hline
W4 & \texttt{/article/body/template/template/wikipedialink[following-sibling::} \\
   & \texttt{collectionlink]\textcolor[rgb]{0.50,0.50,0.50}{[contains(.,'William de Longespee')]}} \\
\hline
\end{tabular}
} \label{tab:WikiQueries}
\end{table}

Tables~\ref{tab:RSS2QueryTimes} through \ref{tab:Wiki45QueryTimes}
show the times for obtaining the answer documents and evaluating
the queries in our collections using a variety of SDs. The
\textbf{SD} column indicates the type of SD used to obtain the
candidate documents (next column) on which the entire query is
evaluated. The three columns under \textbf{Answer} show the time
it takes to evaluate the query in the candidate documents, and the
number of documents and elements in the final answer. Since these
are XPath queries, the number of documents and elements returned
by each query are independent of the SD used for evaluation.

\begin{table}
\centering  \caption{RSS2 query results and times} {\small
\begin{tabular}{|c|l|r||r|r|r|}
\hline
\multicolumn{1}{|c|}{\textbf{Query}}  & \multicolumn{1}{|c|}{\textbf{SD}} & \multicolumn{1}{|c||}{\textbf{Candidate}} & \multicolumn{3}{c|}{\textbf{Answer}} \\
\cline{3-6}
                                      &                                   & \multicolumn{1}{|c||}{\textbf{\# Docs}}      & \multicolumn{1}{c|}{\textbf{Times (s)}} & \multicolumn{1}{c|}{\textbf{\# Docs}}  &  \multicolumn{1}{c|}{\textbf{\# Elems}} \\

\hline \hline
     & label                 &  3518   & 7.7     &      &   \\
\cline{2-4}
R1   & $p^*$ ($r_{468}$)     &  3296   & 7.4     & 79   & 79 \\
\cline{2-4}
     & $p^*|c^*$             &  387    & 1.3     &      &   \\
\cline{2-4}
     & specific              &  172    & 0.6     &      &   \\
\hline \hline
     & label                 &  8122   & 19.9    &      &   \\
\cline{2-4}
R2   & $p^*$ ($r_{449}$)     &  6509   & 15.1    & 6    & 32   \\
\cline{2-4}
     & $p^*|c^*$             &  181    & 1.2     &      &   \\
\cline{2-4}
     & specific              &  9      & 0.1     &      &   \\
\hline \hline
     & label                 &  31   &  0.4    &      &    \\
\cline{2-4}
R3   & $p^*$ ($r_{653}$)     &  18   &  0.3    & 6    & 26 \\
\cline{2-4}
     & $p^*|c^*$             &  15   &  0.3    &      &    \\
\cline{2-4}
     & specific              &  6    &  0.1    &      &    \\
\hline \hline
     & label                 &  8221  & 19.7   &      &    \\
\cline{2-4}
R4   & $p^*$ ($r_{452}$)     &  6253  & 14.1   & 241  & 1344 \\
\cline{2-4}
     & $p^*|c^*$             &  6253  & 14.6   &      &    \\
\cline{2-4}
     & specific              &  688   & 2.0    &      &    \\
\hline
\end{tabular}
} \label{tab:RSS2QueryTimes}
\end{table}

\begin{table}
\centering  \caption{PSIMI2 query results and times} {\small
\begin{tabular}{|c|l|r||r|r|r|}
\hline
\multicolumn{1}{|c|}{\textbf{Query}}  & \multicolumn{1}{|c|}{\textbf{SD}} & \multicolumn{1}{|c||}{\textbf{Candidate}} & \multicolumn{3}{c|}{\textbf{Answer}} \\
\cline{3-6}
                                      &                                   & \multicolumn{1}{|c||}{\textbf{\# Docs}}      & \multicolumn{1}{c|}{\textbf{Times (s)}} & \multicolumn{1}{c|}{\textbf{\# Docs}}  &  \multicolumn{1}{c|}{\textbf{\# Elems}} \\

\hline \hline
     & label                 &  156    & 45.9 &      &   \\
\cline{2-4}
P1   & $p^*$ ($p_{59}$)      &  156    & 45.7 &  2   &  14 \\
\cline{2-4}
     & $p^*|c^*$             &  156    & 45.7 &      &   \\
\cline{2-4}
     & specific              &    2    & 2.5  &      &   \\
\hline \hline
     & label                 &  156   & 45.7  &      &   \\
\cline{2-4}
P2   & $p^*$ ($p_{18}$)      &  156   & 45.5  &  4   & 8 \\
\cline{2-4}
     & $p^*|c^*$             &  12    & 17.1  &      &   \\
\cline{2-4}
     & specific              &  12    & 17.1  &      &   \\
\hline \hline
     & label                 &  156   & 45.2  &      &    \\
\cline{2-4}
P3   & $p^*$ ($p_{24}$)      &  156   & 44.9  &  4   &  4 \\
\cline{2-4}
     & $p^*|c^*$             &  6     & 6.5   &      &    \\
\cline{2-4}
     & specific              &  4     & 5.8   &      &    \\
\hline \hline
     & label                 &  8    & 9.8   &      &   \\
\cline{2-4}
P4   & $p^*$ ($p_{193}$)     &  4    & 4.9   & 1    & 1 \\
\cline{2-4}
     & $p^*|c^*$             &  4    & 4.8   &      &   \\
\cline{2-4}
     & specific              &  4    & 4.9   &      &   \\
\hline
\end{tabular}
} \label{tab:PSIMI2QueryTimes}
\end{table}

Each row of \textbf{SD}, \textbf{\# Candidate Docs} and
\textbf{Times} corresponds to a different SD used for evaluating
the query. The ``label'' row in each section shows the evaluation
times when using the label SD node corresponding to the element
returned by the query. For instance, query R2 returns ``item''
elements, so the extent documents used are those from the ``item''
node in the label SD (8122 documents in total), taking 19.9
seconds to evaluate the query on them. The $p^*$ rows report the
respective numbers when using the $p^*$ node whose AxPRE contains
the query (note that the SIDs from Tables \ref{tab:RSSEEs},
\ref{tab:PSIMIEEs}, and \ref{tab:WikiEEs} are indicated). For
instance, for query R2 we use node $r_{449}$ from the $p^*$ SD,
taking 15.1 seconds to evaluate the query on the 6509 documents in
the extent of $r_{449}$. Similarly, $p^*|c^*$ rows show the
evaluation times when using $p^*|c^*$ SD nodes (there may be more
than one containing the query). For instance, the $p^*|c^*$
node(s) used for query R2 have 181 documents and evaluating R2 on
them takes 1.2 seconds. Finally, the last row in each section
labeled ``specific'' shows \dx\ performance when using an AxPRE
refinement obtained from the structural subquery. For instance,
for query R2 the refining AxPRE would be
$c[enclosure].fs[enclosure].fs[enclosure]$ (row $r_{449}$ in Table
\ref{tab:RSS2AxPRERefTimes}) which has 9 documents in its extent
and evaluating R2 on them takes just 0.1 seconds. This is the
AxPRE we obtain by adapting the SD to R2.

\begin{table}
\centering  \caption{Wiki5 query results and times} {\small
\begin{tabular}{|c|l|r||r|r|r|}
\hline
\multicolumn{1}{|c|}{\textbf{Query}}  & \multicolumn{1}{|c|}{\textbf{SD}} & \multicolumn{1}{|c||}{\textbf{Candidate}} & \multicolumn{3}{c|}{\textbf{Answer}} \\
\cline{3-6}
                                      &                                   & \multicolumn{1}{|c||}{\textbf{\# Docs}}      & \multicolumn{1}{c|}{\textbf{Times (s)}} & \multicolumn{1}{c|}{\textbf{\# Docs}}  &  \multicolumn{1}{c|}{\textbf{\# Elems}} \\

\hline \hline
     & label                 &  13288  &  54.3  &      &    \\
\cline{2-4}
W1   & $p^*$ ($w_{372}$)     &  242    &  2.5   &   1  &  1  \\
\cline{2-4}
     & $p^*|c^*$             &  5      &  0.2   &      &    \\
\cline{2-4}
     & specific              &  2      &  0.2   &      &    \\
\hline \hline
     & label                 &  1336   &  5.9   &      &   \\
\cline{2-4}
W2   & $p^*$ ($w_{199}$)     &  463    &  2.2   &  1   & 1  \\
\cline{2-4}
     & $p^*|c^*$             &  1      &  0.2   &      &   \\
\cline{2-4}
     & specific              &  1      &  0.2   &      &   \\
\hline \hline
     & label                 &  25192  &  97.8   &      &   \\
\cline{2-4}
W3   & $p^*$ ($w_{333}$)     &  128    &   1.4   & 1    &  1 \\
\cline{2-4}
     & $p^*|c^*$             &  92     &   1.1   &      &    \\
\cline{2-4}
     & specific              &  39     &   0.6   &      &    \\
\hline \hline
     & label                 &  5370   & 25.7  &      &   \\
\cline{2-4}
W4   & $p^*$ ($w_{967}$)     &  155    & 1.4   & 1    & 1 \\
\cline{2-4}
     & $p^*|c^*$             &  155    & 1.5   &      &   \\
\cline{2-4}
     & specific              &  4      & 0.2   &      &   \\
\hline
\end{tabular}
} \label{tab:Wiki5QueryTimes}
\end{table}

Not surprisingly, our results indicate that query evaluation
performance gains are heavily dependant on both the query and the
collection. In some cases, just having the label SD is description
enough and provides good performance, whereas the label SD is not
of much help in others. For instance, using the most specific SD
for PSIMI2 query P4 (Table \ref{tab:PSIMI2QueryTimes}) only
reduces query evaluation time by less than 50\% over the label SD.
At the other end of the spectrum, using the most specific SD for
query W1 on Wiki45 (Table \ref{tab:Wiki45QueryTimes}) produces a
performance improvement of almost four orders of magnitude, going
from half an hour (label SD) to sub-second (specific SD)
evaluation time. In that same table, there are also cases (like
query W3) in which a $p^*$ by itself provides a big gain, whereas
the most specific SD only brings a modest further improvement. In
contrast, query W4 gets the greatest gain from the most specific
SD (over two orders of magnitude against both the $p^*$ and the
$p^*|c^*$ SDs).

\begin{table}
\centering  \caption{Wiki45 query results and times} {\small
\begin{tabular}{|c|l|r||r|r|r|}
\hline
\multicolumn{1}{|c|}{\textbf{Query}}  & \multicolumn{1}{|c|}{\textbf{SD}} & \multicolumn{1}{|c||}{\textbf{Candidate}} & \multicolumn{3}{c|}{\textbf{Answer}} \\
\cline{3-6}
                                      &                                   & \multicolumn{1}{|c||}{\textbf{\# Docs}}      & \multicolumn{1}{c|}{\textbf{Times (s)}} & \multicolumn{1}{c|}{\textbf{\# Docs}}  &  \multicolumn{1}{c|}{\textbf{\# Elems}} \\

\hline \hline
     & label                 &  182598  & 1775.8  &     &   \\
\cline{2-4}
W1   & $p^*$ ($w_{372}$)     &  898     & 30.9    & 1   & 1 \\
\cline{2-4}
     & $p^*|c^*$             &  7       & 0.3     &     &   \\
\cline{2-4}
     & specific              &  3       & 0.2     &     &   \\
\hline \hline
     & label                 &  7341  & 131.5   &     &   \\
\cline{2-4}
W2   & $p^*$ ($w_{199}$)     &  1479  &  37.8   & 1   & 1 \\
\cline{2-4}
     & $p^*|c^*$             &  13    &  0.8    &     &   \\
\cline{2-4}
     & specific              &  3     &  0.3    &     &   \\
\hline \hline
     & label                 &  459296 & 3541.0 &     &   \\
\cline{2-4}
W3   & $p^*$ ($w_{333}$)     &  736    & 25.8   & 1   & 1 \\
\cline{2-4}
     & $p^*|c^*$             &  442    & 16.5   &     &   \\
\cline{2-4}
     & specific              &  155    & 5.4    &     &   \\
\hline \hline
     & label                 &  61183  & 872.9  &     &    \\
\cline{2-4}
W4   & $p^*$ ($w_{967}$)     &  2330   & 65.1   & 1   &  1 \\
\cline{2-4}
     & $p^*|c^*$             &  2330   & 67.3   &     &    \\
\cline{2-4}
     & specific              &  9      & 0.4    &     &    \\
\hline
\end{tabular}
} \label{tab:Wiki45QueryTimes}
\end{table}

These results show that, even though creating the most refined SD
is not always valuable, having the right SD for the right query
does have an important impact on the overall performance, and \dx\
provides a powerful mechanism for defining and creating them.

\subsection{Comparison with summary proposals}

The results in Tables~\ref{tab:RSS2QueryTimes} through
\ref{tab:Wiki45QueryTimes} also provide a comparison with the
summary literature. Proposals like 1-index \cite{MS99}, APEX
\cite{CMS02}, A(k)-index \cite{KSBG02}, and D(k)-index
\cite{QLO03} can provide, at best, a description equivalent to the
$p^*$ SD and thus a similar performance to that reported on the
first row of each query. The $p^*|c^*$ rows give an indication of
the performance provided by the F+B-Index \cite{KBNK02}. \dx\ can
create SDs tailored to a workload that yield query evaluation
times one to three orders of magnitude faster than these proposals
(last row of each query). Using a precise SD can have a
significant impact on both candidate and answer documents
selection, and thus on overall query evaluation. Note that no
summary in the literature (even recent proposals that cluster
together nodes with the same subtree structure \cite{BCF+05}) can
capture AxPREs such as $c|fs.fs$ or $fc.ns$.

In addition, we compared \dx's initial construction time against
an open-source XML summarization tool, XSum \cite{ABMP08}, which
constructs an annotated $p^*$ SD graph (a dataguide). Table
\ref{tab:Comparative} shows comparable results for SD graph
construction times between \dx\ and XSum. We restricted the
comparison to SD graph construction times because XSum does not
store either the materialized extents or the EEs; it only creates
a $p^*$ SD graph. To the best of our knowledge, this is the only
structural summarization system publicly available. Moreover, no
other work in the extensive literature on summaries \cite{GW97,
MS99,KBNK02,KSBG02,QLO03,BCF+05,PG06b} reports construction times
for their systems.

\begin{table}
\centering \caption{System comparison: SD graph construction times
(s)} {\small
\begin{tabular}{|l|r|r|r|}
\hline
\textbf{Collection} &  \multicolumn{1}{|c|}{\textbf{Size (MB)}}    & \multicolumn{1}{|c|}{\textbf{\dx\ }}   &  \multicolumn{1}{|c|}{\textbf{XSum}}      \\
\hline \hline
XMark1               & 115             & 17.3     &   12.8         \\
XMark5               & 580             & 60.8     &   62.2         \\
XMark10              & 1150            & 118.1    &   122.1         \\
\hline
\end{tabular}
} \label{tab:Comparative}
\end{table}

Since XSum can only summarize individual files, we were not able
to test it with our benchmark collections. Thus, we decided to do
the comparative evaluation using the XMark benchmark \cite{XMark},
which creates one single file of a chosen size.

These results show that \dx\ provides SD graph construction times
comparable to an open-source structural summarization tool that is
tailored to only one particular kind of SD ($p^*$).

\subsection{Comparison with XPath evaluators}

We performed a comparative analysis against two DB systems, one
commercial (X-Hive/DB\footnote{\small
\tt{http://www.x-hive.com/products/db/}}), and the other one open
source (XQuest DB\footnote{\small
\tt{http://www.axyana.com/xquest/}}). X-Hive/DB and XQuest DB were
selected because of their good performance in published XQuery
benchmarks \cite{AFM06}. In addition, a comparison against a
Saxon\footnote{{\small \tt{http://saxon.sourceforge.net/}}}
evaluation without summaries is provided. Saxon was selected for
being a popular processor that can also evaluate XQuery and XSLT
in a file-at-a-time fashion. Saxon is the XPath processor
integrated in the \dx' default implementation  (see
Chapter~\ref{Section:Implementation}), but for this comparison we
use the XPath processor stand-alone.

Keep in mind that the selected DB-like XML processors may have
additional functionality (such as transaction processing
capabilities). The comparison aims to show that the \dx\
architecture with the default implementation (combining summaries
with Saxon) can achieve results competitive with that of XML
indexing engines, even with gigabyte sized collections. In
addition, comparing against Saxon provides a performance base line
for a file-at-a-time evaluation when the collection is stored as
XML text files in the file system and no summary structures are
available. The results confirm that, without summaries, Saxon
itself lags by several orders of magnitude. We also tried to run
our queries on DB2
v9\footnote{\small\tt{http://www-306.ibm.com/software/data/db2/9/}},
but the version we currently have does not support
following-sibling or preceding-sibling axes, so our benchmark
queries could not be run on DB2.

\begin{table}
\centering \caption{RSS2 query evaluation comparative times (s)}
{\small
\begin{tabular}{|c|r|r|r|r|r|}
\hline
\textbf{Query} & \textbf{\dx} & \textbf{X-Hive} & \textbf{XQuest} & \textbf{Saxon}\\
\hline \hline
R1          & 0.6      & 7.9      &   3.1      & 91  \\
R2          & 0.1      & 7.4      &   2.9      & 93 \\
R3          & 0.1      & 7.5      &   2.0      & 92 \\
R4          & 2.0      & 7.6      &   2.6(*)  & 92 \\
\hline
\end{tabular}
} \label{tab:Comparative2}
\end{table}

\begin{table}
\centering \caption{Wiki5 query evaluation comparative times (s)}
{\small
\begin{tabular}{|c|r|r|r|r|r|}
\hline
\textbf{Query} & \textbf{\dx} & \textbf{X-Hive} & \textbf{XQuest} & \textbf{Saxon}\\
\hline \hline
W1          & 0.2      & 25.3      &    12.2     &   337 \\
W2          & 0.2      & 26.6      &    15.7     &   342 \\
W3          & 0.6      & 24.7      &    7.5(*)  &   354 \\
W4          & 0.2      & 25.8      &    6.5(*)   &   350 \\
\hline
\end{tabular}
} \label{tab:Comparative3}
\end{table}

\vspace{.2in}

Tables \ref{tab:Comparative2} and \ref{tab:Comparative3} report
the times for selecting answer documents using \dx, X-Hive/DB ,
XQuest DB , and Saxon (without summaries) on the RSS2 and Wiki5
collections, respectively. Comparative times for Wiki45 are not
reported because neither XHive/DB nor XQuest DB could load the
entire collection. XQuest DB returned an incorrect answer for some
of the queries, which are marked with an asterisk. \dx\ times span
selecting the answer documents and evaluating the entire query
using the most refined SD (i.e., the ``specific'' AxPRE
refinements reported in Tables \ref{tab:RSS2QueryTimes},
\ref{tab:PSIMI2QueryTimes}, and \ref{tab:Wiki5QueryTimes}). These
times are obtained by adding up the times for getting the
candidate documents and the times for evaluating the entire query
on them (using Saxon).

\newpage

The extensive empirical study presented here shows that \dx's
file-at-a-time XPath evaluation architecture can be a competitive
alternative (in terms of query response times) to DB-like XML
query engines, even on gigabyte sized collections. Our
experimental results also demonstrate that \dx's powerful
mechanism for adapting summaries to a workload can provide
speedups of one to three orders of magnitude compared to other
proposals.

\chapter{Conclusion} \label{section:conc}

This thesis focuses on addressing the need to describe the actual
heterogeneous structure of web collections of XML documents.
Understanding the metadata structure of such collections is
fundamental for writing meaningful XPath queries and evaluating
them efficiently. We propose a novel framework for describing the
structure of a web collection based on highly customizable
summaries that can be conveniently tailored by axis paths regular
expressions (AxPREs).

Our main results demonstrate the scalability of the AxPRE summary
refinement and stabilization (the key enablers for tailoring
summaries) using gigabyte XML collections. In addition, \dx's
powerful mechanism for adapting summaries to a workload can
provide speedups of one to three orders of magnitude compared to
other proposals. The experiments also show that \dx's
file-at-a-time XPath evaluation architecture (supporting fast
evaluation of complex XPath workloads over large web document
collections) can be a very competitive alternative (in terms of
query response times) to DB-like XML query engines, even on
gigabyte sized collections.

Familiar research issues can be re-visited in the context of AxPRE
summaries, such as providing guidelines for selecting good
summaries (similar to schema design) and inferring general and
succinct AxPRE expressions from an XML collection (similar to DTD
inference from instances). Developing tools for metadata
management is also addressed by a recent schema summarization
proposal \cite{YJ06}. In this direction, creating summaries that
describe how metadata labels (including some generated using
schema abstraction and summarization techniques) are used in a
given instance seems promising.

In the context of XML messaging, we came across the problem of
doing schema mapping when the schemas are too general and only
very small subsets are normally used. The schema mapping problem
consists of defining correspondences between two schemas in order
to translate data from one to the other \cite{PVMH+02}. If we need
to define a complete mapping between two very lax, broad schemas,
we will end up with a large number of correspondences that are
irrelevant for any single instance. An interesting research
direction would be to develop a strategy to do summary mapping in
the same spirit of schema mapping, perhaps using EEs definitions
to create the correspondences in XPath. Another option would be to
use \dx\ summaries to determine what schema elements do not apply
to a given collection and then only define correspondences for
those elements that are actually used. This would significantly
reduce the number of correspondences needed to define a meaningful
mapping hence simplifying the overall data translation process.

The notion of bisimulation originated in fields other than
databases (concurrency theory, verification, modal logic, set
theory), where it continues to find applications. It would be
interesting to explore whether the more flexible notion introduced
in this thesis (selective bisimilarity applied to subgraphs
described by AxPREs) can also find novel applications in such
areas.

Since this XPath-to-AxPRE syntactic translation can be applied to
any XPath query, it can also be used to translate \xp\ queries
\cite{CLR07} to AxPREs. \xp\ expressions have the same syntax as
XPath but a different semantics which provides an explanation in
the form of the intermediate nodes, a kind of data provenance of
the answer.

Open research issues also include creating AxPREs for the \xp\
expressions of a query, so that \dx\ can adapt SDs to accelerate
the retrieval of intermediate nodes. In addition, we plan to study
the impact of adjusting the workload (e.g, by finding frequent
patterns), and also how to optimize SD selection given budget
constraints. There are also opportunities for exploiting the
flexibility available in AxPRE summaries in the context of the
more traditional summary applications to indexing, selectivity
estimation, and query optimization.

\addcontentsline{toc}{chapter}{Bibliography}
\bibliographystyle{alpha}
\bibliography{biblio,books}

\appendix

\appendix

\chapter{XPath 1.0 formal semantics} \label{Section:XPathLanguage}

We provide in this appendix a concise definition of the formal
semantics of XPath 1.0 \cite{W3C:XPath/01}. Several semantic
characterizations of XPath 1.0 have been proposed recently
\cite{GKP03a,MdR05,BFK05}. As part of the foundation of \dx, we
have extended the XPath formalization given in \cite{GKP05} to
better capture all the relevant constructs in the standard. A
significant addition to the rules is the proper treatment of the
interaction of parentheses followed by predicates. Parenthesis use
in XPath does not just affect precedence and grouping of
operators, it does in fact change the semantics \cite{CLR07}.

Since XPath was designed to be embedded in other XML languages, it
provides information about the \emph{context} in which an
expression will be evaluated. Given that XPath manipulates node
sets, in addition to the node from which to start the evaluation,
the context has to contain the node's position relative to a node
set and the node set size. This node set could be the result of
the evaluation of another XPath expression or a construct of the
host language.

\begin{definition}[Context] \label{def:contex} Let
axis graph $\mathcal{A} = (\mathit{Inst},$ $Axes,$ $Label,$
$\lambda)$ be an axis graph, $S \subseteq \mathit{Inst}^*$ and $v
\in S$. The \emph{context} of $v$ in $S$ with respect to $axis$ is
defined as $t= \langle v, pos_{axis}(v, S), |S| \rangle$. We say
that $v$ is the \emph{context node}, $pos_{axis}(v, S)$ the
\emph{context position} of $v$ in $S$ w.r.t. $\prec_{axis}$, and
$|S|$ the \emph{context size}. \qed
\end{definition}

Each expression evaluates relative to a context and returns a
value of one of four types: \emph{number}, \emph{node set},
\emph{string} and \emph{boolean}. Other important XPath syntactic
constructs are \emph{location paths}, which are special cases of
expressions. Location paths come in two flavors: \emph{absolute}
and \emph{relative}. An absolute location path consists of $/$
optionally followed by a relative location path. A relative
location path consists of one or more \emph{location steps}
separated by $/$. (Since location steps are expressions, they also
evaluate relative to a context.)

We define next the formal semantics of XPath expression, location
paths and operator with functions $\mathcal{E}$, $\mathcal{L}$ and
$\mathcal{F}$.

\begin{definition}[Semantic Functions $\mathcal{E}$, $\mathcal{L}$ and $\mathcal{F}$] \label{def:semfunc}
Let $Op$ be a place holder for operators $ArithOp$ $\in \{+,$ $-,
*,$ $div, mod \}$, $RelOp$ $\in \{$ $=, \neq,$ $\leq, <,$ $\geq, >
\}$, $EqOp$ $\in \{=, \neq, \}$, and $GtOp$ $\in \{\leq, <, \geq,
> \}$. Let $e, e_1 \ldots e_m$ be expressions and $locpath,$
$locpath_1,$ $\ldots,$ $locpath_m$ location paths. The semantics
of XPath expressions are defined by semantics functions
$\mathcal{E}$ and $\mathcal{L}$ in Figure \ref{sem:expressions}
and \ref{sem:locationPaths}, and the semantics of operators are
defined by $\mathcal{F}$ in Figures \ref{fig:BasicOperators} and
\ref{fig:AdditionalOperators}. Function $\mathcal{E}$ defines the
semantics of expressions on a context, whereas function
$\mathcal{L}$ defines the semantics of locations paths on a node.
\qed
\end{definition}

\begin{figure}
\begin{equation}
\mathcal{E}[\![locpath]\!]( \langle v, k, n \rangle ) := \mathcal{D}[\![locpath]\!] (v)
\label{sem:locPath}
\end{equation}
\begin{equation}
\mathcal{E}[\![position()]\!]( \langle v, k, n \rangle ) := k \label{sem:position}
\end{equation}
\begin{equation}
\mathcal{E}[\![last()]\!]( \langle v, k, n \rangle ) := n \label{sem:last}
\end{equation}
\begin{equation}
\mathcal{E}[\![Op(e_1, \ldots, e_m)]\!]( \langle v, k, n \rangle ) :=
\mathcal{F}[\![Op]\!](\mathcal{E}[\![e_1]\!]( \langle v, k, n \rangle ), \ldots,
\mathcal{E}[\![e_m]\!]( \langle v, k, n \rangle )) \label{sem:op}
\end{equation}
\caption{Semantic definitions of XPath expressions}
\label{sem:expressions}
\end{figure}

\begin{figure*}
\begin{equation}
\mathcal{D}[\![locpath_1 | \ldots | locpath_m]\!](v):= \bigcup^m_{i=1}
\mathcal{L}[\![locpath_i]\!](v) \label{sem:union}
\end{equation}
\begin{equation}
\mathcal{L}[\![(locpath)[e_1]\ldots[e_m] \, ]\!](v) := \{ \, w \,
| \, w \in S \, \wedge \, S = \mathcal{D}[\![locpath]\!](v)
\label{sem:parenthesis}
\end{equation}
\[\bigwedge^m_{i=1}(\mathcal{E}[\![e_i]\!](w, pos_{doc}(w,S), |S|) =
true ) \} \]
\begin{equation}
\mathcal{L}[\![locpath_1 / locpath_2]\!](v) := \bigcup_{w \in
\mathcal{L}[\![locpath_1]\!](v)} \mathcal{L}[\![locpath_2]\!](w) \label{sem:composition}
\end{equation}
\begin{equation}
\mathcal{L}[\![/ locpath ]\!](v) := \mathcal{L}[\![locpath]\!](v_0) \label{sem:root}
\end{equation}
\begin{equation}
\mathcal{L}[\![axis::l[e_1]\ldots[e_m] \, ]\!](v) := \{ \, w \, |
\, w \in S \, \wedge \, S= \{v' \, | \, \langle v, v' \rangle \in
axis \, \wedge \, \lambda(v')=l \} \label{sem:axis}
\end{equation}
\[ \bigwedge^m_{i=1}(\mathcal{E}[\![e_i]\!](w, pos_{axis}(w,S), |S|)
= true ) \} \] \caption{Semantic definitions of XPath location
paths} \label{sem:locationPaths}
\end{figure*}

The distinction between context-based and node-based evaluation
comes from the fact that some functions like $position()$ and
$last()$ need to be evaluated on a context (they return the
context position and the context size respectively). The
evaluation of location paths, on the other hand, requires only the
context node.

Below we illustrate through a series of examples how these
semantic functions are used for evaluating XPath expressions. The
examples cover the following four expressions: find all expRoles,
find the last expRole, find the first expRole of each expRoleList,
and find the first expRole in the entire collection. For these
examples we use XPath abbreviated syntax and the XML axis graph
$\mathcal{A}$ of our running example.

Let us start with an expression with a single step that returns
all expRoles in the collection.

\begin{example}
Let $t_0$ be the context $\langle v_0, 1, 1 \rangle$ and let \[e_1
=descendant::expRole \] The evaluation of $e_1$ on $\mathcal{A}$
and $t_0$ returns all expRoles in the collection. In order to
evaluate $e_1$ on $t_0$ we apply the semantic rules from Figures
\ref{sem:expressions} and \ref{sem:locationPaths}. Since $e$ is an
expression containing a location path, the first rule we apply is
(\ref{sem:locPath}) obtaining
\[ \mathcal{E}[\![descendant::expRole]\!](t_0)
:= \mathcal{L}[\![descendant::expRole]\!] (v_0)\] Rule
(\ref{sem:locPath}) translates the evaluation on the entire
context $t_0 = \langle v_0, 1, 1 \rangle$ to an evaluation on just
the context node $v_0$. Since $e_1$ consists of only one location
step, we finish the evaluation by applying rule (\ref{sem:axis})
with no predicates $[e_1] \ldots [e_m]$ and get
\[\mathcal{L}[\![descendant::expRole]\!](v_0) := \] \[ \{ \, w \, |
\, w \in S \, \wedge \, S= \{v \, | \, \langle v_0, v \rangle \in
descendant \, \wedge \, \lambda(v)=expRole \} \} \] which returns
all $w$'s that are descendant expRoles of $v_0$. \qed
\end{example}

Now, we consider a single step expression with a predicate that
returns the last expRole in the collection.

\begin{example}
Let $t_0$ be the context $\langle v_0, 1,1 \rangle$ and let
\[ e_2=descendant::expRole[position()=last()] \] The evaluation of $e_2$ on $\mathcal{A}$
and $t_0$ returns the last expRole in the collection. As in the
previous example, the application of rule (\ref{sem:locPath})
transforms the evaluation on context $t_0 = \langle v_0, 1, 1
\rangle$ to an evaluation on node $v_0$. Since $e_2$ consists only
of a location step, we apply rule (\ref{sem:axis}) and get
\[ \mathcal{L}[\![descendant::expRole[position()=last()]]\!](v_0) := \] \[ \{ \,
w_1 \, | \, w_1 \in S \, \wedge \, S= \{v \, | \, \langle v_0, v
\rangle \in descendant \, \wedge \, \lambda(v)=expRole \} \,
\wedge \, \] \[t_1 = \langle w_1, pos_{descendant}(w_1,S), |S|
\rangle \, \wedge \, \mathcal{E}[\![position()=last()]\!](t_1) =
true \}
\] which returns all $w_1$'s that are descendant expRoles of
$v_0$ and satisfy the predicate $position()=last()$. In order to
evaluate this predicate on each $w_1$, we go back to function
$\mathcal{E}$ by invoking rule (\ref{sem:op}) of Figure
\ref{sem:expressions} with the ``='' operator, the new context
$t_1$ and $\mathcal{F}[\![\, = \,: num \times num \rightarrow bool
]\!]$ obtaining
\[ \mathcal{E}[\![(position()=last())]\!](t_1)
:= \] \[ \mathcal{F}[\![\, =
\,]\!](\mathcal{E}[\![position()]\!](t_1),
\mathcal{E}[\![last()]\!](t_1)) :=
\mathcal{E}[\![position()]\!](t_1) =
\mathcal{E}[\![last()]\!](t_1)
\] This step of the evaluation checks whether each $w_1$ is in
fact the last in the sequence of descendant expRoles by invoking
$\mathcal{E}[\![position()]\!](t_1)$ and
$\mathcal{E}[\![last()]\!](t_1)$ (rules (\ref{sem:position}) and
(\ref{sem:last}) respectively) and comparing their returned values
for equality. If they are equal for some $w_1$ the evaluation of
the location step returns the $w_1$ or else the empty node set.
This completes the evaluation of $e_2$. \qed
\end{example}

Next, we introduce composition with a more complex expression that
returns the first expRole of each expRoleList in the collection.

\begin{example}
Let $t_0$ be the context $\langle v_0, 1,1 \rangle$ and let \[
e_3= descendant::expRoleList/child::expRole[1] \]  The evaluation
of $e_3$ on $\mathcal{A}$ and $t_0$ returns the first expRole of
each expRoleList in the collection. As in the previous example,
the application of rule (\ref{sem:locPath}) transforms the
evaluation on context $t_0 = \langle v_0, 1, 1 \rangle$ to an
evaluation on node $v_0$. Since $e_3$ consists of a composition of
a location step and a location path, we apply rule
(\ref{sem:composition}) and get
\[ \mathcal{L}[\![descendant::expRoleList/child::expRole[1]]\!](v_0) := \] \[ \bigcup_{w_1 \in
\mathcal{L}[\![descendant::expRoleList]\!](v_0)}
\mathcal{L}[\![child::expRole[1]]\!](w_1)\] which entails
evaluating the location path on the union of all $w_1$'s that are
returned by the evaluation of the location step. In order to
obtain those $w_1$'s we apply rule (\ref{sem:axis}) to the
location step $descendant::expRoleList$ obtaining
\[ \mathcal{L}[\![descendant::expRoleList]\!](v_0) := \] \[ \{ \, w_1 \, | \, w_1 \in S
\, \wedge \, S= \{v \, | \, \langle v_0, v \rangle \in descendant
\, \wedge \, \lambda(v)=expRoleList \} \}
\] and finish with the evaluation of the first part of the
composition. Next we evaluate the location path by applying rule
(\ref{sem:axis}) and get \[
\mathcal{L}[\![child::expRole[1]]\!](w_1) := \{ \, w_2 \, | \, w_2
\in S \, \wedge \, S= \{v \, | \, \langle w_1, v \rangle \in child
\, \wedge \, \lambda(v)=expRole \} \, \wedge \,
\] \[t_2 = \langle w_2, pos_{child}(w_2,S), |S| \rangle \, \wedge \,
\mathcal{E}[\![position()=1]\!](t_2) = true \}
\] which returns all $w_2$'s that are child expRoles of the $w_1$'s and
satisfy the predicate $[1]$ (which is $[position()=1]$ in
unabbreviated syntax). In order to evaluate the predicate we
invoke rule (\ref{sem:op}) of Figure \ref{sem:expressions} with
the ``='' operator, the new context $t_2$ and $\mathcal{F}[\![\, =
\,: num \times num \rightarrow bool ]\!]$ and get
\[ \mathcal{E}[\![(position()=1)]\!](t_2) := \mathcal{F}[\![\, = \,]\!](\mathcal{E}[\![position()]\!](t_2),
1) :=  \mathcal{E}[\![position()]\!](t_2) = 1
\] This step of the evaluation checks whether the $w_2$ is in
fact the first in the sequence of child expRoles of $w_1$ by
invoking $\mathcal{E}[\![position()]\!](t_2)$ (rule
(\ref{sem:position})) and see if it returns $1$. Since the
predicate is part of the location path, it is evaluated on each of
the $w_2$'s. Thus, the evaluation of $e_3$ will return one $w_2$
(the first one in the sequence) for each $w_1$. \qed
\end{example}

Finally, let us illustrate the impact of parentheses by
considering an expression that returns the first figure in the
entire document.

\begin{example}
Let $t_0$ be the context $\langle v_0, 1, 1 \rangle$ and let
\[ e_4=(descendant::expRoleList/child::expRole)[1] \]
The evaluation of $e_4$ on $\mathcal{A}$ and $t_0$ returns the
first expRole in the entire collection. (Notice the difference
with the previous example which returns the first expRole of each
expRoleList.) As before, we begin the evaluation by invoking rule
(\ref{sem:locPath}). Then we apply rule (\ref{sem:parenthesis}) to
the parenthesized expression obtaining
\[\mathcal{L}[\![(descendant::expRoleList/child::expRole)[1] \, ]\!](v_0) := \]
\[ \{ \, w_2 \, | \, w_2 \in S \, \wedge \, S = \mathcal{L}[\![descendant::expRoleList/child::expRole]\!](v_0) \,
\wedge \, \] \[t_2 = \langle w_2, pos_{doc}(w_2,S), |S| \rangle \,
\wedge \, \mathcal{E}[\![position()=1]\!](t_2) = true ) \}\] Now
we have to evaluate the composition by invoking rule
(\ref{sem:composition}) and get
\[ \mathcal{L}[\![descendant::expRoleList/child::expRole]\!](v_0) := \]
\[ \bigcup_{w_1 \in \mathcal{L}[\![descendant::expRoleList]\!](v_0)}
\mathcal{L}[\![child::expRole]\!](w_1)\] which entails evaluating
the location path on the union of all $w_1$'s that are returned by
the evaluation of the location step. Notice that, in contrast with
the previous example, the predicate is applied to the result of
the composition instead of being part of the location path. We
continue by applying rule (\ref{sem:axis}) to location step
$descendant::expRoleList$ in order obtain the $w_1$'s to be used
by the location path obtaining \[
\mathcal{L}[\![descendant::expRoleList]\!](v_0) := \] \[ \{ \, w_1
\, | \, w_1 \in S \, \wedge \,  S= \{v \, | \, \langle v_0, v
\rangle \in descendant \, \wedge \, \lambda(v)=expRoleList \} \}
\] Next we invoke (\ref{sem:axis}) to evaluate the location step $child::expRole$ and get \[
\mathcal{L}[\![child::expRole]\!](w_1) := \{ \, w_2 \, | \, w_2
\in S \, \wedge \, S= \{v \, | \, \langle w_1, v \rangle \in
descendant \, \wedge \]\[ \lambda(v)=expRole \} \, \wedge \,t_1 =
\langle w_2, pos_{child}(w_2,S), |S| \rangle \}
\] which returns all $w_2$ that are child $expRole$ of the $w_1$.
We then evaluate the predicate
$\mathcal{E}[\![(position()=1)]\!](t_2)$ as before. However, one
difference with the previous example is that here there is only
one sequence of $w_2$'s (rather than one sequence for each $w_1$).
That is the reason why the context in the evaluation of
$position()=1$ changed from $t_1$ (based on the $descendant$ axis)
to $t_2$ (based on the entire axis graph). Thus, the evaluation of
$e_4$ returns only one node: the last expRole in the collection.
\qed
\end{example}

\begin{figure}
$\mathcal{F}[\![ArithOp : num \times num \rightarrow num ]\!](n_1, n_2) := n_1 \; ArithOp \; n_2$ \\
$\mathcal{F}[\![constant \ number \ n: \rightarrow num ]\!]():= n$ \\
$\mathcal{F}[\![count: nset \rightarrow num ]\!](S) := |S|$ \\
\\
$\mathcal{F}[\![and: bool \times bool \rightarrow bool ]\!](b_1, b_2) := b_1 \wedge b_2$ \\
$\mathcal{F}[\![or: bool \times bool \rightarrow bool ]\!](b_1, b_2) := b_1 \vee b_2$ \\
$\mathcal{F}[\![not: bool \rightarrow bool ]\!](b) := \neg b$ \\
$\mathcal{F}[\![true: \rightarrow bool ]\!]() := true $\\
$\mathcal{F}[\![false: \rightarrow bool ]\!]() := false $ \\
$\mathcal{F}[\![boolean: nset \rightarrow bool ]\!](S) :=$ if $S \neq \emptyset$ then $true$ else $false$ \\
$\mathcal{F}[\![boolean: str \rightarrow bool ]\!](s) :=$ if $s \neq$ ``'' then $ true $ else $ false$ \\
$\mathcal{F}[\![boolean: num \rightarrow bool ]\!](n) :=$ if $n \neq 0 $ and $n \neq NaN$ then $true$ else $false$ \\
\\
$\mathcal{F}[\![EqOp : bool \times (str \bigcup num \bigcup bool) \rightarrow bool ]\!](b, x) := b \; EqOp \; \mathcal{F}[\![boolean]\!](x)$ \\
$\mathcal{F}[\![EqOp : num \times (str \bigcup num) \rightarrow bool ]\!](n, x) := n \; EqOp \; \mathcal{F}[\![number]\!](x)$ \\
$\mathcal{F}[\![EqOp : str \times str \rightarrow bool]\!](s_1, s_2) := s_1 \; EqOp \; s_2 $ \\
$\mathcal{F}[\![RelOp : nset \times nset \rightarrow bool ]\!](S_1, S_2) := \exists v_1 \in S_1, v_2 \in S_2: strval(v_1) \; RelOp \; strval(v_2) $\\
$\mathcal{F}[\![RelOp : nset \times num \rightarrow bool ]\!](S, n) := \exists v \in S: to\_number(strval(v)) \; RelOp \; n$ \\
$\mathcal{F}[\![RelOp : nset \times str \rightarrow bool ]\!](S, s) := \exists v \in S: strval(v) \; RelOp \; s$ \\
$\mathcal{F}[\![RelOp : nset \times bool \rightarrow bool ]\!](S, b) := \mathcal{F}[\![boolean]\!](S) \; RelOp \; b$ \\
$\mathcal{F}[\![GtOp : (str \bigcup num \bigcup bool) \times (str \bigcup num \bigcup bool) \rightarrow bool ]\!](x_1, x_2) :=$ $\mathcal{F}[\![number]\!](x_1) $ $GtOp$ $\mathcal{F}[\![number]\!](x_2) $\\
\caption{Semantic definitions of XPath basic operators}
\label{fig:BasicOperators}
\end{figure}

\begin{figure}
$\mathcal{F}[\![id: nset \rightarrow nset ]\!](S) := \bigcup_{v \in S}\mathcal{F}[\![id]\!](strval(v))$\\
$\mathcal{F}[\![id: str \rightarrow nset ]\!](s) := deref\_ids(s) $\\
\\
$\mathcal{F}[\![number: str \rightarrow num ]\!](s) := to\_number(s)$ \\
$\mathcal{F}[\![number: bool \rightarrow num ]\!](b) :=$ if $b = true $ then $1$ else $0$ \\
$\mathcal{F}[\![number: nset \rightarrow num ]\!](S) := \mathcal{F}[\![number]\!](\mathcal{F}[\![string]\!](S)) $\\
$\mathcal{F}[\![sum: nset \rightarrow num ]\!](S) := \sum_{v \in S} to\_number(strval(v))$\\
\\
$\mathcal{F}[\![constant \ string \ s: \rightarrow str ]\!]():=s$ \\
$\mathcal{F}[\![string: num \rightarrow str ]\!](n) := to\_string(n)$ \\
$\mathcal{F}[\![string: nset \rightarrow str ]\!](S) := $ if $S = \emptyset$ then ``'' else $strval(first_{<_{doc}}(S))$ \\
$\mathcal{F}[\![string: bool \rightarrow str ]\!](b) := $ if $b = true$ then ``true'' else ``false'' \\
\caption{Semantic definitions of XPath additional operators}
\label{fig:AdditionalOperators}
\end{figure}

\chapter{Declarative debugging of XPath queries with \dx}
\label{sec:debugExample}

\begin{figure}
    \centering
        \includegraphics[width=.90\textwidth]{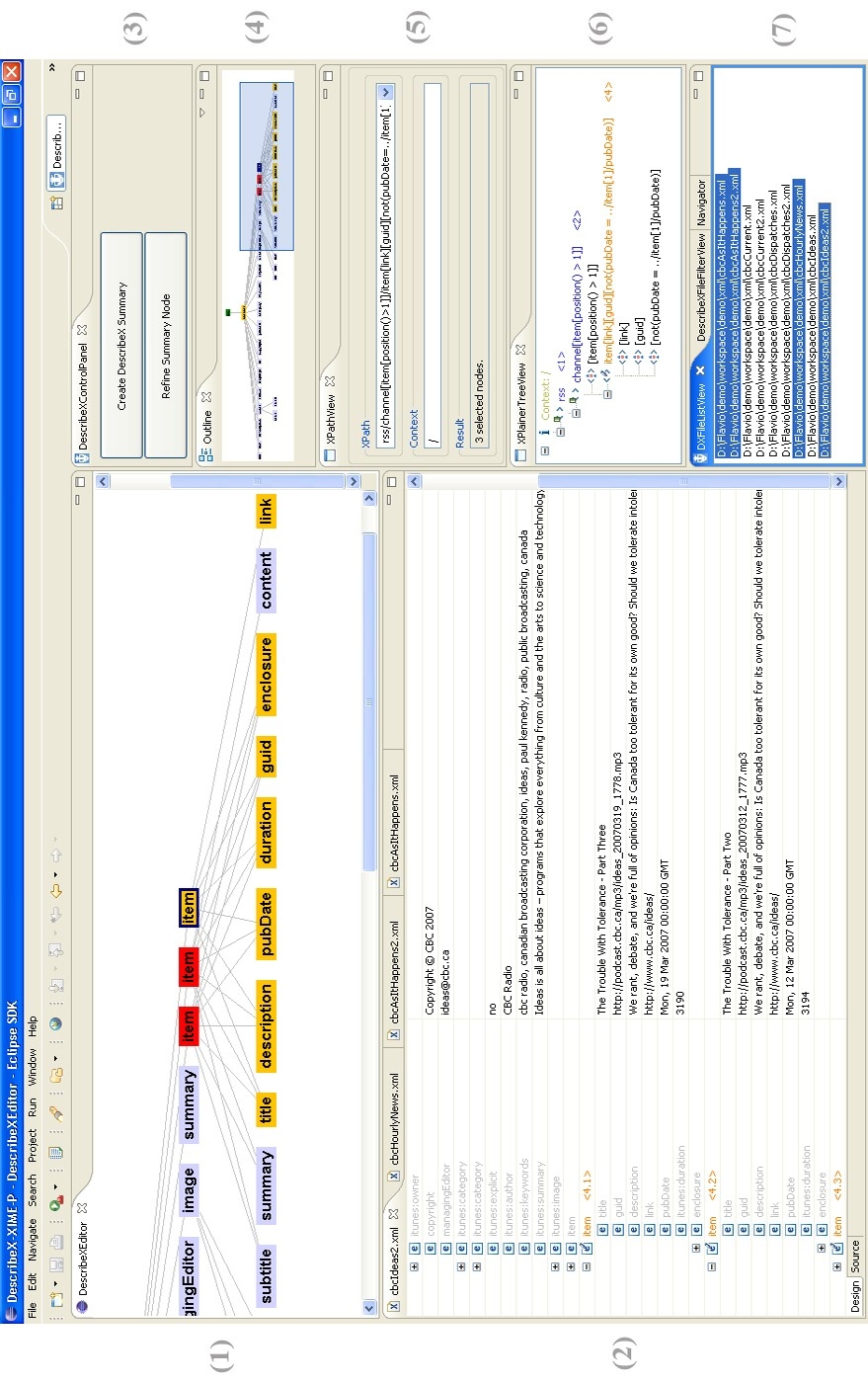}
  \caption{\dxp\ user interface}
    \label{fig:screen}
\end{figure}

In this appendix, we present \dxp, a visual interactive tool for
exploring XML collections and explaining XPath queries. \dxp\ is
built on top of the \dx\ engine implementation presented in
Chapter~\ref{Section:Implementation} and includes a GUI and
additional XML retrieval tools implemented by other colleagues
\cite{ACSR08}. This visual tool provides additional evidence of
the wide-range of applications the \dx\ framework has in the area
of XML processing.

\dxp\ is written in Java for the Eclipse\footnote{{\small
\tt{http://www.eclipse.org/}}} plug-in framework and its existing
tools, views, and editors. Eclipse is a popular open source
platform built by an open community of tool providers.  \dxp\ is
also integrated with the \xpe\ plug-in \cite{CLR07}, and fully
supports \emph{declarative debugging} of any arbitrary XPath
engine, including implementation dependant intermediate results.
\xpe\ extends the XML and XPath development facilities available
in the Eclipse environment with the ability to support explanation
queries. The \dxp\ tool extends \xpe\ with the structural
description capabilities of the \dx\ framework.

\dxp\ allows developers to navigate between different views of the
local and global structure of large (multi-gigabyte size)
collections in order to obtain enough structural information for
writing and debugging XPath queries. The graph based visualization
employed by \dxp\ makes it straightforward to see the different
path structures that are present in the collection. \dxp\
functionality helps a user in quickly understanding what parts of
a collection schema (if present) are used in practice.

In order to explain how \dxp\ works, let us go back to our
developer, Sue, who is trying to aggregate podcasts that are part
of a series. In a series feed, items are sorted from the newest
(the first) to the oldest (the last). The feed may span several
days or weeks, and there might be more than one item per day. In
particular, Sue is interested now in items containing pubDate,
link and enclosure elements. In addition, she aggregates the
item(s) of the latest day in the series separately from the rest.
For obtaining all items that do not belong to the latest day she
runs the query
\begin{center}
{\small \texttt{Q4 =
/rss/channel[item[position()>1]]/item[link][enclosure]}}

{\small \texttt{ [not(pubDate=../item[1]/pubDate)] } }
\end{center}
which returns all items containing link, enclosure and pubDate
from previous days.

Using \dxp\ Sue can create an SD like the one shown in the
screenshot of Figure \ref{fig:screen}. The screenshot shows seven
views, and the SD graph is displayed in the \dx Editor (view (1))
and outline (view (4)). The outline view shows the entire SD graph
with the fragment that appears in the \dx Editor highlighted in a
light blue box.

The SD of Figure \ref{fig:screen} has a node for each item that
has a different substructure. The edges represent $c$ axis
relations between elements. The fragment of the SD graph displayed
in view (1) tells the developer there are three kinds of items in
the collection, each one containing a different combination of
elements. For instance, the third item node in the SD (in yellow)
has title, description, pubDate, duration, guid, enclosure and
link (all in yellow) and represents all item elements in the
collection with that particular structure.

Since the behavior of a query is instance dependant, in order to
debug the query effectively Sue needs to run it on all
\emph{candidate documents}. A document is a candidate for a query
$Q$ when it returns a non-empty answer for the structural subquery
of $Q$.

A visual explanation, which shows the XPath result and
intermediate nodes, is provided by views (2) and (6). Given an
XPath query and an input XML document, an explanation of the query
gives as answer all the XPath result nodes together with
intermediate nodes. The intermediate nodes are those nodes
resulting from the partial evaluation of the subexpressions of the
original XPath query that contribute to the answer. Obtaining the
explanation of a complex XPath query can be challenging, as shown
in the example. Visual explanations provide a representation of
the basic mechanism at play during XPath processing.

Views (2) and (6), together with view (5), correspond to the \xp
-Eclipse plug-in \cite{CLR07}. The \xp\ Editor (view (2)) shows
one of the candidate documents with an explanation of Q4. View (5)
displays the query and the number of elements in the answer. View
(6) displays the \xp\ tree, a particular parse tree for the query
that provides an intuitive representation of its structure. Each
node in the \xp\ tree corresponds to one step or predicate in Q4.
The intermediate document nodes of each step are identified by the
same sequence number in both \xp\ tree and the \xp\ Editor. For
instance, item nodes $\langle 4.1 \rangle$, $\langle 4.2 \rangle$
and $\langle 4.3 \rangle$ (view (2)) are the answer of the query,
which corresponds to step $\langle 4 \rangle$ in the \xp\ tree
(view (6)).

Since current XPath query evaluation tools do not provide
intermediate nodes, the only available debugging techniques
involve either partial evaluation of subexpressions or evaluating
reversed axis. A partial evaluation cannot see beyond the current
evaluation step, so it has no way of filtering out nodes that will
have no effect in the final answer. For instance, a partial
evaluation of the $/rss/channel$ subexpression would return all
channels below rss elements, including those that do not satisfy
the $[item[position()>1]]$ predicate and the rest of the query. An
evaluation that reverses the axis will not necessarily give us
exactly the intermediate nodes either when recursive axes like
$descendant$ or $following$ are involved. Thus, visual
explanations are necessary in order to obtain the exact set of
intermediate nodes that contribute to the answer. An in-depth
study of visual explanations can be found in \cite{CLR07}.

The DXFileListView (view (7)) lists the documents in the extent of
the active node (the yellow item in the \dx Editor) that are also
candidates for the query shown in View (5). The documents
highlighted in the DXFileListView are the \emph{explanation
documents}, i.e. those candidates that satisfy the complete query.
The notions of candidate an explanation documents are key to the
integration of \xp\ into the \dx\ framework (see Chapter
\ref{sec:candidates}). The developer can then open any candidate
or explanation document in the DXFileListView with the \xp\ Editor
and obtain explanations of either Q4 or different
\emph{relaxations} of Q4.

A relaxation of a query is obtained by selectively collapsing
portions of the XPath expression to eliminate constraints. This is
useful when there are no answers to a complete query, but then
after removing constraints the relaxed query can be satisfied. A
very useful relaxation is the one that removes all non-structural
predicates from a query $Q$ thus obtaining its structural
subquery.

A very important property of \dxp\ is that it is not tied to any
particular XPath implementation. Instead, an arbitrary XPath
evaluator can be invoked through a standard interface. This is a
critical engineering decision that allows \dxp\ to provide
explanations for different XPath engines. Beyond differences in
the capabilities of the implementations, the XPath language itself
has several areas where the semantics are implementation defined.
This effectively means that only the original XPath engine can
explain one of its own implementation defined features.

\vspace{.2in}

With the \dxp\ tool, debugging and exploration complement each
other: Sue can decide interactively to get different descriptions
of the collection by changing the SD definition or obtain more or
less strict visual explanations by relaxing a query in different
ways. Thus, \dxp\ provides Sue with a flexible, integrated
environment for understanding a collection and the queries she
needs to run on it.


\end{document}